\documentclass[12pt,onecolumn,draftcls]{IEEEtran}

\usepackage[capitalize]{cleveref}

\usepackage{subfig}
\usepackage{graphicx}
\usepackage{amsthm,amsfonts,amssymb,amsmath}
\usepackage{algorithm}
\usepackage{algpseudocode}
\usepackage{color}
\usepackage{bm,bbm}

\newtheorem{prop}{Proposition}

\newtheorem{theorem}{Theorem}
\newtheorem{remark}{Remark}
\newtheorem{lemma}{Lemma}
\newtheorem{definition}{Definition}
\newtheorem{example}{Example}
\newtheorem{corollary}{Corollary}

\newcommand{\OO}{{\mathcal O}}

\begin{document}
\title{Non-Commutative Ring Learning With Errors From Cyclic Algebras}

\author{Charles~Grover$^1$, Cong~Ling$^1$, Roope~Vehkalahti$^2$}


\author{Charles~Grover, Cong~Ling and Roope~Vehkalahti\thanks{

C. Grover is with the Department of Electrical and Electronic Engineering, Imperial College London, London SW7 2AZ, UK (e-mail: c.grover15@imperial.ac.uk).

C.  Ling  is  with  the  Department  of  Electrical  and  Electronic  Engineering, Imperial College London, London SW7 2AZ, UK (e-mail: cling@ieee.org).

R. Vehkalahti is with the Department of Communications and Networking, Aalto University, Espoo, FI-02150, Finland (e-mail: roope.vehkalahti@aalto.fi).

}}

\maketitle
\begin{abstract}
The Learning with Errors (LWE) problem is the fundamental backbone of modern lattice based cryptography, allowing one to establish cryptography on the hardness of well-studied computational problems. However, schemes based on LWE are often impractical, so Ring LWE was introduced as a form of `structured' LWE, trading off a hard to quantify loss of security for an increase in efficiency by working over a well chosen ring. Another popular variant, Module LWE, generalizes this exchange by implementing a module structure over a ring. In this work, we introduce a novel variant of LWE over cyclic algebras (CLWE) to replicate the addition of the ring structure taking LWE to Ring LWE by adding cyclic structure to Module LWE. The proposed construction is both more efficient than Module LWE and conjecturally more secure than Ring LWE, the best of both worlds. We show that the security reductions expected for an LWE problem hold, namely a reduction from certain structured lattice problems to the hardness of the decision variant of the CLWE problem. As a contribution of theoretic interest, we view CLWE as the first variant of Ring LWE which supports non-commutative multiplication operations. This ring structure compares favorably with Module LWE, and naturally allows a larger message space for error correction coding.

\end{abstract}
\section{Introduction}
With the predicted advent of quantum computers compromising the bulk of existent cryptographic constructions, lattice based cryptography has emerged as a promising foundation for long term security. In particular, the Learning with Errors (henceforth LWE) problem introduced in \cite{regev_lattices_2009}, as well as its variants over rings (RLWE) \cite{lyubashevsky_ideal_2010} and modules (MLWE) \cite{langlois_worst-case_2015}, provides a natural intermediate step to base cryptographic hardness on lattice short vector problems in a post quantum setting. Indeed, second round submissions to the NIST post quantum standardisation process such as NewHope \cite{alkim_post-quantum_2016-2} and KYBER \cite{avanzi_kyber_2019} rely on the hardness of LWE variants. Cryptography based on the classical LWE problem is typically somewhat impractical, in part due to large key sizes. To solve this, the ring variant was introduced as a way to provide extra structure in LWE to trade a potential loss of security for an increase in efficiency. MLWE generalizes ring  and classical LWE, providing a smoother transition between security and efficiency than the binary option presented by ring or classical LWE. The flexibility of MLWE is highly desirable in practice, as demonstrated by third-round NIST finalists KYBER and SABER, both based on MLWE \cite{NIST-Round3}.

Conceptually, one may view all these problems as variations on a single problem. The (search) LWE problem tasks a solver with recovering a secret vector $\textbf{s} \in \mathbb{Z}_q^n$ from a collection of pairs $(\textbf{a}_i, b = \langle \textbf{a}_i,\textbf{s} \rangle + e_i)$, where $\langle \cdot,\cdot \rangle$ denotes the inner product, each $\textbf{a}_i \in \mathbb{Z}_q^n$ is uniformly random and the $e_i$'s are small random errors. In practice, we view this collection of equations in matrix-vector form:
\begin{align*}
A \textbf{s} + \textbf{e} = \textbf{b},
\end{align*}
where all operations and entries are over $\mathbb{Z}_q$ and the challenge is to recover $\textbf{s}$ from $A, \textbf{b}$. A popular ring variant replaces $A, \textbf{s}, \textbf{e}$ with elements $a,s,e$ from the ring $R_q := \frac{\mathbb{Z}_q[x]}{x^n+1}$, requiring the solver to obtain $s$ from samples $a_i \cdot s + e_i$. For power-of-two $n$ this can be expressed in matrix-vector form by considering the matrix rot$(a)$, the negacyclic matrix obtained from the coefficients of $a$. Explicitly, for $a = a_0 +a_1 x +... +a_{n-1} x^{n-1}$ and bold faced letters denoting coefficient vectors, a sample from the RLWE distribution takes the form:
\begin{align*}
\begin{pmatrix}
a_0 & -a_{n-1} & \dots & -a_{1} \\
a_1 & a_{0} & \dots & -a_2 \\
\vdots & \vdots & \ddots & \vdots \\
a_{n-1} & a_{n-2} & \dots & a_0
\end{pmatrix} \textbf{s} + \textbf{e} = \textbf{b}
\end{align*}
where once again operations and entries are over $\mathbb{Z}_q$. This is exactly a structured version of the classical LWE problem, where the uniformly random matrix $\textbf{A}$ has been replaced by the negacyclic matrix rot$(a)$. Of course, this should be an easier problem to solve, yet no substantial progress has been made in using the structure of rot$(a)$ to solve the problem efficiently. We can extend this matrix-vector view to MLWE as well. An MLWE instance takes place in a module $M$ of dimension $d$ over $R_q$, such that a solver has to recover $\textbf{s} \in M$ from a collection of pairs $(\textbf{a}_i, \langle \textbf{a}_i, \textbf{s} \rangle + e_i)$ where $\textbf{a}_i$ is a uniformly random element of $M$ and each $e_i$ is a small random element of $R_q$. A collection of such pairs can be viewed as $A \textbf{s} + \textbf{e} = \textbf{b}$, where the ambient space $\mathbb{Z}_q$ has been replaced by $R_q$ e.g. with $d$ samples:
\begin{align*}
\begin{pmatrix}
a_{1,1} & a_{1,2} & \dots & a_{1,d} \\
a_{2,1} & a_{2,2} & \dots & a_{2,d} \\
\vdots & \vdots & \ddots & \vdots \\
a_{d,1} & a_{d,2} & \dots & a_{d,d}
\end{pmatrix} \textbf{s} + \textbf{e} = \textbf{b}
\end{align*}
where all operations are over $R_q$ and each $a_{i,j}$ is uniformly random. Of course, we could extend this to have operations over $\mathbb{Z}_q$ by applying the rot$(\cdot)$ operation coordinatewise, to obtain a structured LWE instance in dimension $nd$.

An advantage of these structured matrices is that they allow for streamlined storage and operations. For example, storing a uniformly random matrix $A$ requires one to store all $n^2$ of its entries, but rot$(a)$ requires a factor $n$ less memory since one need only store its first column. Equivalently, one RLWE sample generates $n$ LWE samples while reducing the storage space and key sizes. Multiplication can also be speeded up by using the Chinese Remaindering Theorem (CRT) or other techniques.

This concept of improving efficiency by adding structure motivates this work; can we perform an analog of the transformation taking an LWE matrix $A$ to an RLWE matrix rot$(a)$ for the module $M$? We solve this by constructing a new variant of the LWE problem over a certain non-commutative space known as a \emph{cyclic algebra}. In recent years, cyclic algebras have received significant attention in the field of coding theory (see e.g. \cite{vehkalahti_densest_2009, oggier_cyclic_2007, luzzi_almost_2018}) due to the particular nature of the matrix lattices they induce, and we view them as a suitable option for defining an LWE problem over a non-commutative ring. Though some efforts have been made to construct non-commutative LWE problems, for example \cite{baumslag_generalized_2011}, \cite{cheng_lwe_2016}, the majority of non-commutative cryptography has relied on group theoretic constructions, whose underlying hard problems are often less robust than those of lattice cryptography. Somewhat informally, for a cyclic algebra $\mathcal{A}$ and well chosen parameters there exists an automorphism $\theta$ of $R_q$ and a $\gamma \in R_q$ such that an LWE style sample $a \cdot s + e$ over $\mathcal{A}$ can be written in matrix-vector form
\begin{align*}
\begin{pmatrix}
a_0 & \gamma \theta(a_{d-1}) & \gamma \theta^2(a_{d-2}) & \ldots &\gamma \theta^{d-1}(a_{1}) \\
a_1 & \theta(a_{0}) & \gamma \theta^2(a_{d-1}) & \ldots &\gamma \theta^{d-1}(a_{2}) \\
a_2 & \theta(a_{1}) & \theta^2(a_{0}) & \ldots &\gamma \theta^{d-1}(a_{3}) \\
\vdots & \vdots & \vdots &\ddots & \vdots \\
a_{d-1} & \theta(a_{d-2}) & \theta^2(a_{d-3}) & \ldots & \theta^{d-1}(a_{0}) \\
\end{pmatrix} \textbf{s} + \textbf{e} = \textbf{b}
\end{align*}
where all entries and operations are now over $R_q$. Though more complex than the transformation taking LWE to RLWE this fulfills our goal of providing a structured version of MLWE, since we have replaced the uniformly random matrix $A$ over $R_q$ with a structured matrix which we denote $\phi(a)$ that requires a factor of $d$ less storage. Of course, by applying the rot$(\cdot)$ operation coordinatewise, one can extend this to a high dimensional version of the LWE problem, now with two sets of structure lying on top of each other.
\subsection{Contributions and Methodology}
The main novel contribution of this work is a definition of Cyclic Algebra LWE (CLWE), together with justifications for its construction and a polynomial time reduction from short vector problems over matrix lattices induced by ideals in a cyclic algebra to CLWE, establishing its security on the assumption that such problems are hard.
As in \cite{lyubashevsky_ideal_2010}, the algorithm bases the security of CLWE on short vector problems over ideal lattices in $\mathcal{A}$; similarly to ideal lattices in $K$, these have some extra underlying structure that might make computational problems easier. However, we leave the relative complexity of these problems an open area of investigation.

Overall we consider it plausible that LWE in cyclic algebras could be both more efficient than MLWE and more secure than RLWE in a quantum setting. CLWE represents a middle ground between RLWE and MLWE, with the salient feature of its non-commutative ring structure. Cyclic algebra is equipped with a proper ring multiplication which preserves the dimension of the lattice. This is in sharp contrast to MLWE which only supports scalar multiplication and to RLWE whose multiplication is commutative. Specifically, we consider the following advantages of our CLWE construction:

\begin{itemize}
\item Efficiency. CLWE can be seen a structured variant of MLWE. Assuming for simplicity that the public key in LWE based schemes is a sample $(A,\textbf{b})$, a public key generated as $A =$ rot($\phi(a)$) requires only as much storage as that of an equivalent dimension RLWE public key\footnote{In practice, a seed is often used to generate the matrix $A$, which however requires a pseudorandom generator under the random oracle model. By contrast, CLWE does not require the random oracle model. Moreover, certain applications do not permit the use of a seed, \textit{e.g.}, pseudorandom functions \cite{Banerjee}.}. Multiplication in cyclic algebras can be implemented over a product of skew polynomial rings following a CRT-style decomposition (see Appendix \ref{appendix:multiplication-complexity}), for which well known fast algorithms, such as those of \cite{caruso_fast_2017-4} and \cite{puchinger_fast_2018}, can applied to compute the operation $A \cdot \textbf{s}$ more efficiently in the case where $A = \phi(a)$ than in the module case where $A$ is uniform.
\item Security. Following recent works on quantum attacks on related ideal lattice problems (e.g. \cite{biasse_quantum_2015},\cite{cramer_recovering_2016}, \cite{cramer_short_2017}, \cite{campbell_soliloquy:_2015} amongst others), we observe that the non-commutativity of multiplication in cyclic algebras may be viewed as a security advantage. This is because the Hidden Subgroup Problem (HSP), an integral part of the majority of algorithms using quantum computing to gain an advantage over classical computation, requires that the underlying group, in this case the unit group of $\mathcal{O}_K$, is commutative, see e.g. \cite{jozsa_quantum_2001}, which is untrue for a non-commutative algebra. We conjecture that the security level is higher than RLWE, but welcome further cryptanalysis. We actively avoid known attacks on previous attempts to create structured MLWE (see \cref{BCWappendix}).
\item Decryption failure rates. The scalar multiplication of MLWE is dimension-lossy. In other words, the message space of MLWE is restricted in $R_q$, whose dimension is smaller than that of the module lattice. It leaves less room for error correction coding in MLWE-based schemes (e.g., a KYBER instance for a key size of $256$ within $R_q$ of dimension $256$). This limitation of MLWE appears to be fundamental, due to its module structure. In contrast, the dimension of the message space of CLWE is that of the (non-commutative) ring, which is higher by a factor of $d$. Thus, it accommodates better error correction coding (see \cref{clwecrypto}), and low decryption failure rates are desired under chosen ciphertext attacks (CCA). Even trivial repetition coding can dramatically reduce decryption failure rates (e.g., NewHope).
\item Functionality. We view the ring structure of CLWE as a major advantage over MLWE, which opens up the prospect of extra functionality. For example, since operations are composable and non-commutative, one could hope to construct FHE in this non-commutative ring. We leave this frontier open for separate work.
\end{itemize}

\subsection{Related Work and Organization}\label{related}

This work is related to a number of different areas: lattice-based cryptography, information theory and number theory.

In lattice-based cryptography, an alternative construction for structured module LWE, called multivariate-RLWE, was presented in \cite{pedrouzo-ulloa_ring_2016,Revisit-MRLWE}, where they tensor product two (or more) number fields in order to provide a structured module matrix. However, an efficient implementation of \cite{pedrouzo-ulloa_ring_2016} was attacked in \cite{bootland_security_2018}, together with a warning about taking care when putting structure on a module. In short, \cite{bootland_security_2018} attacks certain instances of multivariate-RLWE by providing a homomorphism  to some underlying subfield $K$, dramatically reducing the dimension of the lattice problem to be attacked. Fortunately for this work, a somewhat technical condition on the choice of $\gamma$ known as the \emph{non-norm condition} precludes such a homomorphism existing to reduce the dimension of CLWE (see \cref{BCWappendix}). It is worth pointing out that that their problem has been addressed in \cite{Revisit-MRLWE}, and in fact this fix looks somewhat like our non-norm condition (\emph{e.g.}, unlike the original version, full rank is maintained in \cite{Revisit-MRLWE}).

This paper is inspired by the abundant literature of space-time coding based on cyclic division algebras (see the monographs \cite{oggier_cyclic_2007,Berhuy_2013} and references therein). On a high level, our construction is reminiscent of multi-block space-time codes \cite{Lu_2008,lahtonen_construction_2008}, rather than single-block codes \cite{Oggier_2006,Elia_2007}, with the caveat of scaling up the number of blocks to make the codes practically undecodable. In the context of space-time coding, our construction generalizes \cite{lahtonen_construction_2008} and offers greater flexibility in the code parameters (the number of blocks vs. the number of antennas). Multi-block space-time codes have been used in \cite{luzzi_almost_2018} to achieve information-theoretic security over wiretap channels, as opposed to computational security in a classic cryptographic setting of this paper. Maximal orders were shown in \cite{Hollanti_2008,vehkalahti_densest_2009} as advantageous to the so-called natural orders; both types of orders play a crucial role in this paper. There is a major difference between the roles of cyclic algebras in coding and cryptography, though: the primary concern for coding is the non-vanishing determinant (NVD), while the non-commutative ring structure becomes crucial for cryptography. For efficient multiplication of elements in a cyclic algebra, we heavily rely on the CRT technique of \cite{oggier_quotients_2012-4}; a similar technique has been used in lattice index codes \cite{Huang_DA_CRT_2019,Huang_CRT_2017}.

We present two approaches (subfields and compositum fields) to the construction of novel cyclic division algebras, which enlarge the pool of algebras and may find other applications. Specifically, our proof that the natural order of the family of cyclic division algebras constructed in \cref{goodalgebras} (including those in \cite{lahtonen_construction_2008}) is in fact maximal, is an original contribution.

The rest of this paper is organized as follows. In \cref{sec:preliminaries} we provide necessary background material on lattices, number fields, and cyclic algebras. In \cref{sec3} we provide a definition and discussion of CLWE, together with novel constructions of cyclic division algebras for the CLWE problem. In \cref{security_proof} we provide a reduction from structured lattice problems to search CLWE, as well as a search-worst case decision reduction for CLWE. In \cref{crypto} we show a sample CLWE cryptosystem and provide an estimate of its asymptotic operation complexity. Finally, the paper is concluded in \cref{conclusions} with a discussion of open problems. For a smooth flow of the main text, certain proofs, sideline discussions and technical details are deferred to appendices.

\section{Preliminaries}\label{sec:preliminaries}
\subsection{Lattices}
A lattice is a discrete additive subgroup of a vector space $V$. If $V$ has dimension $n$ a lattice $\mathcal{L}$ can be viewed as the set of all integer linear combinations of a set of linearly independent vectors $B = \lbrace \textbf{b}_1,...,\textbf{b}_k \rbrace$ for some $k \leq n$,  written $\mathcal{L} = \mathcal{L}(B) = \lbrace \sum_{i=1}^k z_i \textbf{b}_i: z_i \in \mathbb{Z} \rbrace$. If $k=n$ we call the lattice full-rank, and we will only consider lattices of full-rank. We can extend this notion of lattices to matrix spaces by stacking the columns of a matrix. We recall two standard lattice definitions.
\begin{definition}
Given a lattice $\mathcal{L}$ in a space $V$ endowed with a metric $\Vert \cdot \Vert$, the minimum distance of $\mathcal{L}$ is defined as $\lambda_1 (\mathcal{L}) = \min_{\textbf{v} \in \Lambda/\lbrace 0 \rbrace} \Vert \textbf{v} \Vert$. Similarly, $\lambda_n(\mathcal{L})$ is the minimum length of a set of $n$ linearly independent vectors, where the length of a set of vectors $ \lbrace \textbf{x}_1,..., \textbf{x}_n \rbrace$ is defined as $\max_i(\Vert \textbf{x}_i \Vert)$.
\end{definition}
\begin{definition}
Given a lattice $\mathcal{L} \subset V$, where $V$ is endowed with an inner product $\langle \cdot, \cdot \rangle$, the dual lattice $\mathcal{L}^*$ is defined $\mathcal{L}^* = \lbrace \textbf{v} \in V : \langle \mathcal{L} , \textbf{v} \rangle \subset \mathbb{Z} \rbrace$.
\end{definition}

\subsection{Gaussian Distributions}
\begin{definition}
For a vector space $V$ with norm $\Vert \cdot \Vert$ and an $r >0$, we define the Gaussian function $\rho_r: V \rightarrow (0,1]$ by $\rho_r(\textbf{x}) = \exp(- \pi \Vert \textbf{x} \Vert/r^2)$.
\end{definition}

We can use this function to define the spherical Gaussian distribution $D_r$ over $V$, which outputs $\textbf{v}$ with probability proportional to $\rho_r(\textbf{v})$. Similarly, we can sample an elliptical Gaussian $D_\textbf{r}$ in a basis $\textbf{b}_1,...,\textbf{b}_n$ of $V$, for $\textbf{r} = (r_1,...,r_n)$ a vector of positive reals, by sampling $x_1,...,x_n$ independently from the one dimensional Gaussian distributions $D_{r_i}$ and outputting $\sum_{i =1}^n x_i \textbf{b}_i$.

When sampling a Gaussian over a lattice $\mathcal{L}$ we will use the discrete form of the Gaussian distribution. We define the distribution $D_{\Lambda,r}$ over $\Lambda$ by outputting $\textbf{x}$ with probability $\frac{\rho_r(\textbf{x})}{\rho_r(\mathcal{L})}$ for each $\textbf{x} \in \mathcal{L}$. This version of the discrete Gaussian is centered at $0$, which in general need not be the case.

An important lattice quantity, known as the smoothing parameter, was introduced in \cite{micciancio_worst-case_2007}. The motivation for the name is provided by \cref{1} following the definition.
\begin{definition}
For a lattice $\mathcal{L}$ and $\varepsilon > 0$, the smoothing parameter $\eta_\varepsilon(\mathcal{L})$ is defined as the smallest $r > 0$ satisfying $\rho_{1/r}(\mathcal{L}^* / \lbrace \textbf{0} \rbrace ) \leq \varepsilon$.
\end{definition}
The following is a special case of \cite{micciancio_worst-case_2007}, Lemma 4.1.

\begin{lemma}\label{1}
For a lattice $\mathcal{L}$ over $\mathbb{R}^n$, $\varepsilon > 0, r \geq \eta_\varepsilon(\mathcal{L})$, and $\textbf{x} \in \mathbb{R}^n$, the statistical distance between $(D_r + \textbf{x}) \mod \mathcal{L}$ and the uniform distribution modulo $\mathcal{L}$ is bounded above by $\varepsilon/2$. Equivalently, $\rho_r(\mathcal{L} + \textbf{x}) \in [\frac{1- \varepsilon}{1+ \varepsilon},1] \cdot \rho_r(\mathcal{L})$.
\end{lemma}
We introduce well known lemmas used to relate the smoothing parameter to standard lattice properties. The first comes from \cite{banaszczyk_new_1993}, the second from \cite{peikert_pseudorandomness_2017-2}.
\begin{lemma}
For a lattice $\mathcal{L}$ of dimension $n$ and $c \geq 1$ it holds that $c \sqrt{n}/ \lambda_1(\mathcal{L}^*) \geq \eta_\varepsilon(\mathcal{L})$ for $\varepsilon = \exp(-c^2n)$.
\end{lemma}
\begin{lemma}\label{smoothing1}
For a lattice $\mathcal{L}$ and $\varepsilon \in (0,1)$ it holds that $\eta_\varepsilon(\mathcal{L}) \geq \frac{\sqrt{\log(1/\varepsilon)/\pi}}{\lambda_1(\mathcal{L}^*)}$.
\end{lemma}

\subsection{Algebraic Number Theory}
\begin{definition}
A number field $K$ is a finite degree extension of the rationals $\mathbb{Q}$. Typically, we define a number field by adjoining some algebraic element $\alpha \in \mathbb{C}$ and set $K = \mathbb{Q}(\alpha)$. The degree of $K$ refers to its degree as a field extension.
\end{definition}
To define a cyclic algebra, we will need to take an additional extension of $K$. In particular, we will need the extension to be Galois over $K$, defined as follows.
\begin{definition}
Let $L/K$ be an extension of number fields of dimension $d$. The Galois group of $L$ over $K$ is the group Aut$(L/K)$ of automorphisms of $L$ that fix $K$. We say that the extension is Galois if the subfield of $L$ fixed by Aut$(L/K)$ is exactly $K$.
\end{definition}
We define a cyclic Galois extension $L/K$ to be a Galois extension such that the Galois group of $L$ over $K$ is the cyclic group generated by some element $\theta$ of degree $d := [L:K]$. Finally, we require the ring of integers of a number field.
\begin{definition}
Given a number field $K$, its ring of integers $\mathcal{O}_K$ is the ring consisting of those elements of $K$ whose minimal polynomial over $\mathbb{Q}$ lie in $\mathbb{Z}[x]$.
\end{definition}
It is easy to check that if $L/K$ is an extension of number fields then $\mathcal{O}_L \cap K = \mathcal{O}_K$.

\subsubsection{The Canonical Embedding}
Let $K = \mathbb{Q}(\alpha)$ be a number field of degree $n$. It is a well known fact that there are exactly $n$ distinct ring embeddings $\sigma_i : K \rightarrow \mathbb{C}$. These embeddings correspond to the $n$ distinct injective ring homomorphisms mapping $\alpha$ to the roots of its minimum polynomial $f$. We split these embeddings and say that there are $r_1$ real embeddings (whose image lie in $\mathbb{R}$) and $r_2$ conjugate pairs of complex embeddings (the complex embeddings come in pairs since complex roots of $f$ occur in conjugate pairs), such that $r_1 +  2 r_2 = n$. The standard convention is to order the embeddings such that the $r_1$ real embeddings come first and the complex embeddings are arranged such that $\sigma_{r_1+ j} = \overline{\sigma_{r_1 + r_2 + j}}$ for $1 \leq j \leq r_2$.
\begin{definition}
Let $K = \mathbb{Q}(\alpha)$ be a number field of degree $n = r_1 + 2r_2$. The canonical embedding $\sigma$ is the ring homomorphism $\sigma : K \rightarrow \mathbb{R}^{r_1} \times \mathbb{C}^{2 r_2}$ defined by
\begin{align*}
\sigma(x) = (\sigma_1(x),...,\sigma_n(x)).
\end{align*}
Formally, $\sigma$ maps into the space
\begin{align*}
H = \lbrace (x_1,...,x_n) \in \mathbb{R}^{r_1} \times \mathbb{C}^{2 r_2} \, | \, x_{r_1 + r_2 + j} = \overline{x_{r_1 + j}} \, \, \forall 1 \leq j \leq r_2 \rbrace \subset \mathbb{C}^n,
\end{align*}
which is isomorphic to $\mathbb{R}^n$ as an inner product space.
\end{definition}
We can equip $H$ with the orthonormal basis $\lbrace \textbf{h}_i \rbrace$, where $\textbf{h}_i = \textbf{e}_i$ for $1 \leq i \leq r_1$ and $\textbf{h}_j = \frac{1}{\sqrt{2}}(\textbf{e}_j + \textbf{e}_{j + r_2}), \textbf{h}_{j+r_2} = \frac{\sqrt{-1}}{\sqrt{2}}(\textbf{e}_j - \textbf{e}_{j + r_2})$ for $r_1 < j \leq r_1+r_2$, and use the well defined $\ell_p$ norm induced by viewing $H$ as a subset of $\mathbb{C}^n$. Observe that multiplication in $K$ maps to coordinatewise multiplication in $H$. The $\ell_2$ norm on $H$ allows us to efficiently sample a Gaussian distribution $D_\textbf{r}$ over $K$ by sampling such a Gaussian coordinatewise over $H$, although technically this distribution is over the field tensor product $K_\mathbb{R} = K \otimes_\mathbb{Q} \mathbb{R} \cong H$. Furthermore, it satisfies the property that for any $x \in K_\mathbb{R}$ we have the equality of distributions $x \cdot D_\textbf{r}$ and  $D_{\textbf{r}'}$, where $r_i' = r_i \cdot \vert \sigma_i(x)\vert$. When we have an extension of number fields $L/K$ we will denote their respective canonical embeddings $\sigma_L$ and $\sigma_K$ as maps into $H_L$ and $H_K$ to avoid confusion.

\subsubsection{Relative Embeddings}
In the case of an extension $L$ of a number field $K$ it is sometimes more convenient to apply a different order on its embeddings induced by extending embeddings of $K$ to those of $L$. Given a tower $L/K/\mathbb{Q}$ where $K$ has degree $n$ and $L$ has degree $d$ over $K$, there are precisely $n$ embeddings $\sigma_1,...,\sigma_n$ of $K$ into $\mathbb{C}$. Assuming $L/\mathbb{Q}$ is Galois, each of these can be extended to an embedding $\alpha_i: L \rightarrow L$ such that $\alpha_i \vert_K = \sigma_i$. However, these extensions are not unique, and it is easy to see that there are $[L:K] = d$ choices for each $\alpha_i$. In particular, in the case where $L/K$ is a cyclic extension with Galois group generated by $\theta$ it holds that the composite automorphisms $\alpha_i \circ \theta^j( \cdot ), 1 \leq j \leq d$, run through the $d$ choices of $\alpha_i$. Hence for a fixed choice of $\alpha_1,...,\alpha_n$ the $nd$ automorphisms of $L$ can each be uniquely represented by some $\alpha_i \circ \theta^j(\cdot)$, which we denote by $\alpha_i^j(\cdot), 1 \leq i \leq n, 1\leq j \leq d$. Given the usual ordering of embeddings of $K$ this induces two systematic orderings on the embeddings of $L$ by running through either the $i$ or $j$ coordinates first.

\subsection{Cyclic Algebras}
\begin{definition}
Let $K$ be a number field with degree $n$, and let $L$ be a Galois extension of $K$ of degree $d$ such that the Galois group of $L$ over $K$ is cyclic of degree $d$, Gal$(L/K) = \langle \theta \rangle$. For non-zero $\gamma \in K$ we define the resulting cyclic algebra
\begin{align*}
\mathcal{A} = (L/K, \theta, \gamma) := L \oplus u L \oplus ... \oplus u^{d-1}L
\end{align*}
where $\oplus$ denotes the direct sum, $u \in \mathcal{A}$ is some auxiliary generating element of $\mathcal{A}$ satisfying the additional relations $xu = u \theta(x), \forall x \in L$ and $u^d = \gamma$. We will call $d$ the degree of the algebra $\mathcal{A}$. We call such an algebra a division algebra if every element $a \in \mathcal{A}$ has an inverse $a^{-1} \in \mathcal{A}$ such that $aa^{-1} = 1$.
\end{definition}

The relations among $K$, $L$ and $\mathcal{A}$ are illustrated in Fig. \ref{fig:relation}.
\begin{figure}[htb]
               \centering
   \includegraphics[trim=37mm 88mm 40mm 71mm, clip, width=1\linewidth]{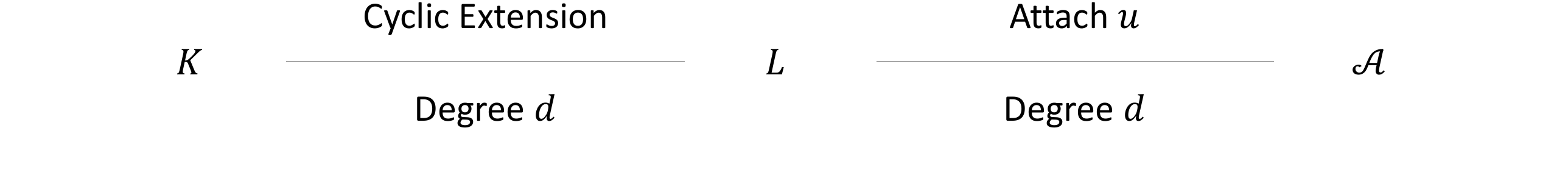}
   \caption{Structure of a cyclic algebra.}\label{fig:relation}
\end{figure}

Since $\theta$ fixes $K$, the center of the cyclic algebra is precisely $K$. Oftentimes the condition $\gamma \in K$ is replaced by the stronger condition $\gamma \in \mathcal{O}_K$, and we will use this condition in our work to guarantee the existence of a certain subring known as the natural order. Note that the division property does not hold for arbitrary $\gamma$, and such algebras are not always easy to construct, which we will discuss later in this section.

We present a matrix representation of elements of $\mathcal{A}$ which proves useful for computing multiplication in cyclic algebras. We can naturally view an element $a \in \mathcal{A}$ as an $d$-dimensional vector Vec$(a)$ over $L$, in which case we can view left multiplication of elements as matrix-vector operations. This is done by defining the map $\phi: \mathcal{A} \rightarrow M_{d \times d}(L)$, where for $x = x_0 + ux_1 + ... + u^{d-1} x_{d-1} \in \mathcal{A}$ with each $x_i \in L$,
\begin{align*}
\phi(x) = \begin{pmatrix}
x_0 & \gamma \theta(x_{d-1}) & \gamma \theta^2(x_{d-2}) & \ldots &\gamma \theta^{d-1}(x_{1}) \\
x_1 & \theta(x_{0}) & \gamma \theta^2(x_{d-1}) & \ldots &\gamma \theta^{d-1}(x_{2}) \\
x_2 & \theta(x_{1}) & \theta^2(x_{0}) & \ldots &\gamma \theta^{d-1}(x_{3}) \\
\vdots & \vdots & \vdots &\ddots & \vdots \\
x_{d-1} & \theta(x_{d-2}) & \theta^2(x_{d-3}) & \ldots & \theta^{d-1}(x_{0}) \\
\end{pmatrix}.
\end{align*}
We call this mapping a left regular representation of $\mathcal{A}$, because it holds for any $a,b \in \mathcal{A}$ that $\phi(a) \text{Vec}(b) = \text{Vec}(ab)$, and that $\phi(ab) = \phi(a) \cdot \phi(b)$. In the case where $\mathcal{A}$ is a division algebra it follows that each $\phi(a)$ is an invertible matrix. Since $\theta$ is well defined on $L_\mathbb{R}$ we abuse notation and extend this map to $\phi: \bigoplus_{i=0}^{d-1}u^i L_\mathbb{R} \rightarrow M_{d \times d} (L_\mathbb{R})$. We derive  lattices from subrings of a cyclic algebra by vectorising their images under $\phi$.
\begin{definition}\label{def1}
Let $\mathcal{A} = (L/K, \theta, \gamma)$ be a cyclic division algebra. A $\mathbb{Z}$-order $\Lambda$ in $\mathcal{A}$ is a finitely generated $\mathbb{Z}$-module such that $\Lambda \cdot \mathbb{Q} = \mathcal{A}$ and that $\Lambda$ is a subring of $\mathcal{A}$ with the same identity element as $\mathcal{A}$. We call $\Lambda$ maximal if there is no $\mathbb{Z}$-order $\Gamma$ such that $\Lambda \subsetneq \Gamma \subsetneq \mathcal{A}$. Here, $\Lambda \cdot \mathbb{Q} = \lbrace \sum_{i=1}^m a_i q_i : a_i \in \Lambda, q_i \in \mathbb{Q}, m \in \mathbb{Z}_{\geq 1} \rbrace$.
\end{definition}
Since we are only concerned with $\mathbb{Z}$-orders in this paper, we will just refer to them as orders.
\begin{example}
The ring of integers $\mathcal{O}_K$ of a number field $K$ is the unique maximal order of a number field. In the case of cyclic algebras a maximal order is not necessarily unique.
\end{example}
An order of particular interest that we will use in our LWE construction is known as the \textit{natural order}, defined as $\Lambda := \bigoplus_{i=0}^{d-1} u^i \mathcal{O}_L$. Unlike in the case of $\mathcal{O}_K$, this order is not necessarily maximal (however, we are going to work with natural orders that are also maximal). Note that in order for $\Lambda$ to be closed under multiplication the element $\gamma$ must lie in $\mathcal{O}_K$.

\subsubsection{Non-Norm Condition}\label{existence}
It is not a priori obvious whether well-defined cyclic algebras or orders actually exist. As observed earlier, the existence of $\gamma$ enforcing the division algebra condition is a key component in constructing such objects. Fortunately, it is sufficient for $\gamma$ to satisfy the so called `non-norm condition' \cite{vehkalahti_densest_2009}.

\begin{prop}\label{non-norm}
The cyclic algebra $\mathcal{A} = (L/K, \theta, \gamma)$ of degree $d$ is a division algebra if and only if none
of the elements $\gamma^t$, $1\leq t \leq d-1$, appears in $N_{L/K}(L)$, where $N_{L/K}$ represents the relative norm of $L$ into $K$.
\end{prop}

In other words, this condition states that the lowest power of $\gamma$ that is norm of some element of $L$, is $\gamma^d$.

\subsubsection{Order Ideals}
Analogous to the use of $\mathcal{O}_K$ ideals in RLWE, we will be interested in ideals of an order $\Lambda$ of a cyclic division algebra $\mathcal{A}$. Although $\Lambda$ is a ring, it is non-commutative - thus there are three types of ideals. A left (respectively right) ideal $\mathcal{I}$ of $\Lambda$ is an additive subgroup of $\Lambda$ such that for any $i \in \mathcal{I}, r \in \Lambda$, we have $r \cdot i \in \mathcal{I}$ (respectively $i \cdot r \in \mathcal{I}$). A two-sided ideal of $\Lambda$ is an additive subgroup that is closed under left and right scaling by $\Lambda$, i.e. a right ideal that is also a left ideal. The sum and product of two ideals $\mathcal{I}, \mathcal{J}$ are defined as usual; $\mathcal{I} + \mathcal{J} = \lbrace i + j: i \in \mathcal{I}, j \in \mathcal{J} \rbrace$ and $\mathcal{I} \cdot \mathcal{J} = \lbrace \sum^m_{l=1} i_l \cdot j_l : i_l \in \mathcal{I}, j_l \in \mathcal{J}, m \in \mathbb{N}\rbrace$. In the case of two-sided ideals we have the standard notion of a fractional ideal; $\mathcal{I}$ is a fractional ideal of $\Lambda$ if $c \mathcal{I} = \mathcal{J}$ for a two-sided ideal $\mathcal{J}$ and some $c \in K$. In the rest of this paper, a (fractional or integral) ideal is always restricted to be two-sided, unless otherwise stated.

We remark that the structure of the collection of two-sided ideals of the natural order is not as simple as those of $\mathcal{O}_K$, or indeed those of an arbitrary maximal order. In a maximal order, the group of two-sided ideals is a free abelian group generated by the prime (e.g. maximal) ideals \cite[Theorem 22.10]{reiner_maximal_1975}, from which one can deduce obvious definitions of inverse and coprime ideals. For a general order $\Lambda$, we define its prime ideals as its maximal two-sided ideals and the inverse of an ideal $\mathcal{I} \subset \Lambda$ is
\begin{align*}
\mathcal{I}^{-1} = \lbrace x \in \mathcal{A} : \mathcal{I} \cdot x \cdot \mathcal{I} \subset \mathcal{I} \rbrace,
\end{align*}
which lines up with the expected definition in the two-sided case (e.g. $\mathcal{I} \cdot \mathcal{I}^{-1} = \mathcal{I}^{-1} \cdot \mathcal{I} = \Lambda$).

For the case of the natural order we do not have such a well-behaved ideal group, but a nice exposition is given in \cite[Section 3]{oggier_quotients_2012-4}. In particular, for a two-sided ideal $\mathcal{I} \subset \Lambda$, $\mathcal{I} \cap \mathcal{O}_K$ is an ideal of $\mathcal{O}_K$. For an ideal $\mathcal{I} \subset \mathcal{O}_K$, $(\mathcal{I}\cdot \Lambda) \cap \mathcal{O}_K = \mathcal{I}$, from which it follows that this intersection map is a surjection onto the ideals of $\mathcal{O}_K$. However, it is not in general an injection since several ideals of $\mathcal{A}$ may have the same intersection with $\mathcal{O}_K$. Since the ideals of $\Lambda$ do not in general form a finitely generated abelian group, we define two ideals $\mathcal{I}, \mathcal{J}$ of $\Lambda$ to be coprime if $\mathcal{I} + \mathcal{J} = \Lambda$.

Nonetheless, since the orders to be constructed in \cref{goodalgebras} are both natural and maximal, it will always hold for a two-sided ideal $\mathcal{I}$ that $\mathcal{I} \cdot \mathcal{I}^{-1} = \mathcal{I}^{-1} \cdot \mathcal{I}= \Lambda$ and $(\mathcal{I}^{-1})^{-1} = \Lambda$. These properties will be required in the proofs of \cref{3,homomorphism}.

\subsubsection{Some Useful Ideals}
For an order $\Lambda$ we define the codifferent ideal
\begin{align*}
\Lambda^\vee = \lbrace x \in \mathcal{A} : \text{Tr}(x \Lambda) \subset \mathbb{Z} \rbrace
\end{align*}
where Tr refers to the reduced trace, defined Tr$(a) := \text{Tr}_{K/\mathbb{Q}}(\text{Trace}(\phi(a)))$. Similarly, for an ideal $\mathcal{I}$ we define the dual ideal
\begin{align*}
\mathcal{I}^\vee = \lbrace x \in \mathcal{A} : \text{Tr}(x\mathcal{I}) \subset \mathbb{Z}\rbrace.
\end{align*}
Since the matrix trace satisfies Trace$(AB)$ = Trace$(BA)$, this definition is two-sided. Note that the codifferent ideal and a general dual ideal may be fractional ideals rather than full ideals, and they satisfy the equality $\mathcal{I}^{\vee} = \Lambda^{\vee} \cdot \mathcal{I}^{-1}$ for any ideal $\mathcal{I}$.

We will also be interested in principal ideals, but must take more care with these than in commutative settings. For a central element $t \in K$, we can define simply $\langle t \rangle = t \cdot \Lambda$, the set of elements of $\Lambda$ divisible by $t$. However, for a general $t$ that does not lie in the center of $\Lambda$ we need the slightly more complex definition
\begin{align*}
\langle t \rangle =  \left \{ \sum^m_{i=1} r_i t s_i: r_i, s_i \in \Lambda, m \in \mathbb{N} \right \},
\end{align*}
which can easily be seen to be a two-sided ideal, moreover the smallest one that contains $t$.

\subsubsection{Orders and Ideals as Integer Lattices}
Any order $\Lambda$ of a cyclic algebra $\mathcal{A} = (L/K,\theta, \gamma)$ has dimension $n d^2$ over $\mathbb{Z}$ and thus generates a lattice of dimension $nd^2$ over $\mathbb{Z}$. We will consider the following representation of these lattices, which extends naturally to ideals of orders as well. Consider an element $x = \bigoplus_{i = 0}^{d-1} u^i x_i \in \Lambda$. We can consider $x$ as a vector over $H_L$ of dimension $d$ by $\sigma_\mathcal{A}(x) := \lbrace \sigma_L(x_0), \sigma_L(x_1),...,\sigma_L(x_{d-1}) \rbrace$. Then, the collection $\sigma_\mathcal{A}(\Lambda)$ forms an integer lattice of dimension $nd^2$. We will refer to this representation as the ``module representation" and will sometimes double index the element $x$, denoting by $x_{i,j}$ the embedding $\sigma_j(x_i)$, and extend this notation in the obvious manner to the space $\bigoplus_{i=0}^{d-1} u^iL_\mathbb{R}$. Though this representation is conceptually simple, we remark that it has some drawbacks in the case where $\vert \sigma_i(\gamma) \vert \neq 1$ for some $i$ when considering sizes of lattice elements; we will choose $\gamma$ carefully in our constructions to remove this issue.

\subsubsection{Gaussian Distributions Over Cyclic Algebras}
As in (R)LWE, we will need to sample Gaussian distributions over our ambient space in certain norms. In the case of RLWE, the continuous Gaussians are sampled in $K_\mathbb{R} \cong H$. Since a cyclic algebra $\mathcal{A}$ can be viewed as an $n$-dimensional algebra over $L$, we use the visualization from the previous subsection and sample our error distributions over $\bigoplus_{i=0}^{d-1} u^iL_\mathbb{R}$, which has the same structure as a vector space as ${H_L}^d$. For simplicity we restrict ourselves to the case when $\vert \sigma_i(\gamma) \vert = 1$ for each $i$. Although this is a strong condition on $\gamma$ it holds in the case where it is a root of unity, which we will enforce later. Otherwise, in order to maintain a norm that is sub-multiplicative the norm and shape of $\gamma$ must be considered.

Explicitly, we just consider the norm of an element of $\mathcal{A}$ to be equal to the norm of the corresponding module element in $L^d$ of dimension $nd^2$ used in \cite{langlois_worst-case_2015}, e.g. $\Vert x \Vert = \Vert (\sigma_L(x_0), \sigma_L(x_1),...,\sigma_L(x_{d-1})) \Vert_2$ for $x = x_0 + u x_1 + ... + u^{d-1} x_{d-1} \in \mathcal{A}$. It is straightforward to check that this is indeed a norm in the case where $\vert \sigma_i(\gamma) \vert = 1$ for each $i$, since $\gamma$ is fixed under $\theta$ and multiplying by $\gamma$ does not change the norm of an entry of $\sigma_L$. It is clear that this norm extends to any $y \in\bigoplus_{i=0}^{d-1} u^i L_\mathbb{R}$ in a natural manner. Now that we have defined a norm, it is easy to define a Gaussian distribution $D_{\textbf{r}}$ on $\mathcal{A}$, or its discrete analogue on $\Lambda$ by sampling over the module ${L_\mathbb{R}}^d$.

\subsubsection{The CRT}
In this subsection we state the CRT for order ideals, and deduce some important consequences. We note that the following lemmas are merely adaptations of those in \cite[Section 2.3.8]{lyubashevsky_ideal_2010} extended to the case of cyclic algebras. The first is just the CRT.
\begin{lemma}
Let $\mathcal{I}_1,...,\mathcal{I}_r$ be pairwise coprime ideals of an order $\Lambda$ of a cyclic algebra $\mathcal{A}$, and let $ \mathcal{I} = \prod_{i=1}^r \mathcal{I}_i$. Then, the natural map $\Lambda \rightarrow \bigoplus_{i=1}^r (\Lambda/\mathcal{I}_i)$ induces an isomorphism $\Lambda/\mathcal{I} \rightarrow \bigoplus_{i=1}^r (\Lambda/\mathcal{I}_i)$.
\end{lemma}
We call a CRT basis for a set of coprime order ideals $\mathcal{I}_1,...,\mathcal{I}_r$ a basis $C = \lbrace c_1,...,c_r \rbrace$ of elements of $\Lambda$ satisfying $c_i = 1 \mod \mathcal{I}_i, c_i = 0 \mod \mathcal{I}_j$ for $i \neq j$.
\begin{lemma}\label{2}
Given pairwise coprime ideals $\mathcal{I}_1,...,\mathcal{I}_r$ of an order $\Lambda$, there is a deterministic polynomial time algorithm that outputs a CRT basis $c_1,...,c_r \in \Lambda$ for those ideals.
\end{lemma}
The proof is the same as in the ring case \cite[Lemma 2.13]{lyubashevsky_ideal_2010}. Using \cref{2} we can efficiently invert the natural CRT isomorphism. Given $a = (a_1,...,a_r) \in \bigoplus_{i=1}^r (\Lambda/\mathcal{I}_i)$, it can be easily checked that its inverse is $b = \sum_{i=1}^r a_i c_i \mod \mathcal{I}$.

The next two lemmas will be required later to construct an efficiently invertible bijection between quotient spaces $\mathcal{I}/\langle q \rangle \cdot \mathcal{I}$ and $\Lambda/\langle q \rangle$.

\begin{lemma}\label{3}
Assuming $q$ is unramified in $L$. Let $\mathcal{I}$ be an ideal of the natural order $\Lambda$ which is maximal and let $\mathcal{J} = q \cdot \Lambda = \langle q \rangle \cdot  \Lambda$, where $q$ is a prime integer and $\langle q \rangle = \prod_{i=1}^r \mathfrak{q}_i$ is a decomposition into prime ideals in $\mathcal{O}_K$. Assume $\gamma \notin \mathfrak{q}_i$ for each $i$. Then, there exists an element $t \in \mathcal{I}\cap \mathcal{O}_K$ such that the ideal $t \cdot \mathcal{I}^{-1} \subset \Lambda$ is coprime to $\mathcal{J}$, and we can compute such a $t$ efficiently given $\mathcal{I}$ and the prime factorization of $\mathcal{J}$.
\end{lemma}

\begin{remark}
The condition on $\gamma$ will be immaterial in our use case, since when $\gamma$ is a unit the only $\mathcal{O}_K$ ideal that contains $\gamma$ is $\mathcal{O}_K$ itself.
\end{remark}

\begin{proof}
For an ideal $\mathcal{I}$ denote by $\overline{\mathcal{I}}$ its intersection with $K$, which is a non-trivial ideal of $\mathcal{O}_K$ (see \cite[Section 3]{oggier_quotients_2012-4}). We apply the corresponding \cite[Lemma 2.14]{lyubashevsky_ideal_2010} to obtain $t \in \overline{\mathcal{I}}$ such that $t \cdot \overline{\mathcal{I}}^{-1}$ and $\overline{\mathcal{J}}$ are coprime as ideals of $\mathcal{O}_K$ and $t \in \overline{\mathcal{I}} \setminus\bigcup_{i=1}^r \mathfrak{q}_i \cdot \overline{\mathcal{I}}$. Assume, for a contradiction, that $t \cdot \mathcal{I}^{-1} + \mathcal{J} \neq \Lambda$ e.g. the ideals are not coprime. Then, there is some maximal ideal $\mathcal{M}$ of $\Lambda$ containing $t \cdot \mathcal{I}^{-1}$ and $\mathcal{J}$. Since $q$ is unramified in $L$ and  $\gamma \notin \mathfrak{q}_i$, by \cite[Propositions 1 and 4]{oggier_quotients_2012-4}, this ideal must be one of the ideals $\mathfrak{q}_i \cdot \Lambda$ since it contains $\mathcal{J}$. Then $t \cdot \mathcal{I}^{-1} \subset \mathfrak{q}_i \cdot \Lambda$ and consequentially $t \in \mathfrak{q}_i \cdot \mathcal{I}$ because $\mathcal{I}\cdot \mathcal{I}^{-1}=\Lambda$ in a maximal order. Since $t$ and $\mathfrak{q}_i$ are central it follows that $t \in \mathfrak{q}_i \cdot \overline{\mathcal{I}}$, a contradiction.
\end{proof}

%

The next lemma will be the one we use in our reduction. As in RLWE, in practice we are interested in the case where $\mathcal{J} = \langle q \rangle$ for a prime integer $q$ and $\mathcal{P} = \Lambda^\vee$. We will use the familiar notation $\mathcal{I}_q := \mathcal{I}/q \cdot \mathcal{I}$ for an ideal $\mathcal{I}$ and $q \in \mathbb{Z}$ throughout the paper.
\begin{lemma}\label{homomorphism}
Let $\Lambda$, $\gamma$ and $q$ be given in \cref{3}. Let $\mathcal{I}, \mathcal{J}$ be ideals of $\Lambda$, with $t \in \mathcal{I}\cap \mathcal{O}_K$ chosen as above such that $ t  \cdot \mathcal{I}^{-1}$ and $\mathcal{J}$ are coprime as ideals, and let $\mathcal{P}$ denote an arbitrary fractional ideal of $\Lambda$. Then, the function $\chi_t: \mathcal{A} \rightarrow \mathcal{A}$ defined as $\chi_t(x) = t \cdot x$ induces a module isomorphism from $\mathcal{P}/\mathcal{J} \cdot \mathcal{P} \rightarrow \mathcal{I}\cdot \mathcal{P}/ \mathcal{I} \cdot \mathcal{J} \cdot \mathcal{P}$.  Furthermore, in the case $\mathcal{J} = \langle q \rangle$ for a prime integer $q$ we can efficiently compute the inverse.
\end{lemma}

\begin{proof}
The proof is similar to that of \cite{lyubashevsky_ideal_2010}. Since $t$ lies in the center of $\Lambda$ it is clear that multiplication by $t$ induces a module homomorphism. Given the map $\chi_t: \mathcal{P} \rightarrow \mathcal{I}\cdot \mathcal{P}/\mathcal{I} \cdot \mathcal{J}\cdot \mathcal{P}$ and $j \in \mathcal{J} \cdot \mathcal{P}$, $\chi_t(j) = t \cdot j \in \mathcal{I} \cdot \mathcal{J} \cdot \mathcal{P}$, so it is clear that $\mathcal{J} \cdot \mathcal{P}$ is in the kernel of this map. Conversely, if $\chi_t(x) = 0$ then $t \cdot x \in \mathcal{I} \cdot \mathcal{J} \cdot \mathcal{P}$, from which it follows that $\mathcal{I}^{-1} \cdot t \cdot x \subset \mathcal{J} \cdot \mathcal{P}$. From the definition of coprime, $t \cdot \mathcal{I}^{-1} + \mathcal{J} = \Lambda$, from which it follows that there exists $a \in t \cdot \mathcal{I}^{-1}, b \in \mathcal{J}$ such that $a + b = 1$. Hence $x = (a+b)\cdot x = a \cdot x + b \cdot x$. Since $a \cdot x, b \cdot x \in \mathcal{J} \cdot \mathcal{P}$ it follows that $x \in \mathcal{J} \cdot \mathcal{P}$, from which injectivity follows immediately.

To demonstrate efficient invertibility, we must work slightly harder. Now let $ \mathcal{J} = \langle q \rangle$. Compute $t$ as in \cref{3} and observe that the bijection $\chi_t : \Lambda_q \rightarrow \mathcal{I}_q$ is an additive homomorphism. Thus, it suffices to compute the inverse of all elements of a $\mathbb{Z}$ basis of $\mathcal{I}_q$, since then any element can be inverted by computing its representation in this basis and inverting that. We construct such a basis as follows. First, choose $n^2 \cdot d^4$ elements $x_i, i = 1,..., n^2 \cdot d^4$ from $\Lambda_q$ uniformly at random and compute $y_i = \chi_t (x_i)$ for each $i$. It follows that each $y_i$ is a uniformly random element of $\mathcal{I}_q$. Then, with high probability the $y_i$'s form a spanning set of $\mathcal{I}_q$ (see the proceeding lemma), which we can reduce to a $\mathbb{Z}$ basis $y_1',...,y_{n \cdot d^2}'$. This basis satisfies the desired property that each element has a known inverse. If this algorithm fails (e.g. there is no suitable basis $y_1',...y_{n \cdot d^2}'$), we repeat, choosing a fresh set of elements $x_1,...,x_{n^2 \cdot d^4}$ until we succeed.
\end{proof}

\begin{lemma}
Given a set of $n^2 \cdot d^4$ independent and uniformly random elements $\Xi \subset \mathbb{Z}_q^{n \cdot d^2}$, the probability that $\Xi$ contains no set of $n \cdot d^2$ linearly independent vectors (over $\mathbb{Z}$) is exponentially small in $d$.
\end{lemma}
This lemma is a straightforward adaptation of Corollary 3.16 of \cite{regev_lattices_2009}.

\subsection{Lattice Problems}\label{latticeproblems}
Computational problems on lattices represent the foundations of the security of (R)LWE, and will do so for our Cyclic LWE as well. The standard lattice problems are as follows.
\begin{definition}
Let $\Vert \cdot \Vert$ be some norm on $\mathbb{R}^n$ and let $\xi \geq 1$. Then the approximate Shortest Vector Problem (SVP$_\xi$) on input a lattice $\mathcal{L}$ is to find some non-zero vector $\textbf{x}$ such that $\Vert \textbf{x} \Vert \leq \xi \cdot \lambda_1(\mathcal{L})$.
\end{definition}
\begin{definition}
Let $\Vert \cdot \Vert$ be some norm on $\mathbb{R}^n$ and let $\xi \geq 1$. Then the (approximate) Shortest Independent Vectors Problem (SIVP$_\xi$) on input a lattice $\mathcal{L}$ is to find $n$ linearly independent non-zero vectors $\textbf{x}_1,...,\textbf{x}_n$ such that $\max_{i}(\Vert \textbf{x}_i \Vert) \leq \xi \cdot \lambda_n(\mathcal{L})$.
\end{definition}
\begin{definition}
Let $\Vert \cdot \Vert$ be some norm on $\mathbb{R}^n$, let $\mathcal{L}$ be a lattice, and let $d < \lambda_1(\mathcal{L})/2$. Then the Bounded Distance Decoding problem (BDD$_{\mathcal{L}, d}$) on input $\textbf{y} = \textbf{x} + \textbf{e}$ for $\textbf{x} \in \mathcal{L}$ and $\Vert \textbf{e} \Vert \leq d$ is to compute $\textbf{x}$, or equivalently $\textbf{e}$.
\end{definition}

The above problems are all well investigated, and believed to be sufficiently hard to base post-quantum cryptographic security on; there are no known algorithms for any of these problems (for suitable parameters) running in polynomial time in dimension $n$.

Unfortunately, these problems are not directly suitable for CLWE, where we will be interested in their adaptations to lattices generated by order ideals, similarly to how ideal lattices are used the ring case. Specifically we have the same problems on lattices that they induce under the map $\sigma_\mathcal{A}(\cdot)$. So, SVP becomes:
\begin{definition}
Let $\mathcal{A}$ be a cyclic algebra, let $\mathcal{I}$ be some (possibly fractional) ideal of the natural order $\Lambda$. Then, for an approximation factor $\xi \geq 1$, the $\mathcal{A}$-SVP$_\xi$ is to find a non-zero element $a \in \mathcal{I}$ such that $\vert a \vert := \Vert \sigma_\mathcal{A}(a) \Vert_2 \leq \xi \cdot \lambda_1(\mathcal{I})$, where as usual $\lambda_1(\mathcal{I})$ denotes the minimal length of elements of $\mathcal{I}$ in the given norm.
\end{definition}
\begin{remark}
When we use these problems in our security reductions, we will assume that the ideals are in fact \textit{integral} ideals (e.g. we exclude fractional ideals). Observe that this may be done without loss of generality, since solving the $\mathcal{A}$-SVP problem on the fractional ideal $\mathcal{I}$ may be done by solving it on the integral ideal $c \mathcal{I}$ (where $c \in K$ is the element such that $c\mathcal{I}$ is integral) and rescaling the solution.
\end{remark}
Essentially we have a specialized version of the SVP problem; we must find an element of $\mathcal{I}$ with minimal norm (up to approximation factor) in the ideal $\mathcal{I}$. The extension of SIVP to $\mathcal{A}$-SIVP is analogous, but since we consider our objects as $\mathbb{Z}$-lattices we require the independent `vectors' $a_1,...,a_r$ to be linearly independent over $\mathbb{Z}$. For BDD, we need a suitable ambient space, and use the following definition.
\begin{definition}
Let $\mathcal{A}$ be a cyclic algebra, let $\mathcal{I}$ be some (possibly fractional) ideal of a maximal $\mathbb{Z}$-order $\Lambda$, and let $\delta < \lambda_1(\mathcal{I})/2$. Then the $\mathcal{A}$-BDD$_{\mathcal{I}, \delta}$ problem, on input $y = x + e$ for $x \in \mathcal{I}$ and $e \in \bigoplus_{i = 0}^{d-1} u^i L_{\mathbb{R}}$ satisfying $\vert e \vert \leq \delta$, is to compute $x$.
\end{definition}

\subsection{The Learning With Errors Problem}
We will briefly recall the initial Learning With Errors (LWE) problem here; in \cref{sec3} we will extend it to cyclic algebras. The problem comes in two forms; search and decision, both of which are based on the LWE distribution. Let $n$ and $q$ be positive integers, and let $\alpha > 0$ be some error parameter. Define $\mathbb{T} := \mathbb{R}/ \mathbb{Z}$, the unit torus.
\begin{definition}
For a secret $\textbf{s} \in \mathbb{Z}_q^n$, a sample $(\textbf{a}, b) \leftarrow A_{\textbf{s}, \alpha}$ is taken by sampling a uniformly random vector $\textbf{a} \in \mathbb{Z}_q^n$ and $e \leftarrow D_\alpha$ and outputting $(\textbf{a},b) = (\textbf{a}, \langle \textbf{a}, \textbf{s} \rangle /q + e \mod \mathbb{Z})$.
\end{definition}
Given the above distribution, the LWE problem comes in two forms.
\begin{definition}
The search LWE problem is to recover $\textbf{s}$ from a collection of samples $A_{\textbf{s}, \alpha}$. The decision LWE problem on input a collection of samples on $\mathbb{Z}_q^n \times \mathbb{T}$ is to decide whether they are uniform samples or were taken from $A_{\textbf{s}, \alpha}$ for some secret $\textbf{s}$, providing the samples were taken from one of these distributions.
\end{definition}
Typically, the number of samples provided in each of these problems depends on the application. Since the decision problems has a probabilistic element, we will be interested in the advantage of the algorithms that solve it, which is defined as the difference between their acceptance probabilities on samples from an LWE distribution $A_{\textbf{s}, \alpha}$ and the uniform distribution. In practice, the decision problem is of more interest in cryptography.

We will not define the popular extensions of these problems to number fields or modules, known as Ring-LWE and Module-LWE, but the unfamiliar reader may find details in \cite{lyubashevsky_ideal_2010} and \cite{langlois_worst-case_2015} respectively, both of which we reference frequently in this work.


\section{The CLWE Problem}\label{sec3}
In this section we present the general definition of CLWE together with justifications for choices made in the definition, as well as constructions of specific algebras to use. We will save the security properties for \cref{sec4}.
\begin{definition}
Let $L/K$ be a Galois extension of number fields of dimension $[L : K] = d$, $[K: \mathbb{Q}] = n$ with cyclic Galois group generated by $\theta (\cdot)$. Let $\mathcal{A} := (L/K, \theta, \gamma)$ be the resulting cyclic algebra with center $K$  and invariant $u$ with $u^d = \gamma \in \mathcal{O}_K$. Let $\Lambda$ be an order of $\mathcal{A}$. For an error distribution $\psi$ over $\bigoplus_{i = 0}^{d-1} u^i L_{\mathbb{R}}$, an integer modulus $q \geq 2$, and a secret $s \in \Lambda^\vee_q$, a sample from the CLWE distribution $\Pi_{q, s, \psi}$  is obtained by sampling $a \leftarrow \Lambda_q$ uniformly at random, $e \leftarrow \psi$, and outputting $(a,b) =(a, (a \cdot s)/q + e \mod \Lambda^\vee) \in (\Lambda_q, \bigoplus_{i = 0}^{d-1} u^i L_{\mathbb{R}})/\Lambda^\vee$.
\end{definition}
\begin{remark}
Unlike in commutative spaces, the order of multiplication of $a$ and $s$ is important; our choice is $(a \cdot s)$, but similar security properties would hold if one took $(s \cdot a)$ instead. Also observe that our modulo reduction in the second coordinate of the pair is well defined, since $(a \cdot s) \in \Lambda^\vee_q$.
\end{remark}
As usual, the associated CLWE problem will come in search and decision variants.
\begin{definition}
Let $\Pi_{q,s, \psi}$ be a CLWE distribution for parameters $q \geq 2$, $s \in \Lambda^\vee_q$, and error distribution $\psi$. Then, the search CLWE problem, which we denote by CLWE$_{q, s, \psi}$, is to recover $s \in \Lambda^\vee_q$ from a collection of independent samples from $\Pi_{q,s, \psi}$.
\end{definition}
We do not state the number of samples allowed for this (or the next) problem, as typically it depends on the application.
\begin{definition}
Let $\Upsilon$ be some distribution on a family of error distributions over $\bigoplus_{i = 0}^{d-1} u^i L_{\mathbb{R}}$ and $U_\Lambda$ denote the uniform distribution on $(\Lambda_q, (\bigoplus_{i = 0}^{d-1} u^i L_{\mathbb{R}})/\Lambda^\vee)$. Then, the decision CLWE problem, written D-CLWE$_{q, \Upsilon}$, is on input a collection of independent samples from either $\Pi_{q, s, \psi}$ for a random choice of $(s, \psi) \leftarrow U(\Lambda^\vee_q) \times \Upsilon$ or from $U_\Lambda$, to decide which is the case with non-negligible advantage.
\end{definition}

\subsection{Discussions}\label{discussions}

\subsubsection{Relation to Module-LWE}
First, we explain why we choose the order of multiplication $a \cdot s$. As discussed in the introduction, the transformation from a (primal) RLWE sample to $n$ related LWE samples provides our motivation. Here, one RLWE sample $a \cdot s + e$, where $a,s,e \in R_q \cong \frac{\mathbb{Z}_q[x]}{x^n+1}$, generates $n$ LWE samples by considering the multiplication operation as $A\textbf{s} + \textbf{e}$, where $A :=$ rot$(a)$ is a negacyclic matrix. For appropriate choices of error distributions, this is precisely $n$ LWE samples with the exception that there is some structure in the matrix $A$. By ordering the multiplication $a \cdot s$, we get a similar transform from CLWE to MLWE. Assuming for now that we have a discretized form of CLWE, and observing that for $q \in \mathbb{Z}$ we have $\Lambda_q \cong \bigoplus_{i=0}^{d-1} u^i \mathcal{O}_L/q \mathcal{O}_L$ (see \cite{oggier_quotients_2012-4}), we transform a CLWE sample $a \cdot s + e$ into matrix-vector form to get $\phi(a) \cdot \textbf{s} + \textbf{e}$, where $\textbf{s}$ and $\textbf{e}$ are vectors of dimension $d$ over $\mathcal{O}_L/q \mathcal{O}_L$. Setting $A = \phi(a)$, one can see that for appropriate choices of error distribution this is similar to $d$ samples from the MLWE distribution with some additional structure in the matrix $A$, as intended.

\subsubsection{The Natural Order vs. Maximal Order}
We consider $\Lambda$ the natural order or a maximal order. The natural order is simple to construct and represent, whereas finding a maximal order is computationally slow. Additionally, the natural order is somewhat orthogonal, in the sense that it has the same span in each $u^i$ coordinate independently of the other coordinates. This is advantageous when considering the relation to MLWE, where the module is always taken to be the full module $\mathcal{O}_K^d$.

As mentioned above, two-sided ideals in a maximal order form a free abelian group, which is not necessarily the case in the natural order. Further, as lattices, a maximal order gives denser sphere packing than the natural order, since the latter is a sublattice. Fortunately, we will construct in \cref{goodalgebras} cyclic algebras whose natural order is also maximal, thus enjoying both the simplicity of the natural order and the convenience of a maximal order.

\begin{example}
Quaternion algebra over $\mathbb{Q}$ is defined by $\mathbb{H} = \left\{ x + yj: x, y \in \mathbb{Q}(i)\right\}$, with the usual relations $i^2=j^2=-1$ and $ij = -ji$. It can be seen as a cyclic division algebra $(\mathbb{Q}(i)/\mathbb{Q}, \overline{(\cdot)}, -1)$ where $\overline{(\cdot)}$ denotes the complex conjugate and $-1$ is a non-norm element. A quaternion has matrix representation
$$\left(\begin{array}{cc} x & -\overline{y} \\ y & \overline{x} \end{array}\right).$$

The \textit{Lipschitz integers} $\mathcal{L} \subset \mathbb{H}$ form the (non-maximal) natural order $\mathcal{L}=\left\{x+yj : x, y \in \mathbb{Z}[i]\right\}.$
The maximal Hurwitz order is given by
$$\mathcal{H} = \left\{a+bi+cj+d(-1+i+j+ij)/2 : a, b, c, d \in \mathbb{Z}\right\}.$$ It is easy to check that, as $\mathbb{Z}$-lattices of dimension $4$, the Lipschitz order is a sublattice of the Hurwitz order, of index $2$.
\end{example}

\subsubsection{A Pair of Number Fields}
In MLWE, we are free to choose the dimension of our module over the underlying number field $K$. However, in the cyclic algebra case we are restricted to cases where we can find $L,K$, and $\gamma$ such that $\mathcal{A} = (L/K, \theta, \gamma)$ is well defined. From a theoretical standpoint it is not immediately clear whether we want to consider asymptotic security in terms of $n$ or $d$, but following our motivation from MLWE we suggest that $n$ is likely the suitable choice since the module dimension $d$ is typically small in applications using MLWE, whereas the dimension of the underlying field $K$ is large. However, there seems to be no a priori reason why with the right techniques one could not consider both $n$ and $d$ asymptotically; the only case a cyclic algebra precludes is high dimensional MLWE over a low dimension number field $L$, because the parameter $d$ occurs in both the module and field dimension.

\subsection{Evading BCV Style Attacks}\label{BCWappendix}
In our CLWE construction we have enforced that $\gamma$ is selected so that $\mathcal{A}$ is a division algebra. We do this to avoid attacks in the style of \cite{bootland_security_2018} on the $m$-RLWE protocol. For $m = 2$, the $m$-RLWE protocol of \cite{pedrouzo-ulloa_ring_2016} can be considered as a structured variant of MLWE, where the matrix $A$ in the operation $A\textbf{s} + \textbf{e}$ is a negacyclic matrix over some ring $R_q$. More explicitly, $2$-RLWE considers the tensor product of two fields $K = K_1 \otimes K_2$ and runs the LWE assumption in the ring of integers $R_q$. The example use case given in \cite{pedrouzo-ulloa_ring_2016} considers power-of-two cyclotomics $K_1, K_2$ defined by the polynomials $x^{k_1} + 1$ and $y^{k_2} + 1$ respectively, claiming that the resulting problem in $R_q = \frac{\mathbb{Z}_q[x,y]}{(x^{k_1}+1, y^{k_2} +1)}$ effectively corresponds to an RLWE problem of dimension $k_1 \cdot k_2$ due to an obvious homomorphism between $K$ and the two-power cyclotomic field $L$ of degree $k_1 \cdot k_2$. The problem also represents a structured MLWE instance over $\frac{\mathbb{Z}_q[x]}{(x^{k_1}+1)}$ of dimension $k_2$.

However, the observation of \cite{bootland_security_2018} is that there is a smaller field $K'$ containing $K_1$ such that there is a homomorphism from $K$ into $K'$ with a well defined image for $y$. This is because the roots of distinct two-power cyclotomic polynomials are algebraically related. For example, in the case $k_1 = 8, k_2 = 4$, it is clear that the map taking $y$ to $x^2$ and fixing $K_1$ is a well defined homomorphism from $K$ to $K_1$. Using this homomorphism, \cite{bootland_security_2018} simplifies the problem of solving one $2$-RLWE instance by considering it as four RLWE instances in dimension $k_1$ rather than one instance in dimension $k_1 \cdot k_2$, essentially removing the module dimension $k_2$ from the problem.

We argue that the non-norm condition of $\gamma$ precludes the existence of a homomorphism removing the module structure by taking a well defined cyclic algebra $\mathcal{A} = (L/K, \theta, \gamma)$ to a smaller subfield containing $K$. We restrict our search to maximal subfields of $\mathcal{A}$, since any subfield is contained in at least one maximal subfield. It is a well known result on division algebras that any maximal subfield $E$ of $\mathcal{A}$ contains $K$ and satisfies $[E:K] = d$, and that in the case of a cyclic division algebra $\mathcal{A}$ there is a choice of $u' \in \mathcal{A}$ such that the cyclic algebra $\mathcal{A}' := \bigoplus_j u'^j E$ is isomorphic to $\mathcal{A}$ (see Section 15.1, Proposition a of \cite{pierce_cyclic_1982-1}). Assume, for a contradiction, that we had such a homomorphism $\chi : \mathcal{A} \rightarrow L$, where without loss of generality we assume the maximal subfield is $L$ by the aforementioned proposition. Since $L$ is Galois, the restriction of $\chi$ to $L$ is an automorphism of $L$. It is clear that $\chi$ must agree on conjugates, since $\chi(u) \cdot \chi(\ell) = \chi(u \cdot \ell) = \chi(\theta(\ell) \cdot u) = \chi(u) \cdot \chi(\theta(\ell))$ for any $\ell \in L$. However, this contradicts $\chi$ being injective on $L$ and it follows that no such homomorphism exists. Hence we conclude that the attack style of \cite{bootland_security_2018} does not threaten our algebraic structure.

On the other hand, Appendix \ref{badgamma} shows that if $\gamma$ violates the non-norm condition, then those instances of the CLWE problem are potentially vulnerable. To sum up, the non-norm condition is crucial to the hardness of the CLWE problem.

\subsection{Concrete Algebras for CLWE}\label{algebras}

In order to apply the CLWE assumption in a practical cryptosystem one must choose a concrete algebra as an ambient space. More generally, we are interested in finding families of algebras suitable for CLWE that allow for asymptotic analysis and varied security levels. Our search for algebras is motivated by the restrictions and conditions discussed in the previous section. In particular, we are interested in cyclic division algebras satisfying the following properties:
\begin{itemize}
\item The non-norm element $\gamma$ must lie in $\mathcal{O}_K$ to keep the natural order closed under multiplication, and should satisfy $\vert \gamma \vert = 1$ in order to maintain both the coordinatewise independence and sub-multiplicative properties of the norm\footnote{We abbreviate the condition $\vert \sigma_i(\gamma) \vert=1$ for all $i$ by $\vert \gamma \vert = 1$, since in fact these are equivalent for algebraic $\gamma$.}.
\item The dimension $n:= [K: \mathbb{Q}]$ of the division algebra should be large and the degree $d:= [L:K]$ should be small. This is to maintain the analogy with structured MLWE (the degree corresponds to the module rank) and follows from the search-decision reduction, which takes time polynomial in $n$ but not in $d$.
\item The base field $K$ should be cyclotomic and $q$ should split completely in $K$. This is also a result of the methodology of the search-decision reduction, which uses the well understood factorization of $\langle q \rangle$ in $\mathcal{O}_K$. In addition, since the bulk of lattice based cryptography is done over cyclotomic fields, we consider algebras which are small extensions of these as somewhat natural. We observe that an improved proof of decision security may allow this point to be dropped, whereas the other two points feel more integral.
\end{itemize}
Although significant effort has been expended by coding theorists to construct cyclic division algebras satisfying a variety of conditions, such as in \cite{vehkalahti_densest_2009}  or \cite{lahtonen_construction_2008}, we find ourselves with a fairly unique set of restrictions. In particular, for reasons relating to desired applications, the majority of algebras used in coding theory are either of small total dimension or have small $[K:\mathbb{Q}]$ and scale asymptotically in $[L:K]$. Since we are interested in scaling up $K$ asymptotically, we will have to build novel algebras satisfying the above requirements ourselves. We will, however, make heavy use of the following theorem as an intermediate step. Here $\zeta_m$ denotes a primitive $m^\text{th}$ root of unity where $\varphi(m)=n$ is the degree of the base field $K = \mathbb{Q}(\zeta_m)$.

\begin{theorem}[\cite{lahtonen_construction_2008}]\label{lahtonenalgebras}
Let $m = p^a$ be a prime power and let $K = \mathbb{Q}(\zeta_m)$. Then, there exist infinitely many cyclic Galois extensions $L/K$ of degree $m$ such that $\zeta_m^i$ is not a norm of $L/K$ for $0 <i < m$.
\end{theorem}

We remark that the theorem is effective in the sense that it provides an explicit description of $L$, and we provide a summary of the recipe for constructing $L$. The crucial aspect of its construction is that $L$ is a subfield of some cyclotomic extension of $K$, $K(\zeta_q')$ for a prime $q'$, but we present its full description for completeness.

First, find some prime $q'$ such that $q' = 1 \mod p^a$ but $q' \neq 1 \mod p^{a+1}$, so that $p^a$ is the highest power of $p$ dividing $q'-1$\footnote{It is easy to show that infinitely many primes satisfying this condition always exist by appealing to classical theorems of Chebotarev or Dirichlet.}. Set $M = K(\zeta_{q'})$ so that by coprimality $M = \mathbb{Q}(\zeta_{mq'})$. Then Gal$(M/K)$ is a cyclic group of order $q'-1$ generated by some automorphism $\sigma$. Denote by $L$ the subfield of $M$ fixed by $\sigma^m$. Then $[L:K] = m$ by the fundamental theorem of Galois theory and the extension is both cyclic and Galois. Finally, localization theory is used to show that the powers of $\zeta_m$ are not norms in this extension. In this way, the theorem constructs $L$ explicitly.

The part of this theorem of our interest is that it allows us to scale $K$ asymptotically, but this comes with a drawback of very high degree $L$, \emph{i.e.}, it only permits a degree-$m$ extension $L$ of a degree-$\varphi(m)$ base field $K$. We present a new method that uses this theorem as a starting point to construct good algebras satisfying our restrictions. More precisely, our construction will begin with \cref{lahtonenalgebras} and then use elementary methods from Galois theory to build more favourable fields.

\subsubsection{Constructions Using Subfields}\label{goodalgebras}
We squash the field $L$ from \cref{lahtonenalgebras} to a subfield $M$ of small index over the base $K$ satisfying the necessary properties to generate a cyclic algebra.
\begin{theorem}\label{primepoweralgebras}
Let $K = \mathbb{Q}(\zeta_m)$, where $\varphi(m)=n$, be a prime power cyclotomic with $m = p^a$ for some integer $a$ and prime $p$. Then, there exists a cyclic Galois extension $M/K$ of any index $d$ dividing $m$ within which $\zeta_m$ satisfies the non-norm condition.
\end{theorem}

\begin{remark}
Since the proof will provide an explicit description of $M$, the correct interpretation of this theorem is that we can construct cyclic division algebras $\mathcal{A} = (M/K, \theta, \gamma)$ with $\langle \theta \rangle = \text{Gal}(M/K), \gamma = \zeta_m, K = \mathbb{Q}(\zeta_m),$ and $[M:K]$ is any divisor of $m = p^a$. Fig. \ref{fig:L-M-K} shows all possible cases of intermediate field $M$ between $K$ and $L$.
\end{remark}

\begin{proof}
Let $K= \mathbb{Q}(\zeta_m)$ for a fixed $m = p^a$ with prime $p$ and integer $a$. Following the construction of \cref{lahtonenalgebras} fix a cyclic Galois extension $L/K$ of degree $m$ such that  $\zeta_m^i$ is not a norm of an element of $L$ into $K$ for any $i = 1,2,\dots, m-1$. We will choose $M$ as a suitable intermediate extension $L/M/K$. Let $\sigma$ denote the generator of Gal$(L/K)$, an automorphism of degree $m$. For $d$ dividing $m$, $\sigma^{d}$ fixes an extension $M$ of $K$ with $[L:M] = \vert \text{Gal}(L/M) \vert = m/d$ and it follows from the tower lemma that $[M:K] = d$. We will show that $M$ is a satisfactory extension of $K$.

First, since Gal$(L/M)$ is a normal subgroup of Gal$(L/K)$ we see that $M/K$ is a normal, and hence Galois\footnote{Since in this case all extensions are separable.}, extension. It follows from standard Galois Theory that
\begin{align*}
\text{Gal}(M/K) \cong \text{Gal}(L/K)/\text{Gal}(L/M).
\end{align*}
Both groups in the quotient are cyclic, and so Gal$(M/K)$ is cyclic with some generator $\theta$. Furthermore, this isomorphism also allows us to deduce $\vert \text{Gal}(M/K) \vert = d$.

We've shown that $M/K$ is a cyclic Galois extension of degree $d$; we are left to show that $\zeta_m^i$ is not a norm for $i =1,\dots, d-1$. Let $\overline{L}$ denote $N_{L/K}(L^\times)$ and $\overline{M}$ denote $N_{M/K}(M^\times)$. Say $\zeta_m^i \in \overline{M}$, fixing $x \in M$ such that $N_{M/K}(x)= \zeta_m^i$. Now by transitivity of the norm,
\begin{align*}
N_{L/K}(x) &= N_{M/K}(N_{L/M}(x)) \\
&= N_{M/K}(x^{m/d}) \\
&= \zeta_m^{(m/d)i}
\end{align*}
where the first equality follows from $x \in M$ and the second since the norm is multiplicative. $\overline{L}$ does not contain any power of $\zeta_m$ except $\zeta_m^m = 1$ since $\zeta_m$ is a non-norm element in $L/K$, so it follows that $m \vert (m/d)i$ and so $d \vert i$. From this we conclude that $\zeta_m, \zeta_m^2,\dots,\zeta_m^{d-1}$ do not lie in $\overline{M}$ and so $\zeta_m$ satisfies the non-norm condition.
\end{proof}

\begin{figure}[htb]
               \centering
   \includegraphics[trim=60mm 30mm 110mm 53mm, clip, width=0.70\linewidth]{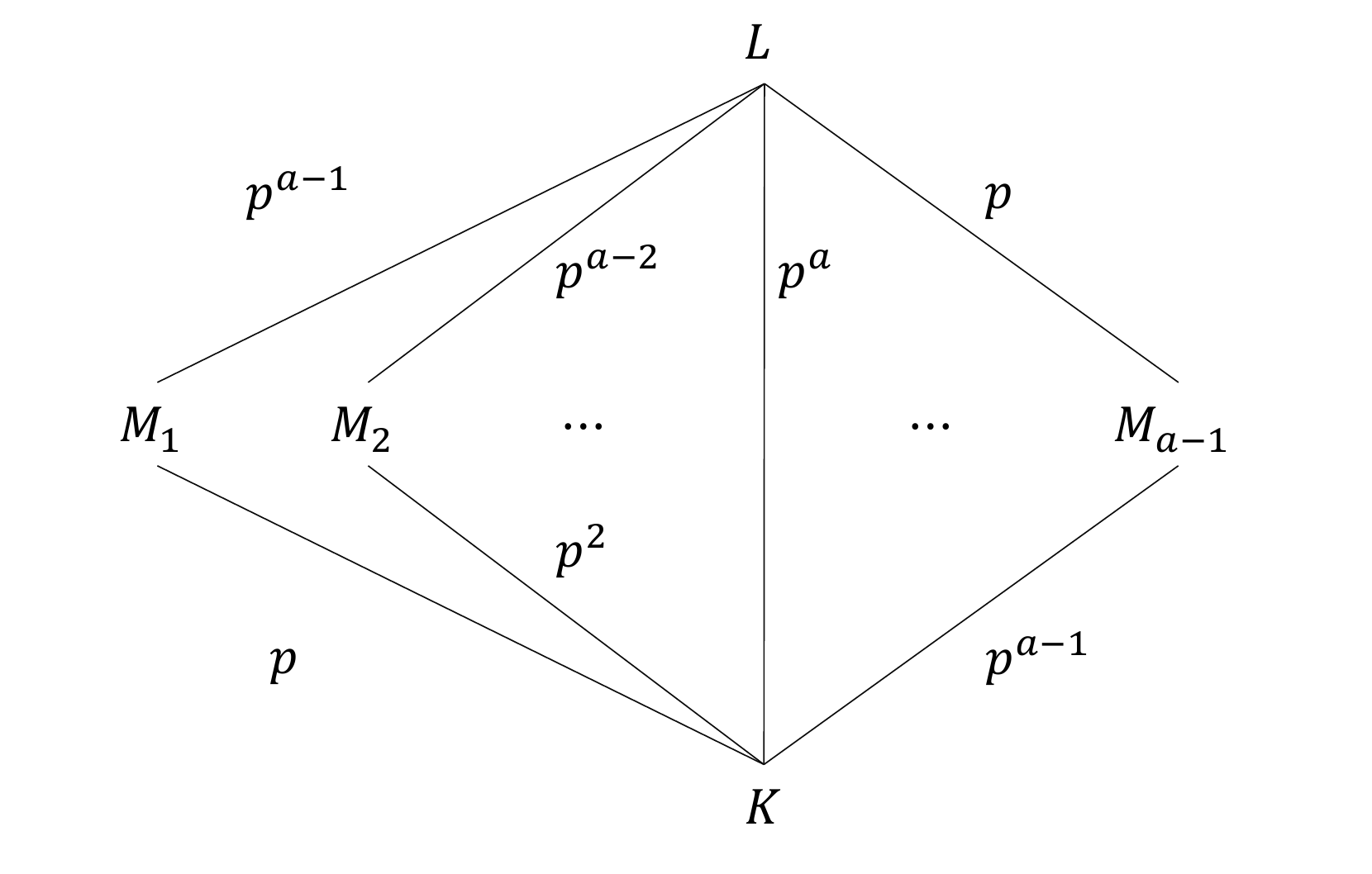}
   \caption{Cyclic subfields between $L$ and $K$.}\label{fig:L-M-K}
\end{figure}

\begin{remark}
We presented the proof in the above form for ease of legibility, but it is straightforward to extend the argument in the final paragraph to show that $\zeta_m^{jd+1}$ satisfies the non-norm condition for any $j =0,1,\dots, (m/d) -1$.
\end{remark}
This is an effective construction that allows us to build cyclic division algebras of the form $\mathcal{A} = (M/K, \theta, \gamma)$ where $\vert \gamma \vert = 1$, $K$ is an arbitrary prime power cyclotomic, and $M$ is an extension of $K$ with degree divisible by the prime $p$. For cryptographically relevant examples, we can consider degree $2$ or $4$ extensions of a $2$-power cyclotomic or degree $3$ extensions of a $3$-power cyclotomic. Given the impossibility result of Appendix \ref{impossiblealgebras} and the restriction on the absolute value of $\gamma$ we view these algebras as essentially the best possible, at least for the case where $K$ is a prime-power cyclotomic.

As discussed in \cref{discussions}, the natural order is not necessarily a maximal order. Nevertheless, the following theorem shows that the specific family of algebras we have constructed in \cref{goodalgebras} represents a lucky case (its proof is given in Appendix \ref{appendix:natural-maximal}).

\begin{theorem}\label{natural_is_maximal}
For the family of cyclic division algebras $\mathcal{A}=(M/K,\theta, \zeta_m)$ constructed in \cref{primepoweralgebras}, the natural order of $\mathcal{A}$ is maximal.
\end{theorem}

This makes our constructed family of algebras very attractive, as it enjoys both the simplicity of the natural order and the nice property of a maximal order.

\begin{remark}
In the context of multiblock space-time coding \cite{lahtonen_construction_2008}, the construction of \cref{lahtonenalgebras} allows for a space-time code for $m$ antennas and $\varphi(m)$ blocks, \textit{i.e.}, a relatively small number of blocks. With our new construction \cref{primepoweralgebras}, any number $\varphi(mk)$, $k \in \mathbb{N}$ of blocks becomes possible. Further, using a maximal order leads to optimum coding gains; it was not realized in \cite{lahtonen_construction_2008} that the natural order from \cref{lahtonenalgebras} is actually maximal.
\end{remark}

\subsection{Sample Parameters}\label{sampleparameters}

Now that we have discussed our techniques for constructing suitable number fields we proceed to demonstrate that these methods are able to attain cryptographically relevant dimensions. In this section, we present a small selection of proof-of-concept dimensions in Table \ref{Sample Cyclic Algebra Parameters} where we take our motivation for choices of dimension from KYBER and NewHope, since they are the successful second round NIST candidates whose methods are most similar to our own. Thus we aim for dimensions in the region of between $512$ and $1024$, dimensions proposed for both NewHope and KYBER (which also achieves dimension $768$). Of course, these schemes are restricted to having power-of-two ring dimension $n$ and so their choices of dimension may not be optimal in general, but FrodoKEM \cite{bos_frodo:_2016}, a plain LWE scheme, suggests dimensions in around the same range, specifically $640$, $976$, and $1344$, so we consider dimensions in this region a sensible starting point. Corresponding to KYBER and other MLWE based schemes we will set a small `module' rank $d:=[\mathcal{A}:L]$. We are constricted in our choice of fields by the fact that $d$ appears as a square in the total dimension $N = nd^2$, but for the most part we are able to work around this problem.
\begin{table}[ht!]
\begin{center}
\caption{Sample Parameters of Cyclic Algebras. The subfield method is given in \cref{algebras}, while the compositum method is given in Appendix \ref{appendix:compositum}.}
 \begin{tabular}{||c c c c c||}
 \hline
 Method & Center $K$ & $n=[K:\mathbb{Q}]$ & $d=[L:K]$ & Total Dimension $N=nd^2$ of $\mathcal{A}$ \\ [0.5ex]
 \hline\hline
 Subfield & $\mathbb{Q}(\zeta_{81})$ & 54 & 3 & 486 \\
 \hline
 Subfield & $\mathbb{Q}(\zeta_{256})$ & 128 & 2 & 512 \\
 \hline
 Subfield & $\mathbb{Q}(\zeta_{64})$ & 32 & 4 & 512 \\
 \hline
 Subfield & $\mathbb{Q}(\zeta_{512})$ & 256 & 2 & 1024 \\
 \hline
 Subfield & $\mathbb{Q}(\zeta_{128})$ & 64 & 4 & 1024 \\
 \hline
 Subfield & $\mathbb{Q}(\zeta_{243})$ & 162 & 3 & 1458 \\
 \hline
 Compositum & $\mathbb{Q}(\zeta_{192})$ & 64 & 3 & 576 \\
 \hline
 Compositum & $\mathbb{Q}(\zeta_{576})$ & 192 & 2 & 768 \\
 \hline
 Compositum & $\mathbb{Q}(\zeta_{384})$ & 128 & 3 & 1152 \\
 \hline
\end{tabular}

\label{Sample Cyclic Algebra Parameters}
\end{center}
\end{table}
\subsubsection{Two-Power Cyclotomic $K$}
We begin with straightforward cases where we can apply \cref{primepoweralgebras} immediately to obtain fields in suitable dimensions. Let $K$ be a two-power cyclotomic field, $K = \mathbb{Q}(\zeta_{2^k})$, with dimension $n := 2^{k-1}$. Since the rank $d = [L:K] = [\mathcal{A}:L]$ is a small power of two, the dimension $n$ of $K$ will be dictated by the choice of module rank $d$. We construct rank $2$ and $4$ examples as follows:
\begin{itemize}
\item For $d = 2$ we have $[\mathcal{A}:K] = 4$, so for total dimension $1024$ we set $K = \mathbb{Q}(\zeta_{512})$.
\item For $d =4$ we have $[\mathcal{A}:K] = 16$, so for total dimension $1024$ we set $K = \mathbb{Q}(\zeta_{128})$.
\end{itemize}
To obtain algebras in dimension $512$ simply pick $K$ with dimension $n/2$ e.g. $\mathbb{Q}(\zeta_{256})$ and $\mathbb{Q}(\zeta_{64})$ respectively. In all cases, \cref{primepoweralgebras} lets us pick the non-norm element $\gamma$ as a root of unity.
\subsubsection{Three-Power Cyclotomic $K$}
Since $3\nmid 1024$, one can not achieve algebras in dimension $1024$ with a $3$-power cyclotomic center and instead we set about searching for algebras of nearby dimensions. Although we are unable to build fields in this case with dimension around $1024$, we can get close to the more lightweight cryptographic dimension of $512$ used in schemes targeting a lower security level. Recall that if $K = \mathbb{Q}(\zeta_{3^k})$ then $K$ has dimension $n := \phi(3^k) = 2 \cdot 3^{k-1}$. Again, the module rank is a power of $3$ and the choice of module rank will define the choice of $n$.
\begin{itemize}
\item For $d=3$ we have $[\mathcal{A}:K] = 9$, so for total dimension $486$ we set $K = \mathbb{Q}(\zeta_{81})$. The next achievable dimension is $1458$, for which $K = \mathbb{Q}(\zeta_{243})$.
\item For $d = 9$ we have $[\mathcal{A}:K] = 81$. To achieve the same total dimensions we take small base fields $K = \mathbb{Q}(\zeta_9)$ and $\mathbb{Q}(\zeta_{27})$ respectively.
\end{itemize}

\subsubsection{Fields Using Compositum Techniques}
The algebras with prime-power cyclotomic centers of the previous subsections use the field construction technique of \cref{primepoweralgebras}, and as such they are restricted to algebras whose dimension $N$ is in the form $p^k(p-1)$ for a prime $p$ and integer $k$. In Appendix \ref{appendix:compositum}, we present another method of constructing algebras using compositum fields that allows us to target dimensions not achievable in this setting. The bottom three algebras of dimensions $576$, $768$ and $1152$ in Table \ref{Sample Cyclic Algebra Parameters} are obtained with this method.

\subsection{Extensions Where $q$ Splits Completely}\label{qsplitscompletely}
All suggested algebras in the previous section satisfy the conditions required for our chosen norm $\Vert \sigma_\mathcal{A}(x) \Vert_2$ to be well-defined. In particular, they have root of unity non-norm $\gamma$ and $K$ is cyclotomic. Because any $q = 1 \mod m$ splits completely in $\mathbb{Q}(\zeta_m)$, it is straightforward to find $q$ which splits completely in $\mathcal{O}_K$.

Later in this paper, in order to enable efficient multiplication algorithms, it will turn out that it is convenient to have a modulus $q$ that splits completely into a product of prime ideals in both $\mathcal{O}_K$ and $\mathcal{O}_L$. Recall \cref{3,homomorphism} also require $q$ be unramified in $L$. An appeal to Chebotarev's Density Theorem suggests that a proportion of $1/d$ of the primes $q$ that split completely in $K$ also do so in $L$. In cases where $d$ is small this suggests that finding such primes should not prove too arduous; but since cryptosystems require specific parameters rather than density arguments, we provide constructions satisfying the requisite conditions on $q$ in Appendix \ref{qsplits}.

\section{Security Proof}\label{security_proof}

The `standard' security reductions used in \cite{regev_lattices_2009} and \cite{lyubashevsky_ideal_2010} firstly reduce certain lattice problems to search LWE and RLWE, then establish hardness of the decision problem via a search-decision reduction. This proof follows a sequence of shorter reductions as shown in \cref{fig:Reductions}.

\begin{figure}[htb]
               \centering
   \includegraphics[trim=30mm 50mm 180mm 55mm, clip, width=0.70\linewidth]{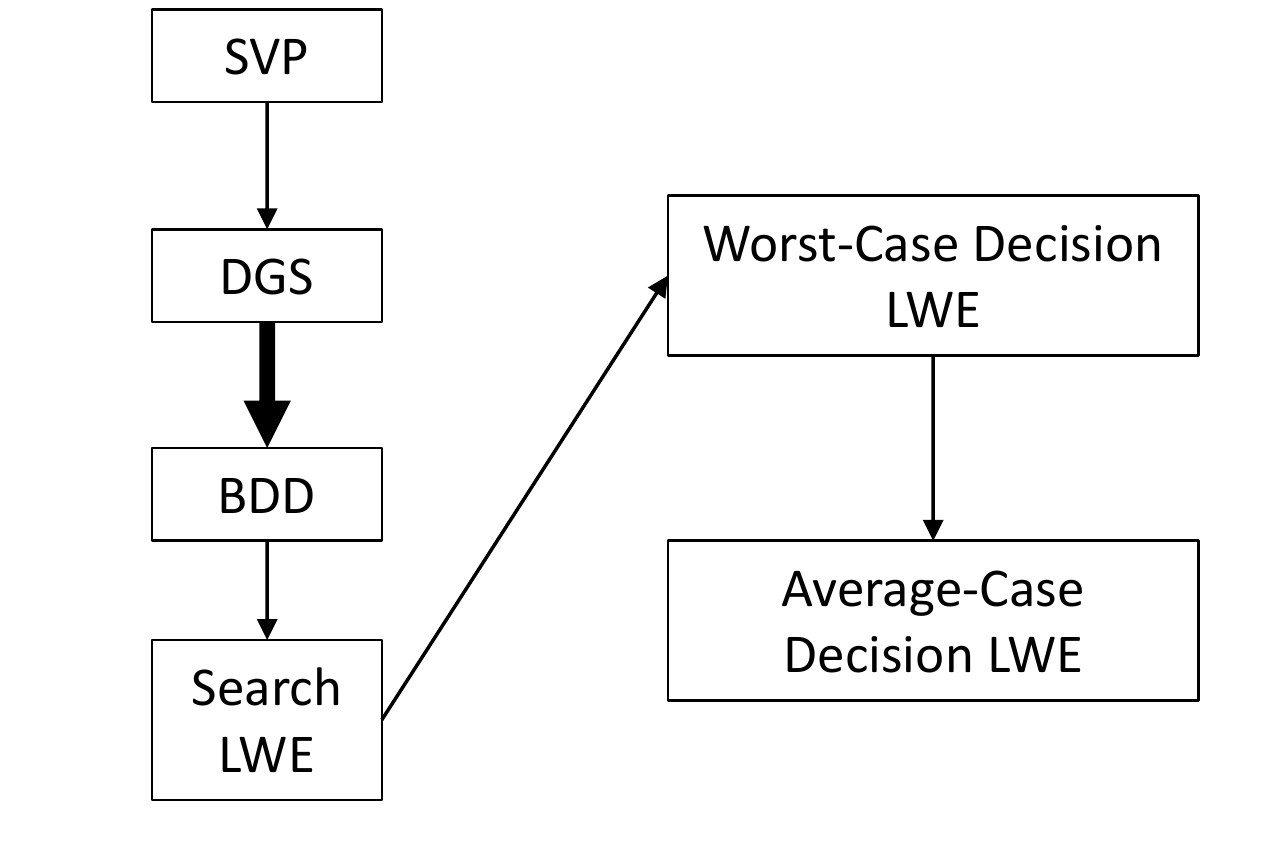}
   \caption{Reductions for LWE. The bold arrow denotes a quantum step.}\label{fig:Reductions}
\end{figure}

The reduction from the approximate SVP to the search LWE problem implies that search LWE is at least as hard as approximate SVP. It can be explained as follows: first, the approximate SVP is reduced to the problem of sampling a discrete Gaussian of narrow variance over a lattice, where intuitively sampling from a sufficiently narrow Gaussian should output a vector whose norm is reasonably short compared to the first minima. Then, a quantum algorithm reduces the problem of sampling from a narrow Gaussian to that of solving the BDD problem on the lattice. Finally, a transformation maps an instance of the BDD problem to an appropriate instance of the LWE problem, reducing the BDD problem to that of search LWE.

For applications in cryptography, the hardness of the decision problem is preferred to that of the search problem.
Assuming that the decision problem is hard implies that LWE samples are computationally indistinguishable from uniform, so intuitively an LWE sample can be used to hide a message $m$ as an element of $\mathbb{Z}_q^n$ by adding it to $b$.

Using similar machinery, we reduce a BDD problem to search CLWE using the same method as in \cite{lyubashevsky_ideal_2010}. The methodology of their search-decision reduction is an adaptation of that of Regev's, which relies on guessing each coordinate of the secret $\textbf{s}$ separately. The adaptation to the ring case instead guesses the coordinate of the secret ring element $s$ modulo a suitable collection of ideals $\mathfrak{p}_i$ such that guessing $s \mod \mathfrak{p}_i \mathcal{O}_K^\vee$ requires only a polynomial number of guesses, from which $s$ is recovered using the CRT. We apply a similar method in suitable subrings to deduce the hardness of our decision problem. The main technical novelty is to deal with non-commutativity in the proof.

For the remainder of this paper, we will always be working in an extension of number fields $L/K$, where $[L:\mathbb{Q}] = [L:K] \cdot [K: \mathbb{Q}] = d \cdot n$. Recall from the motivation of structured MLWE and the sample algebras given that in practice we seek asymptotic security in $n$, since the parameter $d$ corresponds to the typically small module dimension. 

\subsection{Hardness of Search CLWE}\label{sec4}

\begin{definition}
We define the family of error distributions $\Sigma_\alpha$ as the set of all Gaussian distributions $D_\Sigma$ over $\bigoplus_{i=0}^{d-1} u^i L_\mathbb{R}$ with covariance matrix obtained as the distribution of the error in \cref{bddtrans}.
\end{definition}
This is the family of error distributions we will claim hardness of search CLWE for; although specifying this family of matrices precisely is not simple, we demonstrate how the error is obtained in the BDD transformation step. For now, we remark that it is a Gaussian distribution whose marginals are Gaussian with variance at most $\alpha$.

In the following theorem we denote by $\mathcal{A}-$DGS$_\xi$ the problem of sampling a discrete Gaussian $D_{\mathcal{I}, \xi}$, where $\mathcal{I}$ is some ideal of the order $\Lambda$.

\begin{theorem}\label{mainresult}
Let $\mathcal{A}$ be a cyclic division algebra over a number field $L$ with center $K$ and natural, maximal order $\Lambda$ with $\vert \gamma \vert =1$. Let $\alpha = \alpha (n) \in (0,1)$ and $q = q(n) \geq 2$, unramified in $L$, be parameters such that $\alpha \cdot q \geq \omega(1)$. Then, there is a polynomial-time quantum reduction from $\mathcal{A}$-DGS$_\xi$ to search CLWE$_{q, \Sigma_\alpha}$ for any $\xi =   r \cdot \sqrt{d}\omega(\sqrt{\log{(d \cdot n)}})/ \alpha q$, where $r > \sqrt{2} q \cdot \eta_\varepsilon(\mathcal{I})$.
\end{theorem}
From this we deduce the following corollary, similarly to \cite{langlois_worst-case_2015}, since the lattice structure of our algebra is merely a special case of their modules. We denote by $N$ the total dimension of $\mathcal{A}, N := nd^2$.
\begin{corollary}

Let $\mathcal{A}, \Lambda, \alpha$ and $q$ be as above. Then, there is a polynomial-time quantum reduction from $\mathcal{A}$-SIVP$_\xi$ to search CLWE$_{q, \Sigma_\alpha}$ for any $\sqrt{8 Nd} \cdot \xi = (\omega(\sqrt{d n})/\alpha)$.
\end{corollary}

The following theorem is our analogy of Lemma 4.10 of \cite{langlois_worst-case_2015}.
\begin{theorem}\label{bootstrapping}
Given an oracle that solves CLWE$_{q, \Sigma_{\alpha}}$ for input $\alpha \in (0,1)$, an integer $q \geq 2$, an ideal $\mathcal{I} \subset \Lambda$, a number $r \geq \sqrt{2}q \cdot \eta(\mathcal{I})$ satisfying $r' := r \cdot \omega(\sqrt{\log{N}})/(\alpha q) > \sqrt{2N}/\lambda_1(\mathcal{I}^\vee)$, and polynomially many samples from the discrete Gaussian $D_{\mathcal{I},r}$ there exists an efficient quantum algorithm that outputs an independent sample from $D_{\mathcal{I}, r'}$.
\end{theorem}
We can then prove \cref{mainresult} in the standard iterative manner; for a very large value of $r$, e.g. $r \geq 2^{2N}\lambda_N(\mathcal{I})$, start by sampling classically from $D_{\mathcal{I}, r}$. Then apply the above algorithm to obtain a polynomial number of samples from $D_{\mathcal{I}, r'}$. Repeating this step gives samples from progressively narrower distributions, until we arrive at the desired Gaussian parameter $s \geq \xi$. In order to classically sample the initial collection of Gaussian samples, we use the standard Lemma 3.2 of \cite{regev_lattices_2009} to sample $D_{\mathcal{I}, r}$ on the module representation $\bigoplus_{i=0}^{d-1} u^i L_\mathbb{R}$. As usual, we obtain \cref{bootstrapping} in two steps, first the main reduction of \cref{bddtrans}, then the following quantum step adapted from \cite{regev_lattices_2009}. We use a form of $\mathcal{A}-$BDD$_{L,\delta}$ from \cite{langlois_worst-case_2015} where we bound the offset in the norm $\Vert e \Vert_{2,\infty} := \max_j\sqrt{(\sum_{i = 0}^{d-1} \vert \sigma_j(e_i) \vert^2)} \leq \delta$, where $\sigma$ denotes the canonical embedding of $L$.

\begin{lemma}
There is an efficient quantum algorithm that given any $N = n \cdot d^2$ dimensional lattice $\mathcal{L} := \sigma_\mathcal{A}(\mathcal{I})$ for some ideal $\mathcal{I}$, a real $\delta < \lambda_1(\mathcal{L}^*)/(2 \sqrt{2  nd})$, and an oracle that solves $\mathcal{A}$-BDD$_{\mathcal{L}^*,\delta}$ with all but negligible probability, outputs an independent sample from $D_{\mathcal{L}, \sqrt{d} \omega(\sqrt{\log(nd)})/\sqrt{2} \delta}$.
\end{lemma}


For the reduction of BDD to Search CLWE, we begin with the cyclic algebra analogy of the BDD-to-LWE samples transformation from Section 4 of \cite{lyubashevsky_ideal_2010}. As is standard for LWE security, we use the following `modulo $q$' definition of BDD:
\begin{definition}
For any $q \geq 2$ the $q\mathcal{A}-$BDD$_{\mathcal{I},d}$ problem is as follows: given an instance of the $\mathcal{A}-$BDD$_{\mathcal{I}, \delta}$ problem $y = x+e$ with solution $x \in \mathcal{I}$ and error $e \in \bigoplus_{i = 0}^{d-1} u^i L_{\mathbb{R}}$ satisfying $\Vert e \Vert_{2, \infty} \leq \delta$, output $x \mod q\mathcal{I}$.
\end{definition}
We use (a special case of) Lemma 3.5 from \cite{regev_lattices_2009}, which lifts immediately since it is lattice preserving.
\begin{lemma}
For any $q \geq 2$ there is a deterministic polynomial time reduction from $\mathcal{A}-$BDD$_{\mathcal{I}, d}$ to $q\mathcal{A}-$BDD$_{\mathcal{I},d}$.
\end{lemma}
We now present an algorithm which transforms $q\mathcal{A}$-BDD samples to CLWE samples given some additional Gaussian samples. The algorithm is the same in spirit as Lemma 4.7 of \cite{lyubashevsky_ideal_2010}, but has some technical differences induced by the structure of cyclic algebras.

\begin{lemma}\label{bddtrans}
Let $\mathcal{A}$ be as in Theorem \ref{mainresult}. There is a probabilistic polynomial time algorithm that on input a prime integer $q \geq 2$, a fractional ideal $\mathcal{I}^\vee \subset \Lambda$, a $q\mathcal{A}-$BDD$_{L, \alpha q \cdot \omega(\sqrt{\log(nd)})/\sqrt{2nd} \cdot  r}$ instance $y = x + e$ where $x \in \mathcal{I}^\vee$, a parameter $r \geq \sqrt{2}q \cdot \eta(\mathcal{I})$, and samples from the discrete Gaussian $D_{\mathcal{I},r'}$ with $r' \geq r$, outputs samples that are within negligible statistical distance of the CLWE distribution $ \Pi_{q, s, \Sigma}$ for a secret $s = \chi_t(x \mod q \mathcal{I}^\vee) \in \Lambda^{\vee}_q$, where $\chi_t$ is as in \cref{homomorphism} and $\Sigma$ is an error distribution such that in the case where $\vert \gamma \vert = 1$ the resulting error  $e''$ has marginal distribution in its $i,j^{\text{th}}$ coordinate that is Gaussian with parameter $r_{i,j} \leq \alpha$.
\end{lemma}
\begin{proof}
The proof will be in two parts - first, we will describe the algorithm, then we will prove correctness. Recall that in the definition of CLWE, a sample is in the form $(a, b) = (a, (a \cdot s)/q + e \mod \Lambda^\vee)$, where $e$ is taken from an error distribution $\psi \in \Sigma_\alpha$.

Begin by computing an element $t \in \mathcal{I}$ such that $\mathcal{I}^{-1} \cdot \langle t \rangle$ and $\langle q \rangle$ are coprime using \cref{3}. We can now create a sample from the CLWE distribution as follows:
take an element $z \leftarrow D_{\mathcal{I},r'}$ from the Gaussian samples, and compute a pair
\begin{align*}
(a, b) = (\xi^{-1}_t( z \mod q\mathcal{I}), (z \cdot y)/q + e' \mod \Lambda^\vee) \in (\Lambda_q \times (\bigoplus_{i = 0}^{d-1} u^i L_{\mathbb{R}})/\Lambda^\vee)
\end{align*}
where $e' \leftarrow D_{\alpha/\sqrt{2}}$.

We now claim that these samples are within negligible statistical distance of the CLWE distribution and that $s$ is uniformly random. First we show that $a \in \Lambda_q$ is statistically close to uniform. By assumption, $r \geq q \cdot \eta(\mathcal{I})$  and so by appealing to \cref{1} it can be seen that any value $z \mod q \mathcal{I}$ is obtained with probability in the interval $[\frac{1 - \varepsilon}{1 + \varepsilon}, 1] \cdot \beta$ for some positive $\beta$, from which it follows immediately that the statistical distance between $z \mod q\mathcal{I}$ and the uniform distribution is bounded above by $2\varepsilon$. Since $\chi_t$ of \cref{homomorphism} and its inverse are both bijections, we conclude that $a = \chi_t^{-1}(z \mod q\mathcal{I})$ is within statistical distance $2\varepsilon$ of the uniform distribution over $\Lambda_q$.

Now we must show that $b$ is in the form $(a \cdot s)/q + e''$, for some suitable error $e''$ and a uniformly random $s$, where we condition on some fixed value of $a$. By construction,
\begin{align*}
b :&= (z \cdot y)/q + e' \mod \Lambda^\vee \\
 &= (z \cdot x)/q + (z \cdot e)/q + e' \mod \Lambda^\vee,
\end{align*}
so since $z = t \cdot a \mod \Lambda_q^\vee$ and $t$ lies in the center of $\mathcal{A}$ it follows that $(z \cdot x)/q = (a \cdot t \cdot x)/q = (a \cdot s)/q \mod \Lambda^\vee$ for $s := \chi_t(x \mod q \mathcal{I}^\vee)$. It follows that $s$ is uniformly random over $\Lambda^\vee_q$ as long as $x$ is uniform over $\mathcal{I}^\vee$, since $\chi_t$ is a bijection.

Finally it is left to show that, conditioned on a fixed value of $a$, the marginal distribution of the $i,j^\text{th}$ coordinate of the error term $ e'' = (z \cdot e)/q + e'$ is negligibly close to that specified by $\Sigma$. We can explicitly calculate the error as
\begin{align}\label{errorresult}
e'' = \sum_{i = 0}^{d-1} u^i (\sum_{j + k  = i} \theta^k(z_j) \cdot  e_k(1-(1- \gamma) \mathbbm{1}_{j +k \geq d})) + e'
\end{align}
where the sum $j + k$ is taken modulo $d$ and the functon $(1-(1- \gamma) \mathbbm{1}_{j +k \geq d})$ is $1$ if $j+k < d$ and $\gamma$ otherwise\footnote{This term is just indicating whether or not we have had to use the relation $u^d = \gamma$ in this summand or not.}. Since $\vert \gamma \vert = 1$ and $z \leftarrow D_{\mathcal{I},r}$ is spherically distributed, it follows that multiplying by $\gamma$ and applying the permutation of $j$ coordinates induced by $\theta$ does not change the distribution of $z_{i,j}$. Hence, each marginal distribution may be analyzed independently as in the case of MLWE, and the result follows using the analysis of the error from Lemma 4.15 of \cite{langlois_worst-case_2015}.
\end{proof}
Though we do not specify the covariance of $\Sigma$, one can see that each entry of $\sigma_\mathcal{A}(z)$ appears in $\sigma_\mathcal{A}(e'')$ exactly $d$ times, and so by symmetry each element of $\sigma_\mathcal{A}(e'')$ has non-zero correlation with at most $d^2$ other entries. Hence, a proportion of at most $\frac{nd^4}{n^2d^4} = \frac{1}{n}$ of entries of $\Sigma$ are non-zero.

\subsection{Search To Decision Reduction}\label{searchdecisionsection}
In this section we will show that the hardness of decision CLWE follows from that of the search problem. Once again, we will follow a combination of the expositions of \cite{lyubashevsky_ideal_2010} and \cite{langlois_worst-case_2015} for the ring and module cases, making necessary changes for the structure of cyclic algebras. We will make heavy use of the following CRT style decomposition, a rephrasing of \cite[Lemma 4]{oggier_quotients_2012-4}.
\begin{lemma}\label{CRTlike}
Let $\Lambda$ be the natural order of a cyclic algebra $\mathcal{A} = (L/K, \theta, \gamma)$ and let $\mathcal{I}$ be an ideal of $\mathcal{O}_K$ which splits completely as $\mathcal{I} = \mathfrak{q}_1...\mathfrak{q}_n$ as an ideal of $\mathcal{O}_K$. Then, we have the isomorphism
\begin{align*}
\Lambda/\mathcal{I}\Lambda \cong \mathcal{R}_1 \times ... \times \mathcal{R}_n,
\end{align*}
where $\mathcal{R}_i = \bigoplus_{j=0}^{d-1} u^j (\mathcal{O}_L/\mathfrak{q}_i\mathcal{O}_L)$ is the ring subject to the relations $( \ell + \mathfrak{q}_i \mathcal{O}_L)u = u(\theta(\ell) + \mathfrak{q}_i \mathcal{O}_L)$ and $u^d = \gamma + \mathfrak{q}_i$.
\end{lemma}
Of course, this is not a true CRT decomposition, because we are considering ideals of $\mathcal{O}_K$ rather than those of $\Lambda$. In the case where $\gamma$ is a unit, $\Lambda^\vee = \bigoplus_i u^i \mathcal{O}_L^\vee$ and the above lemma is also valid in the case where each instance of $\mathcal{O}_L$ and $\Lambda$ are replaced with their respective duals.

As in \cite{lyubashevsky_ideal_2010}, our reduction will be limited to certain choices of algebras. The above lemma considers the splitting of the ideal $\mathcal{I}$ as an ideal of the base field $K$. Setting $\mathcal{I} = \langle q \rangle$, the ideal generated by the modulus $q$, we will consider cases where $q$ splits completely in the base field.
Now consider the family of algebras $\mathcal{A}$ in \cref{algebras} and let $K = \mathbb{Q}(\zeta_{p^a})$ have dimension $n$. It follows that if $q \equiv 1 \mod p^a$ then $q$ splits completely into a product of prime ideals $\mathfrak{q}_1,...,\mathfrak{q}_n$ as an ideal of $\mathcal{O}_K$. Hence, we obtain the decomposition
\begin{align*}
\Lambda/q\Lambda \cong R_1 \times ... \times R_n
\end{align*}
where $R_i$ is as is \cref{CRTlike}.

Also as in \cite{lyubashevsky_ideal_2010}, we see no way to avoid randomizing the error distribution in the resulting decision problem. Further, we require that an oracle for D-CLWE$_{q, \Upsilon_\alpha}$ on an algebra $\mathcal{A} = (L/K, \theta, \gamma)$ is also an oracle for the decision problem on any algebra $\mathcal{A'} = (L/K, \theta, \gamma')$ over the same number fields $L,K$ and some other root of unity $\gamma' \in \mathcal{O}_K$. Intuitively this implies that for fixed $L$ and $K$ as in \cref{algebras} the hardness of the D-CLWE problem is invariant under the choice of root of unity $\gamma$, and will be required for Lemma 15. This is because there exist efficient, easy-to-compute isomorphisms  isomorphisms sending $\mathcal{A}$ to $\mathcal{A}'$, which we will define shortly.

The main theorem of this section is \cref{searchdecision}; we emphasize that our algorithm is only intended to be efficient in the dimension $n$ of the base field $K$, since we expect to fix $d$ as a small constant in practice. We will prove \cref{searchdecision} in the usual manner: first we show that it is sufficient to recover the value of $s \in \Lambda^\vee/q\Lambda^\vee$ in one of the rings $R_i$ (\cref{crtlemma}). Then, we use a hybrid distribution to define a decision problem in $R_i$, for which we demonstrate a search to decision reduction (\cref{searchdecisionguess}). We then use a hybrid argument to conclude the proof (\cref{hybridlemma}).

\subsubsection{CLWE in $R_i$}
In this section we will abuse notation and denote by $s \mod R_i$ the value of $s \in \Lambda^\vee/q \Lambda^\vee$ in the $R_i$ coordinate under the isomorphism of \cref{CRTlike}.
\begin{definition}
The $R_i-$CLWE$_{q, \Sigma_\alpha}$ problem is to find the value $s \mod R_i$ given access to the CLWE distribution $\Pi_{q,s,\Sigma}$ for some arbitrary $\Sigma \in \Sigma_\alpha$.
\end{definition}
In the following lemmata we make use of the automorphisms of $K$ coordinatewise on the rings $R_i$. Since $K$ is a Galois extension of $\mathbb{Q}$ and $q$ splits completely, it follows that the automorphisms $\sigma_i$ of $K$ act transitively on the ideals $\mathfrak{q}_i$. We demonstrate how to extend these to functions of $\mathcal{A}$. First, extend these automorphisms to automorphisms $\alpha_i$ of $L$ in some arbitrary manner. Then, we can extend these to isomorphisms $\alpha_i : \mathcal{A} \rightarrow \mathcal{A}'$, with $\mathcal{A'} = (L/K, \theta, \gamma')$, which agree with $\alpha_i$ on $L$ and send $u$ to $u'$ with $u'^d = \alpha_i(\gamma)$ and $x u' = u' \theta(x)$ for $x \in L$. By the construction of $K$ from \cite{lahtonen_construction_2008}, $\alpha_i(\gamma)$ is a non-norm element since it is some primitive $n^\text{th}$ root of unity, and so it is easy to check that this $\mathcal{A}'$ is a well defined division algebra and that $\alpha_i$ is indeed an isomorphism which sends $\mathcal{A}$ to $\mathcal{A}'$. Furthermore, it fixes the family of error distributions $\Sigma_\alpha$. This is because each component of $z \cdot e + e'$ is defined coordinatewise over the $d$ copies of $L_\mathbb{R}$ in the module representation of $\mathcal{A}$, and since $\alpha_i$ induces the same permutation of the entries of the canonical embedding of $L$ in each coordinate as an automorphism of $L$ it fixes the family of choices for each of $z, e, e'$; hence since $\alpha_i$ is an isomorphism the family of distributions $z \cdot e + e'$ is fixed. It follows that the extended $\alpha_i$ function maps the $R_i-$CLWE$_{q, \Sigma_\alpha}$ problem in $\mathcal{A}$ to the same problem in $\mathcal{A}'$, and moreover that this map preserves $\Lambda^\vee$ and the CRT style decomposition (\cref{CRTlike}) of $\Lambda^\vee_q$ by sending $R_i$ to some $R_j$, where $j$ depends on the choice of $\sigma_i$. We are now ready for the first step of our reduction.
\begin{lemma}\label{crtlemma}
There is a deterministic polynomial time reduction from CLWE$_{q, \Sigma}$ to $R_i-$CLWE$_{q,\Sigma}$.
\end{lemma}
\begin{proof}
Let $\mathcal{O}_i$ be an oracle for the $R_i-$CLWE$_{q,\Sigma}$ problem. Since \cref{CRTlike} defines an isomorphism, it is sufficient to use $\mathcal{O}_i$ to solve the $R_j-$CLWE$_{q,\Sigma}$ for each $j$. Let $\alpha_{j/i}$ be an extension of the automorphism of $K$ mapping $\mathfrak{q}_j$ to $\mathfrak{q}_i$, which exists by transitivity. Then, given a sample $(a, b) \leftarrow \Pi_{q, s, \Sigma}$, we construct the sample $(\alpha_{j/i}(a), \alpha_{j/i}(b))$. Since $\Lambda_q$ and $\Lambda_q^\vee$ are fixed by each $\alpha_{j/i}$, the resulting pair is a valid CLWE sample in $\mathcal{A}' = (L/K, \theta, \alpha_{j/i}(\gamma))$; feeding these samples into $\mathcal{O}_i$ outputs a value $t_j \mod R_i$.

We claim $\alpha_{j/i}^{-1}(t_j) = s \mod R_j$. Since $\alpha_{j/i}$ is an automorphism, each sample $(a,b)$ is mapped to a new CLWE sample $(\alpha_{j/i}(a), \alpha_{j/i}(a \cdot s/q + e) \mod \Lambda^\vee)$ in a new algebra $\mathcal{A}'$. We may write the second coordinate as $\alpha_{j/i}(a) \cdot \alpha_{j/i}(s)/ q + \alpha_{j/i}(e) \mod \Lambda^\vee$. Since our automorphisms fix our family of error distributions and map the uniform distribution to the uniform distribution, it follows that this is a valid CLWE instance with secret $\alpha_{j/i}(s)$ and error distribution $\Sigma'$. Hence, $\mathcal{O}_i$ outputs $t = \alpha_{j/i}(s) \mod R_i$, from which we recover $\alpha_{j/i}^{-1}(t) = s \mod R_j$, as required.
\end{proof}
\subsubsection{Hybrid CLWE and Search-Decision}\label{sec:Hybrid}
For this section we must introduce the cyclic algebra analog of the Hybrid LWE distribution used in \cite{lyubashevsky_ideal_2010}; we use the decomposition into the rings $R_i$ rather than the CRT.
\begin{definition}
For a secret $s \in \Lambda_q^\vee$, distribution $\Sigma$ over $\bigoplus_{j} u^j L_\mathbb{R}$, and $i \in [n]$, we define a sample from the distribution $\Pi_{q,s,\Sigma}^i$ over $\Lambda_q \times (\bigoplus_{i = 0}^{d-1} u^i L_{\mathbb{R}})/\Lambda^\vee$ by taking $(a,b) \leftarrow \Pi_{q,s,\Sigma}$ and $h \in \Lambda^\vee_q$ which is uniformly random and independent $\mod R_j, j \leq i$ and $0 \mod R_j, j > i$, and outputting $(a,b + h/q)$. If $i=0$,we define $\Pi_{q,s,\Sigma}^0 = \Pi_{q,s,\Sigma}$.
\end{definition}
Using this distribution we define a worst-case decision problem relative to one $R_i$ and reduce it to the search problem $R_i-$CLWE.
\begin{definition}
For $i \in [n]$ and a family of distributions $\Sigma_\alpha$, the W-D-CLWE$^i_{q, \Sigma_\alpha}$ problem is defined as the problem of finding $j$ given access to $\Pi_{q,s,\Sigma}^j$ for $j \in \lbrace i-1,i \rbrace$ and valid CLWE secret and error distribution $s, \Sigma$.
\end{definition}

For a technical reason in the following proof, we restrict our secret $s$ so that $s \mod \mathcal{R}_i$ lies in a set $G_i$ with the property that $g \neq h \in G_i$ implies $g-h$ is an invertible element. Applying this restriction for each $i$ places $s \in G$ for a set $G = G_1 \times \dots \times G_n$ of size $\vert G \vert = \prod_i \vert G_i \vert$. We will call such a set $G$ a \textit{pairwise different set}. We  need to guarantee that there exist sufficiently large choices of $G$. It is not difficult to see that the maximal set sizes $\vert G_i \vert = q^d$ and $\vert G \vert = q^{nd}$, because any set of matrices in $M_{d \times d}(\mathbb{F}_q)$ of size at least $q^d +1$ contains two matrices with the same first row, whose difference is therefore uninvertible. Constructions of such maximal sets $G$ are given in Appendix \ref{appendix:secretspace}.

\begin{lemma}\label{searchdecisionguess}
Assuming $s \in G$, there is a probabilistic polynomial-time reduction from $R_i-$CLWE$_{q,s,\Sigma_\alpha}$ to W-D-CLWE$_{q,\Sigma}^i$ for any $i \in [n]$.
\end{lemma}

\begin{proof}
We follow the standard search-decision methodology of guessing the value of the secret mod $R_i$ and then modifying the samples so that the decision oracle tells us whether or not our guess was correct. Note that there are only $|G_i|$ possible values of $s \mod R_i$, which is bounded above by $q^{d^2}$, polynomial in $n$, and so we may efficiently enumerate over the possible values.

We define the transform which takes a value $g \in \Lambda^\vee_q$ and maps $\Pi_{q,s,\Sigma}$ to $\Pi_{q,s,\Sigma}^{i-1}$ if $g = s \mod R_i$ or $\Pi_{q,s,\Sigma}^i$ otherwise as follows. On input a CLWE sample $(a,b) \leftarrow \Pi_{q,s,\Sigma}$, output the pair
\begin{align*}
(a',b') = (a + v, b+ (h+vg)/q) \in \Lambda_q \times (\bigoplus_{i = 0}^{d-1} u^i L_{\mathbb{R}})/\Lambda^\vee,
\end{align*}
where $v \in \Lambda_q$ is uniformly random mod $R_i$ and $0 \mod R_j$ for $j \neq i$ and $h \in \Lambda^\vee_q$ is uniformly random and independent mod $R_j, j < i$ and $0$ on the other $R_j$. It is clear that $a'$ is still uniformly distributed on $\Lambda_q$, so we are left to show $b'$ is correctly distributed. For a fixed value of $a'$, we write
\begin{align*}
b ' &= b +(h +vg)/q \\
&= (as + h +vg)/q + e \\
&= (a's + h +v(g-s))/q + e,
\end{align*}
where $e$ is still drawn from $\Sigma$. If $g = s \mod R_i$, then $v(g-s) = 0 \mod R_i$, and so the distribution of the pair $(a',b')$ is precisely $\Pi_{q,s,\Sigma}^{i-1}$. Otherwise, $v(g-s)$ is uniformly random mod $R_i$ by assumption on $G$ and $0$ mod the other $R_j$, and so letting $h' = h +v(g-s)$ we see that the distribution of $(a',b')$ is precisely $\Pi_{q,s,\Sigma}^i$.
\end{proof}

\begin{remark}
This is the only stage of the proof which enforces that the asymptotic complexity scales only with $n$ and not with $d$, since we are forced to guess all of $s$ mod $R_i$ at once.
\end{remark}

Since the above reduction is secret preserving the required decision oracle for W-D-CLWE$_{q,\Sigma_\alpha}^i$ has the additional restriction that $s \in G$, but for the purposes of the rest of our proof it will be more convenient to have access to an oracle solving the at least as hard problem where $s$ is arbitrary. Additionally, in practical applications we will use the decision problem for arbitrary $s$, so we see no benefit of the tighter reduction where $s$ is restricted.

\subsubsection{Worst-Case to Average-Case Decision Reduction}

Now that we have removed the restriction that $s \in G$, we are able to follow the skeleton of the RLWE search-decision reduction of \cite{lyubashevsky_ideal_2010} more liberally.

\begin{definition}\label{errorfamily}
The error distribution $\Upsilon_\alpha$ on the family of possible error distributions is sampled from by choosing an error distribution $\Sigma \leftarrow \Sigma_\alpha$ and adding it to $D_\textbf{r}$, where each $r_i:= \alpha((n \cdot d^2)^{1/4} \cdot \sqrt{y_i})$ for $y_1,...,y_{n \cdot d^2}$ sampled from $\Gamma(2,1)$.
\end{definition}

\begin{definition}
For $i \in [n]$ and a distribution $\Upsilon_\alpha$ over possible error distributions, an algorithm solves the D-CLWE$_{q, \Upsilon_\alpha}^i$ problem if with a non-negligible probability over the choice pairs $(s, \Sigma) \leftarrow U(\Lambda_q^\vee) \times \Upsilon_\alpha$ it has a non-negligible difference in acceptance probability on inputs from $\Pi_{q,s, \Sigma}^i$ and $\Pi_{q,s,\Sigma}^{i-1}$.
\end{definition}
This is the average case decision problem relative to $R_i$; in our worst-case to average-case reduction we will need to randomize the choice of error distribution, which we do by sampling from $\Upsilon_\alpha$.
\begin{lemma}
For any $\alpha > 0$ and $i \in [n]$ there is a randomized polynomial-time reduction from W-D-CLWE$_{q, \Sigma_\alpha}^i$ to D-CLWE$_{q, \Upsilon_\alpha}^i$.
\end{lemma}
\begin{proof}
Since the definition of $\Upsilon_\alpha$ is a distribution over the family of distributions obtained by sampling from $\Sigma_\alpha$ and adding an elliptical Gaussian, the proof is the same as Lemma 5.12 of \cite{lyubashevsky_ideal_2010}, except we replace each instance of mod $\mathfrak{q}_i R^\vee$ with mod $R_i$ and each instance of $R_q$ with $\Lambda_q$.
\end{proof}
\begin{remark}
This choice of $\Upsilon_\alpha$ means that our decision problem is closer to diagonal than the corresponding search problem! In fact, if one increased the elliptical error in the decision problem, one could `flood out' the non-diagonal entries of the covariance matrix, leading to elliptical error which is easier to handle in practice.
\end{remark}
Finally, we use a hybrid argument. We must first show that $\Pi_{q,s,\Sigma}^n$ is uniformly random given $\Sigma$ sampled from $\Upsilon_\alpha$, but again this follows the same method as the ring case, except we must replace their use of \cref{1} by \cite[Lemma 2.4]{peikert_efficient_2010}.
\begin{lemma}\label{hybridlemma}
Let $\Upsilon_\alpha$ be as above and let $s \in \Lambda_q^\vee$. Then given an oracle $\mathcal{O}$ which solves the D-CLWE$_{q, \Upsilon_\alpha}$ problem there exists an efficient algorithm that solves D-CLWE$_{q, \Upsilon_\alpha}^i$ for some $i \in [n]$ using $\mathcal{O}$.
\end{lemma}
\begin{proof}
The proof is identical to the ring case, Lemma 5.14 of \cite{lyubashevsky_ideal_2010}, except that the indexing set $\mathbb{Z}_m^*$ is replaced by $[n]$.
\end{proof}

Denote by CLWE$_{q, \Sigma_\alpha,G}$ the search CLWE problem where $s \in G$ for arbitrary
fixed $G \subset \Lambda_q^\vee$. To sum up, we have obtained the main result of this section:

\begin{theorem}\label{searchdecision}
Let $\Lambda$ be the natural order of a cyclic algebra $\mathcal{A} = (L/K, \theta, \gamma)$, $q \in$ \text{poly}$(n)$ and assume that $\alpha \cdot q \geq \eta_\varepsilon(\Lambda^\vee)$ for a negligible $\varepsilon = \varepsilon(n)$. Then, there is a probabilistic reduction from CLWE$_{q, \Sigma_\alpha,G}$ for any pairwise different $G \subset \Lambda_q^\vee$ to D-CLWE$_{q ,\Upsilon_\alpha}$ which runs in time polynomial in $n$.
\end{theorem}

\subsection{Summary of Security Proof}

There are certain technicalities and subtleties in our security proof, which we briefly summarize as follows.

The hardness of Search CLWE in \cref{sec4} requires a natural, maximal order $\Lambda$. Nonetheless, \cref{bddtrans} (due to \cref{3,homomorphism}) is the only stage of the proof that assumes such a natural, maximal order. An improved proof technique may be able to drop this assumption (\textit{e.g.}, to use the natural order).
The search to decision reduction in \cref{searchdecisionsection} requires a natural order $\Lambda$, due to the CRT decomposition of \cref{CRTlike}. A better version of CRT may extend the reduction to a maximal order.
Fortunately, the orders we take from \cref{goodalgebras} are both natural and maximal, thereby meeting these requirements.
The requirement of unramified $q$ in \cref{mainresult} (due to \cref{3}) is minimal: for the algebras of \cref{primepoweralgebras}, the only unsuitable primes are the $p$ and $q'$ used in the construction (cf. \cref{algebras}).

\cref{searchdecisionguess} in \cref{sec:Hybrid} enforces that $s$ lies in a pairwise different set $G$. It is
the only stage of the proof which requires  such a set.
We emphasize that our reduction takes the search CLWE problem where $s \in G$ for \textit{arbitrary fixed} $G$ to the decision CLWE problem for \textit{arbitrary secret} $s$. In other words, we claim hardness for the full decision problem, based on hardness of a restricted search problem. Also, our reduction implies that the decision problem is as hard as the search problem for the hardest choice of $G$. See Appendix \ref{appendix:secretspace} for more details.

\begin{remark}
The so-called normal form is used de facto in LWE-based cryptography. We note that the normal form reduction is agnostic to the secret space $G$. More precisely, secret $s \in G$ gets completely cancelled in the transformation and replaced by a new secret $s'$ over the entire space (see of \cref{lem:normal-form} in \cref{sec:normal-form}). Therefore, the secret space in the normal form of CLWE is the entire space, after all.
\end{remark}

In practice it may be a concern with security of CLWE if these reductions were best possible (\textit{e.g.} decision CLWE is polynomial-time equivalent to restricted search, rather than at least as hard). In any case, our secret space is still exponentially large in $n$.

\section{CLWE in Cryptography}\label{crypto}

In this section we present a proof of concept cryptosystem using CLWE. To demonstrate our comparison against MLWE our scheme will closely resemble the typical `compact' LWE cryptography schemes over modules, in particular KYBER (see \cite{avanzi_kyber_2019}), although it is likely that an adaptation of Regev style encryption from \cite{regev_lattices_2009} would suit CLWE as well.

\subsection{Making CLWE Suitable For Cryptography: Normal Form}\label{sec:normal-form}
We implicitly use some standard LWE facts: firstly, we discretize our error distribution $e$ to $\Lambda_q^\vee$;  discretizing does not reduce security since an attacker may always discretize the samples themselves. Secondly, we can `tweak' the problem so that $e,s \in \Lambda_q$. Fortunately, in the case where $\gamma$ is a unit, $\Lambda^\vee = \bigoplus_i u^i \mathcal{O}_L^\vee$ and so this tweak is precisely multiplying on the right by the tweak factor taking $\mathcal{O}_L^\vee$ to $\mathcal{O}_L$ (see e.g. \cite{peikert_how_2016-1}). Finally, we require hardness of a `normal' form for the CLWE distribution, where $s$ is sampled from the same distribution as the noise $e$.

We require two facts for our proof: firstly, given that $q$ splits completely in $K$ the ring $\Lambda_q$ is isomorphic to the direct product of $n$ full matrix algebras over $M_{d \times d}(\mathbb{F}_q)$, which can be seen by appealing to the CRT-style decomposition of \cref{CRTlike} and Wedderburn's Theorem as in \cite[Propositions 1 and 4]{oggier_quotients_2012-4}. Secondly, we require that a non-negligible fraction in $n$ of elements of $\Lambda_q$ are
invertible, which follows for fixed, small, $d$ and $q \in \text{poly}(n)$ from this direct product decomposition. Otherwise, our proof follows the outline for that of plain LWE from \cite{applebaum_fast_2009}. Given these two facts, we proceed with showing that the normal form of the CLWE distribution is as hard as the case of taking the secret uniformly at random.

\begin{lemma}\label{invertible}
For a fixed $d$ and $q \geq (n+1)$, a non-negligible proportion of elements of $\Lambda_q$ are invertible.
\end{lemma}

\begin{proof}
Following the decomposition of \cref{CRTlike} and Wedderburn's Theorem, it is sufficient to show that a non-negligible proportion of elements of
\begin{align*}
    M_{d \times d}(\mathbb{F}_q) \times \dots \times  M_{d \times d}(\mathbb{F}_q)
\end{align*}
are invertible, where there are $n$ copies of $M_{d \times d}(\mathbb{F}_q)$. The proportion of invertible elements of $ M_{d \times d}(\mathbb{F}_q)$ is precisely
\begin{align*}
    &\dfrac{(q^d-1)(q^d - q)\dots(q^d - q^{d-1})}{q^{d^2}}  \\
    &= (\dfrac{q^d-1}{q^d}) \dots (\dfrac{q^d-q^{d-1}}{q^d}) \\
    &= (1- \frac{1}{q^d})\dots(1-\frac{1}{q}) \\
    &\geq (1-\frac{1}{q})^d,
\end{align*}
from which it follows that the total fraction of invertible elements in $\Lambda_q$ is at least $((1-\frac{1}{q})^d)^n$. By assumption, $q \geq n+1$, and so $(1 -\frac{1}{q})^{nd} \geq ((1- \frac{1}{n+1})^n)^d \geq ({e^{-1}})^d = e^{-d}$, as required.
\end{proof}

\begin{remark}
This lower bound of $e^{-d}$ means that the normal form reduction will be asymptotic in $n$ but only valid for fixed $d$. However, as $d$ increases the number of invertible matrices in $\Lambda_q$ is bounded above by $(1- \frac{1}{q})^{nd}$, and so the reduction would be efficient in $d$ in the case where one enforced a relation on $q$ and $d$, such as $q \geq nd + 1$, or more succinctly $q \geq N$.
\end{remark}

\begin{lemma}\label{lem:normal-form}
There is a probabilistic polynomial time reduction from the CLWE problem with uniformly random secret $s$, possibly over a limited secret space $G$, and error distribution $\chi$ to the CLWE problem with secret $s' \leftarrow \chi$.
\end{lemma}

\begin{proof}
It is sufficient to show that there is an efficient transformation taking samples with secret $s$ to samples with some new secret $s'$ taken from $\chi$. Sample pairs $(a,b) \leftarrow \Pi_{q,s,\chi}$ until a pair $(a_1, b_1:= a_1 \cdot s + e_1)$ such that $a_1$ is invertible in $\Lambda_q$ is obtained. Since a non-negligible fraction of elements of $\Lambda_q$ are invertible by \cref{invertible}, this step takes only polynomial time.

Now, given a pair $(a_i,b_i) \leftarrow \Pi_{q,s,\chi}$, we obtain a sample from the CLWE distribution $\Pi_{q,e_1,\chi}$ by outputting $(\overline{a}_i, \overline{b}_i) = (a_i a_1^{-1}, a_i a_1^{-1}b_1 - b_i)$. Since $a_1^{-1}$ is invertible, $\overline{a}_i$ is uniform. Similarly,
\begin{align*}
a_i a_1^{-1}b_1 - b_i &= (a_i a_1^{-1}(a_1 \cdot s + e_1)) - a_i \cdot s + e_i \\
&= a_i a_1^{-1}e_1 -e_i,
\end{align*}
and so $(\overline{a}_i, \overline{b}_i)$ is a valid CLWE sample with secret $e_1$ and error distribution $\chi$. Relabelling $e_1$ as $s'$ completes the proof.
\end{proof}

\subsection{Sample Cryptosystem}\label{clwecrypto}
Our scheme is parameterized by an algebra $\mathcal{A} := (L/K, \theta, \gamma)$, where $\mathcal{A}$ is as in \cref{algebras}, an error distribution $\Sigma$, and a prime modulus $q \equiv 1 \mod m$ (recall $K = \mathbb{Q}(\zeta_{m})$) which is completely split in $L$. We will denote with bold faced letters the vector form of an element of $\Lambda_q$, e.g. if $a = a_0 + u a_1 +... +u^{d-1} a_{d-1}$ then $\textbf{a} = (a_0, a_1,...,a_{d-1})$. We note that $\mathcal{O}_L/q\mathcal{O}_L$ has a polynomial representation of dimension $n \cdot d$, and so we encode our message $\in \lbrace 0,1 \rbrace^{n \cdot d^2}$ as an entry of $\Lambda_q$ as a vector $\textbf{m}$ of $d$ $\lbrace 0,1 \rbrace$ polynomials. The scheme proceeds as follows:
\begin{itemize}
\item Alice generates a CLWE sample $(a,b := a \cdot s + e)$, where $a \in \Lambda_q$ is uniformly random and $e \leftarrow \Sigma$, and outputs public key $\textbf{a}, \textbf{b}$.
\item To encrypt $\textbf{m} \in \lbrace 0,1 \rbrace^{n \cdot d^2}$, Bob samples $t, e_1, e_2 \leftarrow \Sigma$ and outputs $\textbf{u} := \phi(a)^T \textbf{t} + \textbf{e}_1, \textbf{v} := \phi(b)^T \textbf{t} + \textbf{e}_2 + \lceil \frac{q}{2} \rfloor \cdot \textbf{m}$.
\item To decrypt, Alice computes $\textbf{c} = \textbf{v} - \phi(s)^T \textbf{u}$ and recovers each coordinate of $\textbf{m}$ by rounding the corresponding entry of $\textbf{c}$ to $0$ or $\lceil \frac{q}{2} \rfloor$ and outputting $0$ or $1$ respectively.
\end{itemize}
\begin{remark}
There are two benefits of instantiating this scheme in the cyclic algebra setting rather than over modules as in \cite{avanzi_kyber_2019}, both following from the matrix embedding $\phi$. Firstly, in the module setting Alice must publish a matrix $\textbf{A}$ rather than the vector $\textbf{a}$ in her key, since $\phi(a)$ lets us generate a matrix; this saves a factor of $d$ in the size of the public key. Secondly, by extending $\textbf{b}$ to $\phi(b)$ we are able to increase the dimension of $\textbf{v}$, and correspondingly increase the size of the message by a factor of $d$.
\end{remark}

\begin{example}\label{example:prototype}

Recall our explicit algebras from \cref{algebras}. Without considering streamlined implementation for specific NIST submissions, we will pick toy comparison parameters for equivalent module based systems and ring based schemes, e.g. KYBER and NewHope. For the module case, consider a module of dimension $4$ over a ring $L$ of dimension $256$, with $2$-power cyclotomic base field $[K:\mathbb{Q}] = 64$.  Our public key $(\textbf{a}, \textbf{b})$ requires storing only $8$ elements of $R_q =O_L/q \cdot O_L$ rather than $20$ in the form $(A, \textbf{b})$. Our message consists of $1024$ bits, corresponding to the total dimension of the algebra rather than the module versions $256$ which corresponds to the field dimension; if the private key size is $256$, our CLWE scheme allows a rate-$1/4$ binary error correction code, while KYBER does not. Our ciphertext sizes are the same. As far as the modulus $q$ is concerned, we find $q=3329$ splits completely in a quartic cyclic extension $L$ of $K$. This matches with the modulus $q$ used in KYBER\footnote{The initial version of KYBER uses $q=7681$, but it has been reduced to $3329$ later which does not split completely in $L=\mathbb{Q}(\zeta_{512})$. It is noteworthy that, with a similar technique, further reduction of $q$ in CLWE may also be possible.}. Overall this represents a noteworthy gain in key and message size without loss in efficiency. For the ring case, consider an instantiation of NewHope in dimension $1024$. Both public keys are in the form $(a,s)$ and so require equivalent levels of storage ($8$ elements of a field of dimension $256$ or $2$ in dimension $1024$), and the same phenomenon is true of ciphertext sizes and message length. However, a larger modulus $q=12289$ is ued in NewHope. Hence, we hope to gain in security without losing much efficiency.
\end{example}

Before considering security and correctness we need a somewhat technical lemma allowing the use of the matrix transpose operation. Essentially, it states that if the CLWE problem is hard in an algebra $\mathcal{A}$, then for $a, s, e \in \Lambda_q$, the equation $\phi(a)^T \textbf{s} + \textbf{e}$ is a valid CLWE instance in some other algebra $\mathcal{A}'$ for which the CLWE problem is still hard.
\begin{lemma}\label{technical}
Let $\mathcal{A} = (L/K, \theta, \gamma)$ be a cyclic division algebra with matrix embedding $\phi(a)$ and natural order $\Lambda$. Then there exists another cyclic algebra $\mathcal{A}' = (L/K, \theta, \gamma^{-1})$ with matrix embedding $\phi'(a')$ and natural order $\Lambda'$ such that for $a \in \mathcal{A}$ there exists $a' \in \Lambda'$ satisfying $\phi(a)^T = \phi'(a')$. Moreover, $\mathcal{A}'$ still satisfies the division algebra condition, and $\Lambda_q'$ are $\Lambda_q$ canonically isomorphic as additive groups.
\end{lemma}
\begin{proof}
The fact that $\mathcal{A}'$ is still a division algebra follows from the non-norm property on $\gamma$ and the fact that $N_{L/K}(L^\times)$ is a multiplicative group. $\Lambda_q'$ and $\Lambda_q$ are additive isomorphic because both algebras share the same underlying fields and $\gamma, \gamma^{-1}$ are both units of $\mathcal{O}_L$. Since the first row of $\phi(a)$ is precisely $(x_0, \gamma \theta(x_{d-1}), \gamma \theta^2(x_{d-2}), \ldots,\gamma \theta^{d-1}(x_{1}))$, by setting $a' = x_0 + u \gamma \theta(x_{d-1}) +\dots + u^{d-1} \gamma \theta^{d-1}(x_1)$ and observing that $\theta^d$ is the identity it is easy to check that $\phi(a)^T = \phi'(a')$.
\end{proof}

The proofs of correctness and security are similar in spirit to those of other compact LWE schemes such as e.g. NewHope \cite{alkim_post-quantum_2016-2} or KYBER \cite{avanzi_kyber_2019}. We proceed with a somewhat informal security argument.
\begin{lemma}
The defined scheme is IND-CPA secure under the assumption that the decision CLWE$_{q,\Upsilon}$ problem is hard.
\end{lemma}
\begin{proof}
The goal of an IND-CPA adversary is to distinguish, with non-negligible advantage, between encryptions of two plaintexts $m_1, m_2$. The challenger chooses $i \in \lbrace 0,1 \rbrace$ uniformly at random and encrypts $m_i$ as $\textbf{u}, \textbf{v}$. By the assumption that the decision CLWE problem is hard, the adversary cannot distinguish between the case where $b = as +e$ and the case where it is replaced by a uniform random $b'$, so we replace the challenge ciphertext $\textbf{v}$ with $\textbf{v}'$ by replacing $b$ with $b'$. Setting $\textbf{v}'' := \textbf{v}' - \lceil \frac{q}{2} \rfloor \cdot \textbf{m}_i$, it follows by \cref{technical} that $\textbf{u}, \textbf{v}''$ represent two samples from a valid CLWE distribution with secret $\textbf{t}$, and so the adversary cannot distinguish them from uniform with non-negligible advantage. Hence, the challenger cannot distinguish $\textbf{v}'$ and hence $\textbf{v}$ from uniform with non-negligible advantage and so cannot guess $i$ with non-negligible advantage.
\end{proof}
Finally, we demonstrate conditions on the error term for the scheme to be correct.
\begin{lemma}
The defined scheme is correct as long as the $\ell_\infty$ norm of $\textbf{e}' =(\phi(e)^T \textbf{t} + \textbf{e}_2 - \phi(s)^T \textbf{e}_1)$ is less than $\lceil \frac{q}{4} \rfloor$, where the $\ell_\infty$ norm is over the vector of all polynomial coefficients of each $u^i$ entry of $\textbf{e}'$ of dimension $n \cdot d^2$.
\end{lemma}
\begin{proof}
To decrypt, Alice computes $\textbf{v} - \phi(s)^T \textbf{u}$ and computes $\textbf{m}$ by rounding. Since $\phi(\cdot)$ is a homomorphism, we have
\begin{align*}
\textbf{v} - \phi(s)^T \textbf{u} &= \phi(b)^T \textbf{t} + \textbf{e}_2 + \lceil \frac{q}{2} \rfloor \cdot \textbf{m} - \phi(s)^T( \phi(a)^T \textbf{t} + \textbf{e}_1) \\
&= \phi(e)^T \textbf{t} + \textbf{e}_2 - \phi(s)^T \textbf{e}_1 + \lceil \frac{q}{2} \rfloor \cdot \textbf{m} \\
&= \textbf{e}' + \lceil \frac{q}{2} \rfloor \cdot \textbf{m}.
\end{align*}
from which the result follows immediately.
\end{proof}
We note that the error term $\textbf{e}'$ will be unsurprising to those familiar with LWE based cryptography. Although we do not provide concrete correctness estimations, the error parameters for our decision reduction are equivalent to those of MLWE up to some small covariance terms.
We do not expect this covariance to greatly affect the distribution of the error and thus for equivalent parameter choices we expect a similarly small probability of decryption failure.

\subsection{Operational Complexity in Cyclic Algebras}\label{complexity}
In the previous subsection we showed that the CLWE problem can be used to construct a standard LWE based cryptosystem. Assuming that parameters across all variants of the LWE assumption are roughly equivalent, the CLWE problem supports key and message sizes as advantageous as those of the RLWE problem, and better than those of the module case. Along with storage considerations, another important facet of the ambient space in LWE cryptography is the efficiency of operations. Here, we will construct algorithms and consider the asymptotic complexity of multiplication in a cyclic algebra in order to compare it to the ring and module variants. Since in practice we consider operations modulo some prime $q$, addition in rings, modules, and cyclic algebras can be considered as addition in vector spaces over $\mathbb{Z}_q$, which has complexity dominated by that of multiplication.

Consequentially, we only concern ourselves with a comparison of the cost of computing the multiplication operation $\textbf{A} \textbf{s}$ in the three cases. In order to keep our comparison consistent, we let $N$ denote the total dimension of the underlying LWE instance. In the ring case, $N$ denotes the ring dimension; in the module case, $N = nd$, where $n$ denotes the ring dimension and $d$ the module rank; in the cyclic algebra case $N = nd^2$, where the ring dimension is $nd$ and the algebra has `module' rank $d$. However, since it will be important later we remark here that the cyclotomic part of the ring will be of dimension $n$ rather than $nd$. The three cases can be considered as follows:
\begin{itemize}
\item In the ring case, the operation $\textbf{A} \textbf{s}$ over $\mathbb{Z}_q$ is a representation of the ring operation $a \cdot s$ in $R_q \cong {\mathbb{Z}_q[X]}/{(X^N+1)}$. Using the CRT decomposition in dimension $N$ of \cite{lyubashevsky_toolkit_2013}, this operation is decomposed into coordinatewise multiplication in a vector of dimension $N$ over $\mathbb{Z}_q$,  following which the decomposition is reversed to recover $a \cdot s$. The complexity of this technique is dominated by that of the CRT decomposition, which takes time $O(N \log N)$, although the coordinatewise multiplication also requires time $O(N)$.
\item In the module case, $\textbf{A}$ is a $d \times d$ matrix over $R_q$. In this case, one can compute $\textbf{A} \textbf{s}$ by applying the CRT in dimension $n$ coordinatewise on $\textbf{A}$ and $\textbf{s}$. This requires $d^2+d$ applications of the CRT, for a total asymptotic complexity of $O(d^2 n\log n) = O(Nd \log (N/d))$. Again, this hides a coordinatewise multiplication step which takes time $O(Nd)$ in this setting.
\item In the cyclic algebra case, $\textbf{A}$ is a matrix in the shape $\phi(a)$, where $\phi(a)$ is the left regular representation of $a \in \Lambda_q$. We estimate the complexity of the operation $\phi(a) \cdot \textbf{s}$ in Appendix \ref{appendix:multiplication-complexity}. Explicitly, our algorithm has complexity $O(N \log (N/d^2)) + O(Nd^{\omega-2})$ in the case where $q$ splits completely in $L$, with $\omega \in [2, 2.373]$ denoting the exponent of matrix multiplication. The latter term corresponds to the cost of multiplication in our analog of the finite fields used in the CRT method for RLWE.
\end{itemize}

We see that cyclic algebras compare favourably with modules for multiplication in the same dimension $N$,
depending on the exact relationship between $\log d^2$ and $d^{\omega-2}$. Since $d$ is
likely to be fixed while $n$ scales up, we expect that the $O(N \log N)$ term will
dominate the complexity. Nonetheless, we include the second term in our results to
quantify our claims. The second term $O(Nd^{\omega-2})$ becomes $O(Nd)$ with naive matrix multiplication instead of the algorithms of \cite{caruso_fast_2017-4}, yet its overall multiplication complexity is still lower than that of module multiplication in the same dimension.

\section{Conclusions and Future Work}\label{conclusions}

The primary goal of this work is the introduction of the Learning with Errors problem over Cyclic Algebras, CLWE, adding to the family of available LWE assumptions for use in cryptography. To this end, the central pillars of an LWE problem are provided for the cyclic algebra case. First, in order to provide a foundation for the construction the notion of lattices derived from ideals of the natural order of a cyclic algebra are applied in cryptography for the first time. Then, in \cref{sec3}, the CLWE problem is formally introduced, following which explicit algebras are provided with dimensions and structure appropriate for cryptographic use. Then, in \cref{security_proof}, the usual LWE security reductions are established in the CLWE case: namely, the problem of solving short vector problems on order-ideal lattices is reduced to the search CLWE problem, and then a variant of the search CLWE problem where the secret is restricted to a fixed, well constructed subset of its usual space is reduced to the decision CLWE problem. Under plausible assumptions on this restricted search problem, combining these two reductions gives the necessary security grounding for CLWE based cryptography, which is that samples from the CLWE distribution appear pseudorandom to an onlooker with no knowledge of the secret $s$. Finally, in \cref{crypto}, the necessary steps are taken to mold the CLWE problem into a practical format for cryptography. Normal form reduction is shown and a sample cryptosystem in this form is provided. Additionally, the complexity of operations in CLWE cryptography is compared to that of RLWE and MLWE based schemes.

Cyclic algebras exhibit substantial novel structures within lattice-based cryptography, and discovering use cases for these previously unseen features represents an exciting area of future research. We outline a few directions of future research in the following.

From a theoretical standpoint, the most pressing question to be solved about CLWE is whether or not the search and decision problem are polynomial time equivalent, or instead if the hardness of the decision variant can be based directly on hard lattice problems via some other technique. In this work, the hardness of the decision problem for arbitrary secret is shown to derive from the assumed hardness of a variant of the search problem where the secret is restricted to lie in any so-called pairwise difference set $G$. Although this substantially lowers the size of the secret space, the resulting secret space is still far too large to exhaustively search. Furthermore, the decision problem is as hard as the search problem for the hardest choice of decision set $G$, precluding particularly easy cases. Nonetheless, this does not establish the formal hardness of the decision CLWE problem based on the lattice problems of \cref{latticeproblems}. The reduction fails to permit arbitrary secret since the decomposition into matrix rings of \cref{CRTlike} results in a problem that can not be `guessed' effectively, since the oracle does not necessarily accept inputs as valid when the guess is wrong.

Another method of establishing the hardness of decision RLWE that is not shown for CLWE in this work is a direct to decision reduction, which more generally represents a security proof for the decision problem that holds for wider classes of cyclic division algebras than those of \cref{searchdecisionsection}. The direct to decision reduction of \cite{peikert_pseudorandomness_2017-2} is the only security reduction for RLWE which establishes the hardness of the decision problem without enforcing that $K$ is a cyclotomic field within which $q$ splits completely, as in the search-decision reduction of \cite{lyubashevsky_ideal_2010} and the presented analog for CLWE. Dropping this restriction, and hence widening the possible choices of cyclic algebras supporting the hardness of the decision problem, would provide larger design space for CLWE based cryptography.


As for another direction of future work, we view a drawback of our work to be that we are restricted to certain instances of cyclic algebras. Although in practice most cryptography would use a fixed choice of algebra, this is a function of our methods and may be possible to remove. Additionally, showing the aforementioned direct-to-decision reduction may generalize the choice of algebras.

Finally, this work is focused on the theoretical construction of a non-commutative Ring-LWE assumption, and we leave practical analysis and implementation of cryptography based on CLWE as further research.

\section*{Acknowledgment}

The authors would like to thank Jyrki Lahtonen, Damien Stehle and Martin Albrecht for helpful discussions.
They are also grateful to Andrew Mendelsohn for finding the prime $q=3329$ in \cref{example:prototype}.

\appendices

\section{Attacking non-Division Algebras}\label{badgamma}
In \cref{discussions}, the condition that $\gamma$ is a non-norm element of $L/K$ is required in order to stop parallelizing attacks in the style of that of \cite{bootland_security_2018} applying to the CLWE problem. Thus, $\gamma$ is chosen so that $\gamma^i$ is not in the norm group of $L$ into $K$ for $i =1,2,\dots, d-1$. Here, we demonstrate that picking $\gamma$ that violates this condition leads to potentially vulnerable instances of the CLWE problem. We will need the following lemma, a rephrasing of \cite[Theorem 30.4]{reiner_maximal_1975}.
\begin{lemma}\label{reinerlemma}
Let $\mathcal{A} = (L/K, \theta, \gamma)$ be a cyclic algebra with $[L:K] = d$. Let $\gamma, \delta \in K$ be non-zero. Then
\begin{itemize}
\item $\mathcal{A} \cong \mathcal{A}_i := (L/K, \theta^i, \gamma^i)$ for each $i$ such that $(i,d) = 1$.
\item If $\gamma = 1$ then $\mathcal{A} \cong M_{d \times d}(K)$.
\item If $\delta = N_{L/K}(\beta)\gamma$ for some non-zero $\beta \in L$ then $\mathcal{A} \cong \mathcal{A}' := (L/K, \theta, \delta)$.
\end{itemize}
\end{lemma}
\begin{remark}
If $\gamma \in \mathcal{O}_K$ is a unit then all isomorphisms of this lemma hold when replacing $L$ and $K$ with $\mathcal{O}_L$ and $\mathcal{O}_K$ respectively. The first and third can be seen by examining the proofs in \cite{reiner_maximal_1975}; the first is a re-indexing of $u$ coordinates of $\mathcal{A}$, and the third simply sends $u$ to $\beta u$. The second requires a little more work. We map $\mathcal{A}$ to Hom$_K(L,L)$ by sending $u$ to $\theta$ and $x \in L$ to the $K$-homomorphism on $L$ defined by multiplying by $x$. Finally, we appeal to the standard isomorphism between Hom$_K(L,L)$ and $M_{d \times d}(K)$, which preserves integral elements as long as there exists an integral basis of $L$ over $K$. We discuss the details of this last part later, because we also require it to preserve a notion of smallness.
\end{remark}
Armed with this lemma, we demonstrate potential weaknesses of choosing $\gamma$ poorly. Let $\mathcal{A} = (L/K, \theta, \gamma)$ be a cyclic algebra where $\gamma$ lies in the norm group $N_{L/K}(L^\times)$ (and still lies in $\mathcal{O}_K)$; later we will generalize our argument to the case where instead some power of $\gamma$ less than $d$ is a norm instead. Consider the primal CLWE instance $(a, a \cdot s + e) \in \Lambda_q \times \Lambda_q$, where $a, s$ are uniform\footnote{In practice $s$ is typically sampled from the error distribution, but this normal form variant is no easier than the case where $s$ is uniform. We assume uniform $s$ here for ease of illustration.} and $e \leftarrow \chi$ is drawn from an error distribution which is of Gaussian shape. Applying \cref{reinerlemma} transforms our sample into one over $M_{d \times d}({\mathcal{O}_K}_q) \times M_{d \times d}({\mathcal{O}_K}_q)$. That is, we construct a sample in the form
\begin{align*}
(\textbf{A}, \textbf{A} \cdot \textbf{S} + \textbf{E})
\end{align*}
where $\textbf{A},\textbf{S},\textbf{E} \in M_{d \times d}({\mathcal{O}_K}_q)$. Since isomorphisms are bijections, $\textbf{A}$ and $\textbf{S}$ are uniformly random matrices. Assume for the time being that the isomorphisms are also smallness preserving, so that if $e$ is a small element of $\Lambda_q$ then the corresponding matrix $\textbf{E}$ will have entries that are small elements of $\mathcal{O}_K$.

Let $\textbf{s}_i, \textbf{e}_i$ denote the $i^\text{th}$ columns of $\textbf{S}$ and $\textbf{E}$ respectively. Then, for each $i$ the pair $(\textbf{A},\textbf{A} \textbf{s}_i + \textbf{e}_i)$ constitutes $d$ samples from the MLWE distribution in dimension $n$ and rank $d$. That is, the single CLWE sample provides a collection of $d$ samples from $d$ instances of the MLWE distribution with different secrets $\textbf{s}_1, \dots, \textbf{s}_d$, where each set of samples shares the same uniformly random matrix $\textbf{A}$. Since the difficulty of LWE problems is assumed to be superlinear in dimension $N$, solving $d$ instances of the MLWE problem in dimension $n$ and module rank $d$ is easier than solving a single instance in dimension $nd$ and rank $d$, the targeted dimension of our CLWE problem, which is essentially the parallelizing argument of the attack of \cite{bootland_security_2018} on $m$-RLWE. Furthermore, the matrix $\textbf{A}$ being common to each set of samples potentially weakens the resulting MLWE instances. Thus, assuming that $\textbf{e}_i$ is suitably distributed, it is clear that choosing a $\gamma$ that is the norm of an element of $L$ compromises security.

We are left to consider the distribution of the error matrix $\textbf{E}$. In order to understand this, we must discuss the proof of \cref{reinerlemma} further. Let $\gamma = N_{L/K}(\beta)$, so that the isomorphism mapping $\mathcal{A}$ to $\mathcal{A}' = (L/K, \theta, 1)$ fixes $L$ and sends $u$ to $u \beta$. Following the proof of \cref{primepoweralgebras} we see that the $\gamma$ which are both roots of unity and norm elements are precisely norms of some other root of unity. Hence, $\beta$ is a root of unity and this isomorphism maps a Gaussian distribution on $\mathcal{A}$ to a Gaussian distribution on $\mathcal{A}'$.

The isomorphism mapping $\mathcal{A}'$ to $M_{d \times d}(K)$ begins with a mapping from $\mathcal{A}'$ to Hom$_K(L,L)$ that sends $x \in L$ to the multiplication function $f_x(y) = xy$ for $y \in L$ and sends $u$ to $\theta$. Then, it applies the well known isomorphism sending Hom$_K(L,L)$ to $M_{d \times d}(K)$, which can be defined constructively as follows:
\begin{itemize}
\item Fix a $K$-basis $\lbrace \ell_1,\dots,\ell_d \rbrace$ of $L$ over $K$.
\item Define $f_j: K \rightarrow L$ as $f_j(x) = \ell_j x$, a mapping onto the $j$ coordinate of the basis.
\item Let $\pi_j: L \rightarrow K$ denote the projection map onto the $\ell_j$ sending $\pi_j(\sum_{i=1}^d x_i \ell_i) = \ell_j$.
\item Define $\Delta: \text{Hom}_K(L,L) \rightarrow M_{d \times d} (K)$ coordinatewise as $\Delta(\psi)_{i,j} = (\pi_i \circ \psi \circ f_j)(1)$.
\end{itemize}
Since it permits an arbitrary choice of $K$-basis, this isomorphism is non-unique. Furthermore, an attacker trying to apply this isomorphism would be able to use their choice of basis and still compute the isomorphism efficiently. We are interested in the image of a Gaussian sample $e \in \Lambda_q$ under this isomorphism, with $e = \sum_{i=0}^{d-1} u^i e_i$, having each $e_i$ sampled independently from a discrete Gaussian over ${\mathcal{O}_L}_q$, being sent to $\psi_e = \sum_{i=0}^{d-1} e_i' \sigma^i$. Correspondingly, the $i,j$ coordinate of the matrix $\textbf{E}  = \Delta(e)$ is
\begin{align*}
\pi_i (\sum_{k=0}^{d-1} e_k \theta^k(\ell_j)).
\end{align*}
For the $j^\text{th}$ column of $\textbf{E}$ (the error vector in the set of $d$ MLWE samples with secret $\textbf{s}_i$), the error is precisely the $\ell_i$ coordinate of $\sum_{k=0}^{d-1} e_k \theta^k(\ell_j)$.

Now the distribution of the error in each collection of MLWE samples depends on the properties of the chosen basis. Since the $e_k$ are independent Gaussian samples from $L$, $j$ is fixed and $\theta$ represents a permutation of the canonical embedding coordinates of $L$ elements. Hence, $\sum_{k=0}^{d-1} e_k \theta^k(\ell_j)$ is an elliptical Gaussian with $n$ blocks of $d$ different parameters. Furthermore, if $\lbrace \ell_1, \dots, \ell_d \rbrace$ is a cyclic basis then, since the distribution of $\Vert \sigma_L (\ell_k) \Vert_2$ is independent of $k$, the projection $\pi_i (\sum_{k=0}^{d-1} e_k \theta^k(\ell_j))$ follows an elliptical Gaussian. In addition, these coordinates are not independent and are potentially highly correlated.

The end result of this exposition is that, depending on the properties of the cyclic bases of $L/K$ and given the choice of $\gamma$ as a norm element, from a single CLWE instance we can construct $d$ parallel copies of $d$ MLWE instances in dimension $n$ and rank $d$ with correlated error. These correlated instances of the MLWE problem are plausibly substantially easier than the claimed security of the CLWE instance, which is that it is roughly as hard as an MLWE instance in the same dimension $nd$ and rank $d$. Of course, the error distributions in the underlying MLWE instances are non-standard and we have not presented a concrete attack on them. Instead, we believe this discussion is sufficient to persuade the unconvinced reader that solving the CLWE problem with norm element $\gamma$ can be simplified by some parallelization into MLWE instances, and thus we should stick to our specification that $\gamma$ is a norm.

In the above exposition we restricted ourselves to cases where $\gamma$ is a norm, but the definition of the non-norm condition also precludes $\gamma$ as valid if and only if $\gamma^i$ is a norm for some $i < d$ that is coprime with $d$ (see \cite{vehkalahti_densest_2009}). However, we have previously assumed that the hardness of the CLWE problem was independent of the choice of primitive $n^\text{th}$ root. In the constructions of \cref{primepoweralgebras} $\gamma$ is an $n^\text{th}$ root of unity and $d$ divides the prime power $n$, so if $i$ is coprime with $d$ then $i$ is also coprime with $n$ and so $\gamma^i$ is a primitive root which defines a cyclic algebra in which the CLWE problem can be parallelized. Thus, we conclude that $\gamma$ must satisfy the non-norm condition rather than just itself not be a norm. Independently, a recent work \cite{Revisit-MRLWE} revisiting $m$-RLWE observes that the underlying property causing the attacks of \cite{bootland_security_2018} on the original instantiations was the presence of zero-divisors in the ambient space. In our case, zero-divisors exist in a cyclic algebra if and only if the non-norm condition is not satisfied, so their argument should preclude not just the $\gamma$ that are themselves norms but also all $\gamma$ which fail the non-norm condition.

\section{Impossible Algebras}\label{impossiblealgebras}
We show that certain algebras that would otherwise be what we are looking for do not exist under our restrictions. As discussed above we would like to begin with a base field that is cyclotomic, $K = \mathbb{Q}(\zeta_m)$ for integer $m$, and proceed to fix some low degree cyclic Galois extension $L/K$ and non-norm element $\gamma \in \mathcal{O}_K$ with $\vert \gamma \vert = 1$ e.g. $\gamma$ is a root of unity. Given these restrictions and the shape of lattice cryptography, the most natural fields to look for are low degree extensions of two-power cyclotomics e.g. $m = 2^k$. Unfortunately, we are able to prove the non-existence of a large class of such extensions.
\begin{theorem}\label{dcoprime}
Let $K = \mathbb{Q}(\zeta_m)$ for some positive integer $m$ and let $p \geq 2$ be some integer which is coprime with $m$. Then, for any Galois extension $L/K$ of degree $p$ each $\zeta_m, \zeta_m^2,\dots,\zeta_m^{m-1}$ lies in $N_{L/K}(K^\times)$.
\end{theorem}
\begin{proof}
Since $L/K$ is a Galois extension of degree $p$, the relative norm map $N_{L/K}(\cdot)$ induces the map $x \rightarrow x^p$ on elements $x \in K^\times$. Let $1 \leq i \leq m-1$ be an integer; we will prove the theorem by finding $1 \leq j \leq m-1$  such that $N_{L/K}(\zeta_m^j) = \zeta_m^i$. Since $\zeta_m$ and its powers lie in $K$, the relative norm map takes $\zeta_m^j$ to $\zeta_m^{jp}$ and we are left to solve the congruence $jp \equiv i \mod m$. By assumption, g.c.d.$(m,p) = 1$ and so $p$ is invertible modulo $m$. Denoting this inverse $p^{-1}$ and letting $j = p^{-1} i \mod m$ it is easy to see that $jp \equiv i p^{-1} p \equiv i \mod m$. The theorem statement follows immediately.
\end{proof}
This theorem precludes the existence of a very large class of cyclic division algebras with cyclotomic base field. In particular, if the degree of $[L:K]$ is coprime with $m$ then we can not have our restrictions that $\vert \gamma \vert = 1$, is integral, and that $K$ is cyclotomic. We draw attention to the specific classes whose non-existence we are interested in: in an ideal world we might instantiate CLWE with $K = \mathbb{Q}(\zeta_{2^k})$ and $[L:K] = d$ for arbitrary small integer $d$ corresponding to the module rank, which in practice is likely to be at most say $5$. However, as a result of \cref{dcoprime} we know that $d$ can not be coprime with $2^k$ and must be even in order to permit a suitable $\gamma$, from which it follows that we can not have $d=3,5$.

\section{Proof of \cref{natural_is_maximal}}\label{appendix:natural-maximal}

Before proving \cref{natural_is_maximal} we need some additional  concepts and a Lemma.
Given  a $K$-central division algebra $\mathcal{A}$ and some $\mathcal{O}_K$ order $\Lambda$ in it, then
the  $\mathcal{O}_K$-\emph{discriminant}  of $\Lambda$, $d(\Lambda/\mathcal{O}_K)$, is a certain ideal in $\mathcal{O}_K$ \cite[p.126]{reiner_maximal_1975}. While $\mathcal{A}$ has many maximal orders they all share the same discriminant, which is called the discriminant of the algebra $d_\mathcal{A}$. Now the key fact about discriminants we need is that an order $\Lambda$ is maximal if and only it's discriminant equals that of $d_\mathcal{A}$.

We will now use the notation of  \cref{algebras}. According to \cite{lahtonen_construction_2008} the field $L$ and therefore also its subfield $M$ are subfields of $\mathbb{Q}(\zeta_m,\zeta_{q'})$, where $m=p^a$, and $q'\neq p$ is some large prime. Let $n=\varphi(m)=p^{a-1}(p-1)$. Furthermore it is known that $q'$ splits completely in the field $K=\mathbb{Q}(\zeta_m)$. Let us now denote with
$$
q'\mathcal{O}_K =\mathfrak{q}'_1 \cdots \mathfrak{q}'_{n},
$$
the prime ideal decomposition of $q'$ in $K$. We then have the following result.

\begin{lemma}
\label{naturaldiscriminant}
Let  $(M/K,\theta, \zeta_m)$ be an index $d$ division algebra of Theorem \ref{primepoweralgebras}
and let   $\Lambda$ be the corresponding natural order.
Then we have that
\begin{equation}\label{natural}
d(\Lambda/\mathcal{O}_K)= (\mathfrak{q}'_1, \cdots \mathfrak{q}'_{n})^{d(d-1)}.
\end{equation}
\end{lemma}

\begin{proof}
According to \cite[Lemma 5.4]{vehkalahti_densest_2009} we have that
$$
d(\Lambda/\mathcal{O}_K)=d(M/K)^d \zeta_m^{d(d-1)}=d(M/K)^d,
$$
where $d(M/K)$ is the relative number field discriminant of the extension  $M/K$. In order to find the discriminant of the natural order, it is now enough to find $d(M/K)$.
By the basic theory of cyclotomic fields we know that $\mathbb{Q}(\zeta_m,\zeta_{q'})=\mathbb{Q}(\zeta_{mq'})$. We also know   that the only ramified primes in  the extension $\mathbb{Q}(\zeta_{mq'})/\mathbb{Q}$ are $p$ and $q'$ and their ramification indices are $e_1=n$ and   $e_2=q'-1$, respectively.  Furthermore   ramification index of $p$ in the extension $\mathbb{Q}(\zeta_m)/\mathbb{Q}$ is $e_1$.  As ramification indices  are multiplicative in towers of extensions we can deduce that the only primes  that are possibly ramified in the extension $\mathbb{Q}(\zeta_{mq'})/\mathbb{Q}(\zeta_m)$ are those that lie above $q'$ in the ring  $\OO_{K}$. As $q'$ is not ramified in $\mathbb{Q}(\zeta_m)$, we get again by the multiplicativity of the ramification indices that all the  primes $\mathfrak{q'_i}$ are totally ramified in the extension $\mathbb{Q}(\zeta_{mq'})/\mathbb{Q}(\zeta_m)$. Therefore they are also totally ramified in the extension $M/\mathbb{Q}(\zeta_m)$.  Because $q'$ does not divide $d$ the prime ideals $\mathfrak{q'_i}$  are  tamely ramified. Dedekind's discriminant theorem now imply
that
\[
d(M/K)=(\mathfrak{q}'_1 \cdots \mathfrak{q}'_{n})^{(d-1)}.
\]
\end{proof}

Now we are ready to prove the natural order in \cref{natural_is_maximal} is actually maximal.

\begin{proof}
The proof is based on the result in \cite{reiner_maximal_1975} that states that an order is maximal if and only if it has the same discriminant as the discriminant of the algebra.
According to Lemma \ref{naturaldiscriminant} we have that
\begin{equation}\label{nat}
d(\Lambda/ \mathcal{O}_K) =d(M/K)^d=(\mathfrak{q}'_1 \cdots \mathfrak{q}'_{n})^{d(d-1)}.
\end{equation}
According to \cite{reiner_maximal_1975} the discriminant of the maximal order will always divide the discriminant of the natural order. Hence we know  that the only prime ideals that can possibly divide the discriminant of the maximal order are $\mathfrak{q}'_i$. Let us now assume that $\mathfrak{Q}_i$ is prime ideal above $\mathfrak{q}'_i$ in $L$. By abusing notation we will denote with $L_{\mathfrak{q}'_i}$ the  $\mathfrak{Q}_i$-adic completion of $L$ and in the same way the respective completion $M_{\mathfrak{q}'_i}$.

Following the proof of  \cite[Theorem 4]{lahtonen_construction_2008} we can see that the authors actually prove that $\zeta_m$ is a non-norm element in the extension
$L_{\mathfrak{q}'_i}/K_{\mathfrak{q}'_i}$ for each prime ideal $\mathfrak{q}'_i$. Using the same proof as in Theorem 2 we can now see that $\zeta_m$ is a non-norm element in the extensions $M_{\mathfrak{q}'_i}/K_{\mathfrak{q}'_i}$, for all $i$. According to \cite[Theorem 30.8]{reiner_maximal_1975}  $A\otimes_{K}K_{\mathfrak{q}'_i}\cong (M_{\mathfrak{q}'_i}/K_{\mathfrak{q}'_i},\theta', \zeta_m)$, where $\theta'$ naturally extends $\theta$. As $\zeta_m$ is a non-norm element, $(M_{\mathfrak{q}'_i}/K_{\mathfrak{q}'_i},\theta', \zeta_m)$  is an index $d$ division algebra. By  definition of the local index we can see that the local indices $m_{\mathfrak{q}'_i}$ are $d$ for all $\mathfrak{q}'_i$.  We now know that $\mathfrak{q}'_i$ are the only possible primes dividing the discriminant and that their local indices are $d$. 
According to \cite[Theorem 32.1]{reiner_maximal_1975}  the discriminant of the algebra $\mathcal{A}$ is
$$
 d_{\mathcal{A}}= \prod_{i=1}^{n} \mathfrak{q'}_i^{(m_{\mathfrak{q}'_i}-1)\frac{d^2}{m_{\mathfrak{q}'_i}}}=\prod_{i=1}^{n} \mathfrak{q'}_i^{(d-1)d},
$$
completing the proof.
\end{proof}

\section{Constructions Using Compositum Fields}\label{appendix:compositum}
Our other method for constructing suitable extensions starts from extensions which are nearly what we are looking for and applies field compositums (cf. \cite[Chapter 30]{reiner_maximal_1975}). We recommend this method to build on top of fields constructed using either \cref{lahtonenalgebras} or \cref{primepoweralgebras}. Say we have a Galois field extension $L'/K'$ with non-norm element $\gamma \in \mathcal{O}_{K'}$ whose Galois group is cyclic of degree $d$. Let $F$ be some other Galois number field with $F \cap  L'= \mathbb{Q}$. Then Gal$(L'F/K'F) \cong \text{Gal}(L'/K')$ and $\gamma$ is a non-norm element in $L'F/K'F$. Relabelling this extension as $L/K$ and letting $\theta$ denote the cyclic generator of the Galois group gives a cyclic field extension with non-norm $\gamma$ such that $[L:K] = d$ and $[K:\mathbb{Q}] = [K':\mathbb{Q}] \cdot [F : \mathbb{Q}]$. The relations among these fields are illustrated in \cref{fig:compositum}(a).

One can generalize this method to the case where the base field can not be written conveniently as a compositum of two fields. Let $L'/K'$ be a cyclic Galois extension of degree $d$ with non-norm element $\gamma$ and let $K$ be another Galois number field which contains $K'$. Then $KL'/K$ is a cyclic Galois extension of degree $k$ for some $k$ dividing $d$, and in particular if $K \cap L' = K'$ then $k = d$ since the fields are linearly disjoint above $K'$. See \cref{fig:compositum}(b) for the relations among these fields.

\begin{figure}[htb]
               \centering
   \includegraphics[width=0.70\linewidth]{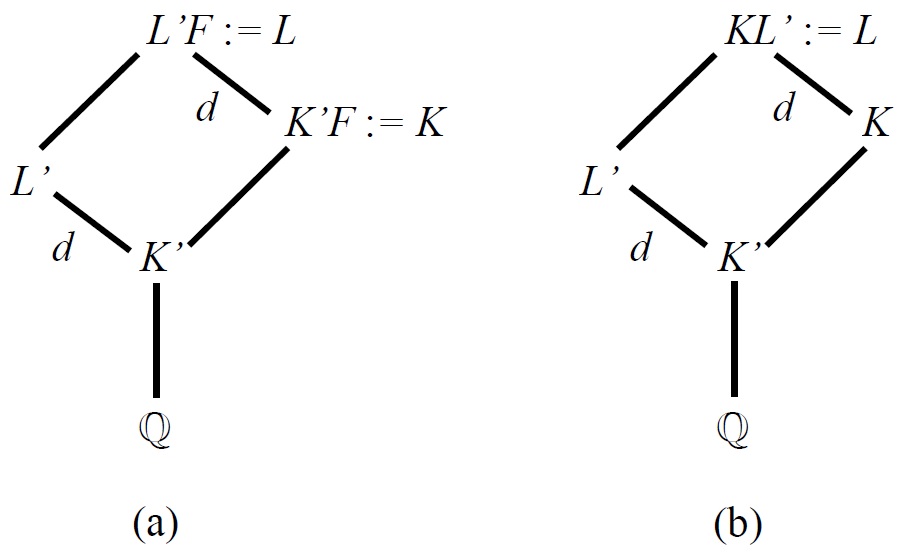}\\
   \caption{Constructions using field compositums: (a) base field $K$ is a compositum $K'F$, (b) $K$ cannot be written as a compositum.}\label{fig:compositum}
\end{figure}

We give example algebras of dimensions $576$, $768$ and $1152$ in Table \ref{Sample Cyclic Algebra Parameters} with less restrictive dimension using field compositum techniques. We propose two alternate methods of applying field compositums in  \cref{fig:compositum}(a): either use \cref{primepoweralgebras} to make an algebra which already has large dimension by selecting large center $K$ and small extension $L$, then compose a small field $F$ onto $K$ and $L$ to tweak the total dimension. Alternatively, one can create  algebras by selecting small fields $L$ and $K$ using \cref{lahtonenalgebras} and composing both with a large field $F$.

We begin with an example of the first method that achieves dimension $768$. Let $L'$ be a degree two extension of the field $K' = \mathbb{Q}(\zeta_{64})$ chosen by \cref{primepoweralgebras} with non-norm root of unity $\gamma$, so that the corresponding algebra $\mathcal{A}'$ has dimension $128$. Compose both $L'$ and $K'$ with the field $F = \mathbb{Q}(\zeta_9)$, denoting the compositums by $L$ and $K$ respectively. Then $\gamma$ is still a non-norm element in the extension $L/K$, a degree two extension that is cyclic and Galois, and the algebra $\mathcal{A}= (L/K, \theta, \gamma)$ is a cyclic algebra of dimension $6 \times 128 = 768$, as required. We observe that here the center $K$ corresponds to the fields with fast operations used in \cite{lyubashevsky_truly_2019}.

Our final method of composing large degree fields onto small degree extensions is aimed at targeting odd module ranks. Begin by choosing the desired module rank $d$ as a (likely small) odd prime. Then set $K' = \mathbb{Q}(\zeta_d)$ and pick $L'$ as a cyclic Galois extension of $K'$ in which the $d^\text{th}$ root of unity is a non-norm element using \cref{lahtonenalgebras}. Let $F := \mathbb{Q}(\zeta_{2^k})$ and again let $L$ and $K$ denote its compositum with $L'$ and $K'$ respectively. Then $\mathcal{A} = (L/K, \theta, \gamma)$ is a cyclic algebra with $n := [K:\mathbb{Q}] = (d-1)2^{k-1}$ and $d = [L:K]$ a small prime. The form of the total dimension $N = d^2(d-1)2^{k-1}$ constrains our choice of dimension, but for examples of cryptographically relevant sizes with $d=3$ one can consider setting $k =6$ or $k=7$ to achieve dimension $N = 576$ or $N =1152$ respectively. If one required additional flexibility of dimension one could also consider increasing $d$ or replacing the power-of-two cyclotomic field with any cyclotomic field whose intersection with $\mathbb{Q}(\zeta_d)$ is precisely $\mathbb{Q}$. This method comes with the subtle drawback that the module rank $d$ is also present in the dimension of the base field $K$, which precludes the case where one wants a large module rank and a small center. On the other hand, since such cases are excluded in our security proof we view this drawback as minor.

\section{Extensions Where $q$ Splits Completely in $L$}\label{qsplits}
We would like $q$ to be of roughly appropriate cryptographic size (say between $3000$ and $15000$ as a soft estimate, once again presuming parameters similar to those of NewHope or KYBER). Having $q$ split completely in $L$ is not as straightforward as in $K$ because $L$ is not a cyclotomic field, so we return to our examination of the proof of \cref{lahtonenalgebras}. Recall that in this proof the extension field $L$ is a subfield of $K(\zeta_{mq'})$ for some prime integer $q'$ satisfying $q' = 1 \mod m$ and, for $m = p^a$, $p^{a+1}$ does not divide $q'-1$. That is, $a$ is the highest power of $p$ that divides $q'-1$. We have several methods to ensure that $q$ splits completely in $L$, of which we start with the most naive.
\subsubsection{Naive Method}
For our general method we rely on the following fact: If $\mathfrak{q}_i$ is an ideal of $\mathcal{O}_K$ which splits completely in an extension $M/K$ then it splits completely in any intermediate field $M/L/K$. As it is conceptually simpler to apply this idea to the integer $q$ than to the $\mathcal{O}_K$-ideals $\mathfrak{q}_i$ we use a simpler statement, that if $\langle q \rangle$ splits completely in some $M$ containing $L$ then it splits completely in $L$. This gives us an easy way to find some $q$ that splits completely by examining a cyclotomic field that contains $L$: let $K = \mathbb{Q}(\zeta_m)$ and let $M = K(\zeta_{q'})$. Then since $q' = 1 \mod m$ it follows that $M = \mathbb{Q}(\zeta_{mq'})$. Thus $q$ splits completely in $M$ if and only if $q = 1 \mod mq'$ and consequentially splits completely in our extension $L$ if $q = 1 \mod mq'$. Since there are infinitely many primes equal to $1 \mod mq'$ this recipe always provides a prime $q$ that splits completely in $L$. The upside of this method is that it is both very general and simple, since all candidate fields $L$ we construct are contained in a larger cyclotomic field. Theoretically, this method can be extended to any abelian extension of $\mathbb{Q}$ using the partial converse of the Kronecker-Weber Theorem. However, using the Kronecker-Weber Theorem constructively is not as straightforward as picking $q'$ as in the proof of \cref{lahtonenalgebras}, so this extension to general abelian $L$ is slightly contrived.

The downside to this method is that it seems that often this will result in unrealistically large $q$. Since $q' = 1 \mod m$ and not $1 \mod p^{a+1}$, $q'$ must be chosen carefully and there are not many `small' primes satisfying these conditions. For example, in our quadratic extension case with $m = 512$ the smallest prime that is $1 \mod m$ but not $1 \mod 2m$ is $q' = 7681$. The smallest $q$ which is $1 \mod (512 \cdot 7681)$ has to be bigger than $512 \cdot 7681 = 3932672$, which is inappropriately large for lattice cryptography. Of course, one could be lucky here and have much smaller $q$ for different choices of $L$ and $K$, but in general we regard this as a theoretical result rather than a practical method. Even for smaller $2$-power cases such as $m = 128$ one must set $q' = 641$, which leads to a smallest valid prime of $q = 820481$.

Remarkably, this is much less bad in the cubic case; $K = \mathbb{Q}(\zeta_{81})$ gives $q' = 163$ as a suitable prime and $q = 26407$ still splits completely. This is perhaps slightly too large, but certainly not so much so that it is completely impractical. Nonetheless, we move on to a better method for quadratic cases.

\subsubsection{Quadratic Case}
In the case where $L/K$ ($K = \mathbb{Q}(\zeta_{512})$) is a quadratic extension we are able to choose substantially smaller $q$ by examining the unique quadratic subfields of $E' :=\mathbb{Q}(\zeta_{q'})$. We rewrite $M$ as the compositum of $E'$ and $K$, and observe that since our chosen $L$ contains $K$ our method of choosing $L$ as a subfield of $M$ allows us to write $L = EK$ for a subfield $E$ of $E'$. In the case where $L$ is a degree two extension of $K$ we know that $E$ is a quadratic field, and since $E'$ is a prime cyclotomic field we have an explicit description for its unique quadratic subfield $E$; namely that $E = \mathbb{Q}(\sqrt{q'})$ if $q' = 1 \mod 4$ and $E = \mathbb{Q}(\sqrt{-q'})$ is $q' = 3 \mod 4$. It is a standard fact that the discriminant $d_E$ of $E$ is $q'$ if $q' = 1 \mod 4$ and $-q'$ otherwise. Finally, we know that a prime $q$ splits completely in $E$ if and only if the congruence $d_E = x^2 \mod q$ has a solution e.g. if $d_E$ is a square mod $q$. Plugging in the prime numbers $q = 12289$ and $q' = 7681$ that are common in cryptography we see that $q' = 1 \mod 4$ and that $7681 = 3788^2 \mod 12289$, so that $q = 12289$ splits completely in $E,K$, and thus $L$, as required. Since this prime is explicitly the prime used in NewHope for all parameter sets we view this method as a substantial improvement on the previous technique.

\subsubsection{Quartic Fields}
Again, we use the method of describing $L$ as a compositum $MK/K$. Now, $M$ will be a quartic subfield of the field $\mathbb{Q}(\zeta_{q'})$ and one can establish the linearly disjoint nature of $M$ and $K$ required to express $L$ as this compositum by e.g. examining their discriminants: since $K$ is a power-of-two cyclotomic field the only prime appearing in its discriminant is $2$, and since $M$ is a subfield of $\mathbb{Q}(\zeta_{q'})$ the only prime in its discriminant is $q'$. Since they have coprime discriminants they are linearly disjoint, and since ramified primes are factors of the discriminant we have a relatively easy way to discount $q$ being ramified ($q \neq 2, q'$), so the remaining case to concern ourselves with is $q$ being inert.

Since the discriminants are coprime we have a method for explicitly describing the integral basis of $L = MK$; the integral basis for $K$ is clear, and an integral basis for $M$ in fixed dimension can be computed relatively easily since it has degree $4$. Then, the product of their integral bases is an integral basis for $L$. Now one only needs to check whether $q$ splits completely in $M$, since splitting in $K$ is well understood. We are unable to provide a general method for finding such $q$, but an easy computation reveals that for $q = 10753$ and $K= \mathbb{Q}(\zeta_{256})$ there is a quartic field $M$ such that $q$ splits completely in $M$ and $K$ and hence $L$. Since we have a relatively small range in which we wish to place $q$ and $M$ has low degree we do not consider the cost of this search as a large drawback since it can be done efficiently on computational software such as SAGE or PARI.

\begin{remark}
In fact, this quartic method can be applied to other instances where we do not have an explicit description of the subfields of $K(\zeta_{q'})$ which have degree $d$ over $K$: define the families of $q$ which split completely in $K$, then check whether those $q$ split completely in $L$ using computational software. Since $q = 1 \mod m$ and $m$ is relatively large, there will not be many $q$ to check of appropriate size for lattice cryptography, and so we conclude that this method is sufficient for fixed choices of fields $L,K$ for which a satisfactory $q$ exists.
\end{remark}

\section{Restricting the Secret Space}\label{appendix:secretspace}

In \cref{searchdecisionguess} we need to use a fact that is implicit in the search-decision reduction of \cite{lyubashevsky_ideal_2010}: for uniformly random $v \in \mathcal{R}_i$ and an incorrect guess $g$ of the secret $s$ modulo $\mathcal{R}_i$, the distribution of $v(g-s)$ is uniformly random. In the ring and module cases, the secret space is decomposed into a direct product of finite fields, so it is clear that $v(g-s)$ is uniformly random in each finite field for $g \neq s$.

In our case, an appeal to Wedderburn's theorem demonstrates that, since for our parameter choices each $\mathcal{R}_i$ is a central simple algebra over ${\mathcal{O}_K}^\vee/ \mathfrak{q}_i{\mathcal{O}_K}^\vee  \cong \mathbb{F}_q$,  each $\mathcal{R}_i$ is isomorphic to the full matrix ring $M_{d \times d}(\mathbb{F}_q)$, for which it is not true in general that $v(g-s)$ is uniformly random for $g \neq s$; in fact, it is uniformly random if and only if $g-s$ is invertible. Thus we restrict our secret $s$ so that $s \mod \mathcal{R}_i$ lies in a set $G_i$ with the property that $g \neq h \in G_i$ implies $g-h$ is an invertible matrix. Applying this restriction for each $i$ places $s \in G$ for a set $G = G_1 \times \dots \times G_n$ of size $\vert G \vert = \prod_i \vert G_i \vert$. Now, an incorrect guess $g \in G_i$ of $s \mod \mathcal{R}_i$ results in a distribution of $v(g-s)$ which is uniformly random mod $\mathcal{R}_i$. We will call such a set $G$ a pairwise difference set.

We also need to guarantee that there exist sufficiently large choices of $G$. A simple method for constructing a valid $G_i$ is by fixing some arbitrary embedding $\beta$ of $\mathbb{F}_{q^d}$ into $M_{n \times n}(\mathbb{F}_q)$ and letting $G_i$ equal the image of this embedding, such that $\vert G_i \vert = q^d$ and $\vert G \vert = q^{nd}$. Indeed, a $G_i$ constructed in this way is maximal because any set of matrices in $M_{d \times d}(\mathbb{F}_q)$ of size at least $q^d +1$ contains two matrices with the same first row, whose difference is therefore uninvertible.

There are a number of choices of embedding $\beta$, and thus set $G_i$, equal to the number of irreducible polynomials of degree $d$ in $\mathbb{F}_q[x]$, which can be calculated by the Necklace polynomial and in general will vastly exceed $q$. We make clear that our reduction will take the decision CLWE problem for \textit{arbitrary secret} $s$ to the search CLWE problem where $s \in G$ for \textit{arbitrary fixed} $G$, which we denote by CLWE$_{q, \Sigma_\alpha, G}$. Thus, our reduction states that the decision problem is as hard as the search problem for the hardest choice of $G$, precluding obvious attacks on the unique case where $G = {{\mathcal{O}_L}_q}^\vee$ and the CLWE problem with $s \in G$ corresponds to $d$ parallel copies in $L$ of the RLWE problem\footnote{Although this case exists only when each $\mathfrak{q}_i \mathcal{O}_L$ is a prime ideal in $\mathcal{O}_L$.}. For a general set $G$, $s \in G$ will not provide parallelization since they need not have the property of $L$ that they are entirely contained in one $u$ coordinate of $\mathcal{A}$. Additionally, even though elements of $G$ constructed this way co-commute, they do not lie in the center of $\Lambda$ and the multiplication $a \cdot s$ in the CLWE instance will not be a commutative operation.

Of course, fixing a $G$ of size $q^{nd}$ restricts the size of the secret space by a factor of $\frac{q^{nd}}{q^{nd^2}}$, a substantial loss in size even for fixed, small $d$. For concrete parameter settings, this may result in a much easier problem, but asymptotically it is still exponential in $n$ and thus establishes a suitable hardness property for decision CLWE. Of course, attacks based on exhaustive search are unlikely to represent the best attacks on the CLWE problem, so this may or may not substantially aid an attacker in practice.

In fact, there is no a priori reason why $G_i$ should be a field, or even closed under multiplication. For example, fixing a pair of invertible matrices $M_1, M_2$ and replacing $G_i$ with $M_1 \cdot G_i \cdot M_2 = \lbrace M_1 X M_2 \vert X \in G_i\rbrace$ results in a new set of size $q^d$ whose pairwise differences are all invertible but is not multiplicatively closed in general. Although the field embedding technique is perhaps the most elegant way of building $G_i$, and certainly the most constructive, it may transpire that taking $s$ from some set with less algebraic structure is advantageous in terms of the hardness of the resulting search problem. One can also construct the valid set $G_i + X$ by adding a fixed matrix $X$ to each element of $G_i$, but this technique is somewhat constrained by the fact that LWE samples are additive in the secret $s$ (e.g. one could just add $a \cdot X$ into the second coordinate of the resulting samples).

Although this restriction is not ideal, we have a remark about the implications on the security of the CLWE problem. Restricting the secret space in (R)LWE problems is not an uncommon idea: tertiary secrets, where each coordinate of $s \in \lbrace -1, 0,1 \rbrace$, are used in the NIST candidate LAC \cite{lu_lac_2018} amongst others, and security whilst restricting the secret to orders or subfields is discussed in \cite{bolboceanu_order_2018}, and to other $K$-lattices in \cite{peikert_algebraically_2019}. Overall, we suspect that the decision CLWE problem is polynomial time equivalent to the search CLWE problem without restriction on $s$, in particular when the number of samples is small as in our applications in \cref{crypto}, and that the restriction is a function of our reduction technique rather than some causal property of the CLWE distribution. For the purposes of constructing a cryptosystem, we assume that this reduction implies that the decision CLWE problem is hard.

\section{Estimating the Multiplication Complexity}\label{appendix:multiplication-complexity}

The overall flow to compute the multiplication is depicted in Fig. \ref{fig:multiplication}, which is explained in detail in the sequel.

\begin{figure}[htb]
               \centering
   \includegraphics[trim= 51mm 40mm 00mm 10mm, clip, width=1.1\linewidth]{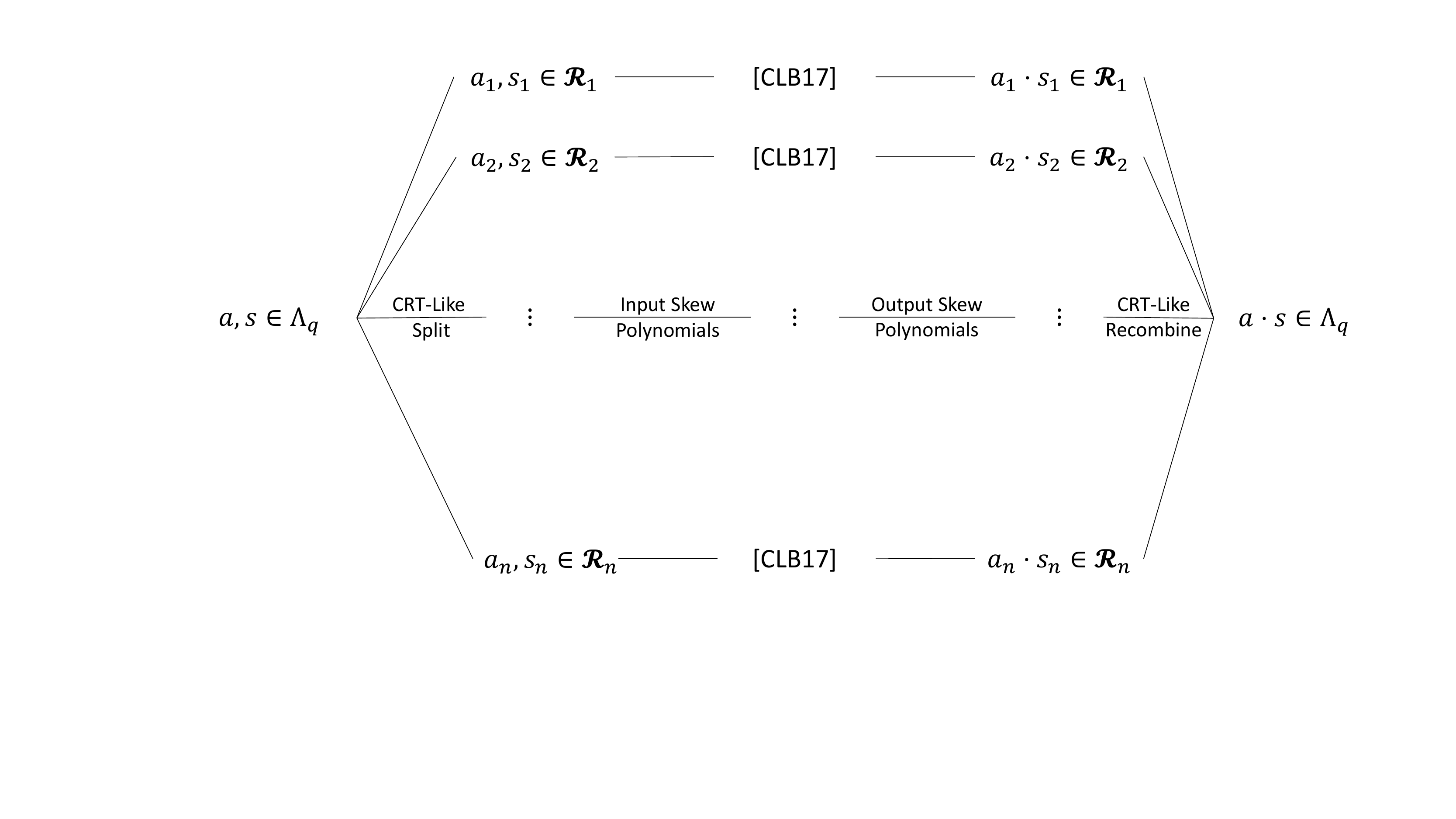}
   \caption{Depiction of the multiplication algorithm for cyclic algebras. [CLB17] is referred to as \cite{caruso_fast_2017-4}.}\label{fig:multiplication}
\end{figure}

\subsection{Algorithm for Multiplication in Cyclic Algebras}
We recall some details necessary to understand our multiplication algorithm. Recall that in the explicit constructions of \cref{primepoweralgebras} the base field $K$ is cyclotomic and $q$ is a prime integer chosen so that $\langle q \rangle$ splits completely in $\mathcal{O}_K$ as $\langle q \rangle = \mathfrak{q}_1\dots\mathfrak{q}_n$, where $n$ is the dimension of $K$ as an extension of $\mathbb{Q}$. Furthermore, the degree of $L$ over $K$ is a typically small $d$. Then, following the CRT-like decomposition of \cref{CRTlike} we write
\begin{align*}
\Lambda_q \cong \mathcal{R}_1 \times \dots \times \mathcal{R}_n
\end{align*}
for $\mathcal{R}_i = \bigoplus_{j=0}^{d-1} u^j \mathcal{O}_L/\mathfrak{q}_i \mathcal{O}_L$. We will show that each $\mathcal{R}_i$ is a skew polynomial ring over $\mathbb{Z}_q$, and in particular a skew polynomial ring for which we can apply the algorithms of \cite{caruso_fast_2017-4} to compute multiplication independently in each $\mathcal{R}_i$ in $O(d^\omega)$ operations in $\mathbb{Z}_q$, which output elements whose $u$ coordinates are in the form $\sum_i \ell_i k_i$ for $k_i \in {\mathcal{O}_K}_q$ and $\{\ell_i\}$ some arbitrary normal basis for ${\mathcal{O}_L}_q$ over ${\mathcal{O}_K}_q$. We remark that the representation as a skew polynomial ring need not contradict the fact that we viewed the rings $\mathcal{R}_i$ as matrix rings in \cref{searchdecisionsection}, since computing matrix multiplication can be reduced to the problem of computing multiplication of skew polynomials (see \cite{caruso_fast_2017-4}). Since $\omega \leq 2.373$ and we can compute the multiplication in each $\mathcal{R}_i$ in parallel, this leads to a complexity of approximately $O(Nd^{0.373})$. However, we must also compute the complexity of the splitting isomorphism.

\subsection{The Rings $\mathcal{R}_i$}
In order to apply the algorithm of \cite{caruso_fast_2017-4}, we must confirm that each $\mathcal{R}_i$ satisfies the following conditions:
\begin{itemize}
\item $\mathcal{R}_i$ is the quotient of a skew polynomial ring with center $\mathcal{O}_K/\mathfrak{q}_i$ by a polynomial in the form $X^d - \gamma$.
\item $\gamma$ is a norm from $\mathcal{O}_L/\mathfrak{q}_i \mathcal{O}_L$ into $\mathcal{O}_K/\mathfrak{q}_i$.\footnote{Due to the modulo reduction this does not contradict the assumption that $\gamma$ is not a global norm.}
\item $\mathcal{O}_L/\mathfrak{q}_i \mathcal{O}_L$ is a field extension of $\mathcal{O}_K/\mathfrak{q}_i$ or an \'{e}tale-$\mathcal{O}_K/\mathfrak{q}_i$ algebra.
\end{itemize}
The first of the conditions follows immediately from the definitions of a skew polynomial ring and a cyclic algebra. The veracity of the latter conditions will depend on how the prime ideal $\mathfrak{q}_i$ of $\mathcal{O}_K$ splits in $\mathcal{O}_L$ as $\mathfrak{q}_i \mathcal{O}_L$. Since $\mathfrak{q}_i$ is prime in $K$ and $L/K$ is Galois, we know
\begin{align*}
\mathfrak{q}_i \mathcal{O}_L = \prod_{j=1}^g (\mathfrak{q}_{i,j})^e
\end{align*}
for some prime ideals $\mathfrak{q}_{i,j}$ in $\mathcal{O}_L$ and integers $e,g$ satisfying $efg = [L:K] = d$, where $f$ denotes the inertial degree. Assuming that $L$ is constructed as a subfield of a cyclotomic field as in \cite{lahtonen_construction_2008}, it is a Galois number field and it follows that each $\mathfrak{q}_i$ splits with the same $e,f,$ and $g$. Furthermore, since they are coprime as ideals of $\mathcal{O}_K$, their factorizations' in $L$ are disjoint. Thus, we are left to consider three cases.

We first consider the case where each $\mathfrak{q}_i \mathcal{O}_L$ remains prime in $\mathcal{O}_L$. It follows that $\mathcal{O}_L/\mathfrak{q}_i \mathcal{O}_L$ is a finite field, and computing the norm of $\mathfrak{q}_i \mathcal{O}_L$ indicates $\mathcal{O}_L/\mathfrak{q}_i \mathcal{O}_L \cong \mathbb{F}_{q^d}$. In this case it is easy to see that $\mathcal{O}_L/\mathfrak{q}_i \mathcal{O}_L$ is a finite field extension of $\mathcal{O}_K/\mathfrak{q}_i \cong \mathbb{F}_q$ and consequentially, because the norm map is surjective over finite field extensions, that $\gamma$ is a norm. Here it is clear that the algorithms of \cite{caruso_fast_2017-4} can be applied.

The second case we consider is $g = d$, $e =f = 1$. Now each $\mathfrak{q}_i \mathcal{O}_L$ splits completely in $\mathcal{O}_L$ into a product of prime ideals $\mathfrak{q}_{i,1}\dots\mathfrak{q}_{i,d}$. By the CRT we have
\begin{align*}
\mathcal{O}_L/\mathfrak{q}_i \mathcal{O}_L \cong \bigotimes_{j=1}^d \mathcal{O}_L/\mathfrak{q}_{i,j}
\end{align*}
where each $\mathcal{O}_L/\mathfrak{q}_{i,j} \cong \mathbb{F}_q$, and it follows that $\mathcal{O}_L/\mathfrak{q}_i \mathcal{O}_L$ is an \'{e}tale-$\mathcal{O}_K/\mathfrak{q}_i$ algebra. We are left to show that $\gamma$ is a norm, which we show via the stronger condition that the norm map in this extension is surjective. By the CRT, $\mathcal{O}_L/\mathfrak{q}_i \mathcal{O}_L$ is isomorphic to a direct product of $d$ copies of $\mathbb{F}_q$. Since the embeddings of $L$ cyclically permute the ideal factors of $\mathfrak{q}_i$ it follows that the relative norm of an element $(x_1,\dots,x_d) \in \bigotimes_{j=0}^d \mathcal{O}_L/\mathfrak{q}_{i,j}$ is precisely $\prod_{k=1}^d x_k \mod q$. It is easy to see that this norm is surjective (because any $x \in \mathbb{F}_q$ is the norm of e.g. $(1,1,\dots,x)$) and now once again we can apply the multiplication algorithms of \cite{caruso_fast_2017-4}.

Intermediate cases, where $\mathfrak{q}_i$ splits into a product of prime ideals with the same norm such that $e=1, fg = d$, can be handled using a straightforward combination of these two methods.

The final case to consider is the ramified case, when $e \neq 1$. Now the factorization of $\mathfrak{q}_i \mathcal{O}_L$ contains some power $\mathfrak{p}_i^{e_i}$ of a prime $\mathcal{O}_L$ ideal $\mathfrak{p}_i$. In this case, we are not able to verify that the necessary conditions for the algorithms of \cite{caruso_fast_2017-4} hold. However, we observe that the ideal $\langle q \rangle$ ramifies in $\mathcal{O}_L$ if and only if $q$ divides the discriminant of $\mathcal{O}_L$. Since only a finite number of primes divide this discriminant, we restrict ourselves to considering the cases where $q$ does not ramify. We emphasize that in the main cases of interest, where $K$ is the $m^\text{th}$ cyclotomic field with $m$ having small divisors and $[L:K]$ is small, it is particularly unlikely that the large modulus $q$ typical in cryptography divides the discriminant of $L$. Indeed, when we pick $L$ as a subfield of $K(\zeta_{q'})$ for some large prime integer $s$ using the techniques of \cite{lahtonen_construction_2008} as in \cref{primepoweralgebras}, it is easy to quantify which primes potentially ramify for a fixed choice of fields: either $s$ or the primes smaller than or equal to the divisors of $m$. As an easy example, the modulus $q =12289$ does not ramify in the example algebras given in the \cref{sampleparameters} achieving dimension $1024$.
\subsection{Complexity of the CRT Style Isomorphism}
We have shown that we may apply the algorithms of \cite{caruso_fast_2017-4} to compute the multiplication operation in each $\mathcal{R}_i$ in complexity $O(d^\omega)$. We are left to consider the complexity of the isomorphism defined by \cref{CRTlike} generating the rings $\mathcal{R}_i$. Essentially, this operation is a coordinatewise split of the $u$ coordinates of $\Lambda_q = \bigoplus_{j=0}^{d-1} u^j \mathcal{O}_L$, where each entry is split into its mod $\mathfrak{q}_i \mathcal{O}_L$ parts. That is, the isomorphism maps
\begin{align*}
\sum_{j=0}^{d-1} u^j x_j \rightarrow \bigotimes_{i=1}^n \sum_{j=0}^{d-1} u^j (x_j \mod \mathfrak{q}_i \mathcal{O}_L).
\end{align*}
Splitting one element $x_i \in \mathcal{O}_K$ can be done in time $O(n \log n)$ using the CRT algorithm of \cite{lyubashevsky_toolkit_2013} when $K$ is a cyclotomic field of dimension $n$. However, $L$ is a not a cyclotomic field, but instead a small degree $d$ cyclic extension of a cyclotomic. Furthermore, we are trying to split the elements of $L$ modulo ideals of $K$ extended to those of $L$. We do not know of an existing general, efficient way of doing this. The naive estimate for an optimal method would take time $O(nd \log nd)$, where $nd$ is the dimension of $L$, but we suspect something this efficient is impossible. We have to perform $d$ such splits, which would result in a total complexity of $O(N \log N/d)$. Note that this compares relatively closely with the $O(Nd^{0.3})$ claimed for the multiplication step, and since these steps are sequential rather than parallel which of them dominates the asymptotic complexity would depend on the exact relationship between $n$ and $d$, but the result is an operational complexity essentially equivalent to that of the ring variant.

Of course, the discussion of the previous paragraph relies on our implausibly low estimate of $O(nd \log nd)$ complexity of the CRT split and so we do not claim such efficiency. Instead, we present techniques in the proceeding sections to work around the problem of splitting the $L$ part modulo the $K$ ideals in the factorization of $q$. Our methods are particularly efficient in the case where $q$ splits completely in $L$, but can be generalized to arbitrary splitting at only a small cost.

\subsection{Fast Cryptography When $q$ Splits Completely in $L$}\label{qfast}
We consider an explicit method for implementing fast cryptography in the special case where the ideal $\langle q \rangle$ splits completely in $\mathcal{O}_L$. By construction, $\langle q \rangle = \prod_i \mathfrak{q}_i$ in $\mathcal{O}_K$, so in this case we split $\langle q \rangle = \prod_{i,j} \mathfrak{q}_{i,j}$ in $\mathcal{O}_L$, where the prime $\mathcal{O}_K$-ideals have prime decomposition in $\mathcal{O}_L$ denoted $\mathfrak{q}_i \mathcal{O}_L = \prod_{j=1}^d \mathfrak{q}_{i,j}$.

We recall some facts about the extension ${\mathcal{O}_L}_q$ of ${\mathcal{O}_K}_q$. It is clear that the extension is cyclic of degree $d$, with Galois group generated by $\theta$. By the CRT,
\begin{align*}
    {\mathcal{O}_K}_q &\cong \prod_i \mathcal{O}_K/\mathfrak{q}_i
    \cong {\mathbb{F}_q}^n \\
    {\mathcal{O}_L}_q &\cong \prod_{i,j} \mathcal{O}_L/\mathfrak{q}_{i,j}
    \cong {\mathbb{F}_q}^{nd}
\end{align*}
where operations on the finite field products are applied coordinatewise. We represent the CRT decomposition of ${\mathcal{O}_L}_q$ as $({\mathbb{F}_q}^d)^n$, where each copy of ${\mathbb{F}_q}^d$ corresponds to the extension $\prod_j \mathcal{O}_L/\mathfrak{q}_{i,j}$ of $\mathcal{O}_K/\mathfrak{q}_i$. In the finite field representation of $\prod_j \mathcal{O}_L/\mathfrak{q}_{i,j}$, the elements of $\mathcal{O}_K/\mathfrak{q}_i$ embed as elements of ${\mathbb{F}_q}^d$ with the same entry in each coordinate, e.g. $(x,x,\dots, x)$, corresponding to scalars over $(\mathbb{F}_q)^d$, which can be seen from the following argument: for $k \in \mathcal{O}_K$, $k = x \mod \mathfrak{q}_i$ implies $k-x \in \mathfrak{q}_i$. Then it follows that $k -x \in \mathfrak{q}_{i,j}$ and thus $k = x \mod \mathfrak{q}_{i,j}$ for each $j$. Furthermore there is a simple, explicit, description of the action of $\theta$ in this representation: since $\theta$ cyclically shifts the ideals in the factorization of $\mathfrak{q}_i$, one can order each copy of ${\mathbb{F}_q}^d$ so that the action of $\theta$ on $({\mathbb{F}_q}^d)^n$ is a cyclical shift of the coordinates of each of the $n$ copies of ${\mathbb{F}_q}^d$ concurrently. We exhibit this with a trivial example: set $d=3, n =2$. Then the action of $\theta$ on $({\mathbb{F}_q}^3)^2$ is
\begin{align*}
    \theta(a_1,a_2,a_3,b_1,b_2,b_3) = (a_3, a_1, a_2, b_3, b_1, b_2).
\end{align*}
A valid $\mathcal{O}_K/\mathfrak{q}_i$ basis for $\mathcal{O}_L/\mathfrak{q}_i\mathcal{O}_L$ of size $d$ is $\textbf{e}_1,\dots,\textbf{e}_d$, where $\textbf{e}_i = (0,\dots, 1, \dots 0)$ denotes the $i^\text{th}$ element of the standard basis of dimension $d$. Furthermore, this basis is orthonormal in the sense that $\textbf{e}_i \cdot \textbf{e}_j = \textbf{e}_i$ for $i = j$ and $0$ otherwise and cyclic\footnote{As long as we choose the ordering in the right way.} in the sense that $\theta(\textbf{e}_i) = \textbf{e}_{i+1}$ (e.g. normal), since the Galois group $\langle \theta \rangle$ of $L$ over $K$ permutes the factors $\mathfrak{q}_{i,j}$ of $\mathfrak{q}_i\mathcal{O}_L$ for each $i$. Because the CRT splits ${\mathcal{O}_L}_q$ into a direct product within which operations are computed coordinatewise, we can extend this to a basis of ${\mathcal{O}_L}_q$ over ${\mathcal{O}_K}_q$ in the finite field representation by concatenating $n$ copies of this basis together, denoting by $\textbf{e}_i^n$ the vector of dimension $nd$ $(\textbf{e}_i, \textbf{e}_i, \dots, \textbf{e}_i)$. This basis is still cyclic, with $\theta$ operating independently on each of the $n$ copies of ${\mathbb{F}_q}^d$ and hence the $n$ copies of $\textbf{e}_i$. Concatenating the bases in this way also preserves the orthonormal property.

Denote the above basis by $\ell_1,\dots, \ell_d$. Recall that the CRT-like decomposition \cref{CRTlike} splits each $u$ coordinate, an element of ${\mathcal{O}_L}_q$, into its mod $\mathfrak{q}_i \mathcal{O}_L$ parts. However, we already know the mod $\mathfrak{q}_i \mathcal{O}_L$ parts of each $\ell_j$ by construction. So, if we store elements of ${\mathcal{O}_L}_q$ as $\ell = \sum_{j = 1}^d \ell_j k_j$ for $k_j \in {\mathcal{O}_K}_q$ we can split $\ell$ into its $\mathcal{O}_L/\mathfrak{q}_i \mathcal{O}_L$ components in time $O(d \cdot n \log n)$ as long as the $k_j$ elements are stored in the polynomial representation of ${\mathcal{O}_K}_q$. Consequentially, we can perform the CRT style decomposition of an element in $\Lambda_q$ whose $u$ coordinates are stored in this manner in time $O(d^2 \cdot n \log n) = O(N \log (N/d^2))$.

Now we see a way to achieve fast multiplication in $\Lambda_q$. We are required to perform the CRT in each of the $d$ $u$ coordinates, after which we can plug the rings $\mathcal{R}_i$ into the fast multiplication algorithm of \cite{caruso_fast_2017-4}. Since the CRT is an isomorphism and we know the image of $\ell_i$ under the CRT, this reduces to $d$ copies of the CRT in $\mathcal{O}_K$, each with complexity $O(dn \log n)$, and therefore a total multiplication complexity of $O(N \log (N/d^2)) + O(Nd^{\omega - 2})$. However, this algorithm comes with complications associated with the chosen representation of elements of ${\mathcal{O}_L}_q$, which we handle in the next section.

\subsubsection{Handling Elements in the Representation}
To use the above multiplication algorithms in the scheme of \cref{clwecrypto} we need to be able to store the elements compactly and sample the elements efficiently. Storing elements in this form turns out to be straightforward: each ${\mathcal{O}_L}_q$ element requires storing $d$ elements of ${\mathcal{O}_K}_q$. An element of $\Lambda_q$ is $d$ elements of ${\mathcal{O}_L}_q$, so in total we store $d^2$ elements of ${\mathcal{O}_K}_q$, corresponding to one element of dimension $N = nd^2$, which is equivalent to storing $d$ elements of dimension $nd$.

We now discuss how to efficiently sample elements of $\Lambda_q$ according to an appropriate error distribution. Recall from the security reduction of \cref{sec3} that the error distributions we recommend in practice are spherical or elliptical Gaussians in the coordinates of the embedding $\sigma_\mathcal{A}$. We sample using the following result.
\begin{theorem}\label{gaussianorth}
Let $L/K$ be a tower of number fields with $[K:\mathbb{Q}] = n$ and $[L:K] = d$ where $K$ is a prime-power cyclotomic field. Let $q \geq 2$ be a prime modulus which splits completely in $\mathcal{O}_L$ and let $\ell_1, \dots, \ell_d$ be the cyclic basis of ${\mathcal{O}_L}_q$ over ${\mathcal{O}_K}_q$ satisfying $\ell_i \cdot \ell_j = \ell_i$ if $i = j$ and $0$ otherwise. Then, the distribution on ${\mathcal{O}_L}_q$ obtained by sampling $k_1, \dots, k_d$ independently from a discrete Gaussian over ${\mathcal{O}_K}_q$ in the polynomial representation and outputting $\ell = \sum_i \ell_i k_i$ is a discrete Gaussian over ${\mathcal{O}_L}_q$ in the $\ell_2$ norm over $L_\mathbb{R}$.
\end{theorem}
\begin{proof}
Recall that in the case where $K$ is a prime power cyclotomic the power basis is a rotation and a scaling of the canonical basis (see e.g. \cite{crockett_challenges_2016}), so a discrete Gaussian in the polynomial representation corresponds to a discrete Gaussian in the canonical basis as well.  Order the canonical embedding of $\mathcal{O}_L$ such that elements of $\mathcal{O}_K$ embed as vectors of $n$ blocks of length $d$ that are the same in each block, e.g.
\begin{align*}
k_1 = (k_{1,1}, k_{1,1} \dots, k_{1,1}, k_{1,2}, \dots, k_{1,n}),
\end{align*}
where each entry $k_{i,j}$ of $k_i$ appears $d$ times. Since the $\ell_i$ form a cyclic basis, in each $d$-block the entries of $\ell_{i+1}$ are just a cyclic shift of those of $\ell_i$ \footnote{Again assuming a sensible ordering.}. For a fixed choice of basis the distribution in each $d$-block of $\ell$ is independent, because the $k_{i,j}$ are sampled independently from a spherical Gaussian. So we can consider one $d$ block of $\ell$ at a time, and write the $d$-block of $\ell_1$ as $a_1, \dots , a_d$. Since multiplication in the canonical embedding is coordinatewise and the $\ell_i$ form a cyclic basis, the first block of $\ell$ can be written as
\begin{align*}
\begin{pmatrix}
a_1 & a_2 & \dots & a_d \\
a_d & a_1 & \dots & a_{d-1} \\
\vdots & \vdots & \ddots & \vdots \\
a_{2} & a_{3} & \dots & a_1
\end{pmatrix} \cdot
\begin{pmatrix}
k_{1,1} \\
k_{2,1} \\
\vdots \\
k_{d,1}
\end{pmatrix}.
\end{align*}
Call the left matrix $\textbf{A}$ and the right vector $\textbf{k}$. $\textbf{k}$ is a Gaussian of parameter $r$, so $\textbf{A} \textbf{k}$ has has a Gaussian distribution with covariance matrix $r \cdot \textbf{A}\textbf{A}^{\dagger}$ by e.g. \cite[Lemma 2.5]{luzzi_almost_2018}, and if this is diagonal and constant on the lead diagonal then we are done. Due to the structure of the canonical embedding and how we picked our basis in the $\mathcal{O}_L/\langle q \rangle$ representation, we have that $a_i = \theta^i(a_1)$, and that for $i \neq j$ $\theta^i(a_1) \cdot \theta^j(a_1) = 0 \mod q$. It follows that the off-diagonal entries of $\textbf{A}\textbf{A}^{\dagger}$ are $0$ (since product being $0$ is preserved under representations) and the diagonal entries are $\sum_{i=1}^{d} |a_i|^2$, where $\vert \cdot \vert$ denotes the absolute value. Hence, the first $d$-block of $\ell$ is a spherical Gaussian distribution, and since this analysis holds for any block it follows that each block of $\ell$ is a spherical Gaussian. One also needs to show that the Gaussian distribution has the same variance in each block, but this follows from the fact that the $K$-embeddings permute the mod $\mathfrak{q}_i$ values and fix the $\ell_2$ norm of $K_\mathbb{R}$. Explicitly, by construction  each $K$ embedding modulo $\langle q \rangle$ can be extended `identically' onto $\mathcal{O}_L \mod \langle q \rangle$ in a way that fixes each $\ell_i$, so they must have the same set of values in each block (this would not be the case if we considered their norm in a global sense, and the restriction modulo $q$ is strictly necessary).
\end{proof}
Note that the statement does not define the resulting parameter of the Gaussian outputting $\ell$, but the proof allows one to compute this: say each $k_i$ was chosen from a discrete Gaussian of parameter $r$. Then each element of $\ell$ has parameter $\sqrt{\sum_i \vert a_i\vert^2} \cdot r$. Computing $\sqrt{\sum_i \vert a_i\vert^2}$ is a one time cost for a fixed choice of $\ell_1,\dots,\ell_d$, so one can sample the required Gaussian over ${\mathcal{O}_L}_q$ of parameter $r'$ by sampling from the discrete Gaussian over ${\mathcal{O}_K}_q$ of parameter $r = r'/\sqrt{\sum_i \vert a_i\vert^2}$.

Finally, to sample elements of $\Lambda_q$ we merely sample each $u$ coordinate independently according to the above technique. If we wanted to use this method in the cryptosystem of \cref{clwecrypto} to attain efficient operations then we would sample and store all elements using this representation over the cyclic basis $\ell_1,\dots \ell_d$.

Unfortunately, we are unable to generalize this theorem to the case where $\mathfrak{q}_i$ remains prime, or even intermediate cases. In this case, there exist cyclic bases of $\mathcal{O}_L/\mathfrak{q}_i \mathcal{O}_L$ over $\mathcal{O}_K/\mathfrak{q}_i$, but since $\mathcal{O}_L/\mathfrak{q}_i \mathcal{O}_L$ is a finite field and thus has no zero-divisors the cyclic bases are not orthogonal. Consequentially, the matrix $\textbf{A}$ does not in general give a diagonal $\textbf{A} \textbf{A}^T$ and thus the distribution of $\textbf{A}\textbf{k}$ has several potentially large covariance terms. If one were able to tolerate the covariance, the method can be extended in this case. It is also possible that a cyclic basis satisfying the condition that $\textbf{A}\textbf{A}^T$ is diagonal may exist for certain choices of field, but we were not able to find such a family of fields. We note that this question can be asked as a more generic question about finite fields: let $F = \mathbb{F}_{q^d}$ be a finite field with $d > 1$ and let $\theta$ denote the Frobenius automorphism of $F$. Does there exist a cyclic basis $b_1, \dots, b_d$ with $b_j = \theta^j (b_1)$ for $F$ over $\mathbb{F}_q$ satisfying
\begin{align*}
    \sum_{i=0}^{d-1} \theta^i(b_1 \cdot \theta^{j-k}(b_1)) = 0
\end{align*}
for all $j \neq k$ less than $d$? Here $j$ and $k$ correspond to $j,k^\text{th}$ entry of $\textbf{A} \textbf{A}^T$. We were unable to come up with a basis satisfying this condition, but neither can we show that no such basis exists.
\begin{example}
We exhibit an example of the basis $\ell_1, \ell_2$ in the simplest setting, that of a degree $2$ extension of $\mathbb{Q}$. Let $L = \mathbb{Q}(i)$, with ring of integers $\mathcal{O}_L = \mathbb{Z}[i]$, and consider the ideal $\langle 5 \rangle$ of $\mathcal{O}_L$. $5$ factorizes in $\mathcal{O}_L$ as $5 = (2+i)(2-i)$, and it is clear that $\langle 5 \rangle = \langle 2+i \rangle \cdot \langle 2-i \rangle$ is a decomposition into a product of prime ideals.

Using the notation $\mathfrak{q}_1 := \langle 2+i \rangle, \mathfrak{q}_2 := \langle 2-i \rangle$, it is easy to check that $2+i = -1 \mod \mathfrak{q}_2$ and thus $-(2+i) = -2-i$ is a valid choice for $\ell_1$. Similarly, $-(2-i) = -2+i$ is an appropriate choice for $\ell_2$. Correspondingly, the distribution obtained by sampling $k_1, k_2 \leftarrow D_r$, the discrete Gaussian of parameter $r$ over $\mathbb{Z}_5$, and outputting $k_1 \cdot (-2 +i) + k_2 \cdot (-2 -i)$ is a discrete Gaussian over $\mathcal{O}_L \mod \langle 5 \rangle$. Furthermore, to multiply two elements $k = k_1 \ell_1 + k_2 \ell_2$ and $g = g_1 \ell_1 + g_2 \ell_2$ modulo $5$ one outputs $kg = (k_1 g_1 \mod 5) \cdot \ell_1 + (k_2 g_2 \mod 5) \cdot \ell_2$, at a cost of two operations in $\mathbb{Z}_5$, and performing the $\mathcal{O}_L \mod 5$ CRT on each $u$ coordinate of an element of the resulting natural order $\Lambda_5$ can be done by merely reading off the $d^2 = 4$ values of $k_i$ and no additional computation.

Furthermore, this is an example where the techniques of our next section may be advantageous. We will generalize the multiplication and CRT technique so that one is free to use any basis of $\mathcal{O}_L$ over $\mathbb{Z}$, for example the basis $\lbrace 1, i \rbrace$. In this basis it is particularly easy to sample a discrete Gaussian in the polynomial representation of $\mathcal{O}_L \mod \langle 5 \rangle \cong \frac{\mathbb{Z}_5[x]}{x^2+1}$, but the resulting multiplication operation and CRT decomposition is not coordinatewise in the basis and so a small amount of efficiency is lost at a gain in parameter of the Gaussian. Specifically, to compute the CRT on an element $k = k_1 + k_2 \cdot i$, one has to precompute\footnote{Note that precomputing the image of $1$ is trivial.} the values $i = -2 \mod \mathfrak{q}_1, i = 2 \mod \mathfrak{q}_2$ and output
\begin{align*}
    (k_1 - 2 k_2 \mod \mathfrak{q}_1, 2 k_2 \mod \mathfrak{q}_2),
\end{align*}
which requires additional operations over $\mathbb{Z}_5$.
\end{example}

\subsection{Generalizing to non-Split $q$ and Arbitrary Bases}
In order to construct the cyclic, orthonormal, basis of \cref{gaussianorth}, the previous section requires that $q$ be completely split in both $K$ and $L$. However, it is possible to drop the splitting condition in $L$ and obtain fast multiplication algorithms in the general case at only a small loss of efficiency. We demonstrate the technique in this section and then briefly describe cases where a general algorithm may be superior to the one requiring that $q$ splits by discussing alternatives to \cref{gaussianorth}.

Observe that, regardless of the prime ideal decomposition of each $\mathfrak{q}_i\mathcal{O}_L$, under the CRT decomposition the quotient ring $\mathcal{O}_L/\mathfrak{q}_i \mathcal{O}_L$ is a vector space of dimension $d$ over $\mathbb{F}_q \cong {\mathcal{O}_K}/\mathfrak{q}_i$. Consequentially, an arbitrary ${\mathcal{O}_K}_q$ basis $\ell_1, \dots, \ell_d$ of ${\mathcal{O}_L}_q$ can be decomposed into $n$ bases $\ell_j = (\ell_{1,j}, \dots, \ell_{n,j})$ so that each collection $\ell_{i,1},\dots,\ell_{i,d}$ of $\mathfrak{q}_i \mathcal{O}_L$ parts is a vector space basis of dimension $d$ over $\mathcal{O}_K/\mathfrak{q}_i$. Indeed, in the split case we constructed each $\ell_i$ in this manner. Armed with this knowledge, we adapt the multiplication algorithm as follows.

Choose an arbitrary integral $\mathcal{O}_K$-basis $\ell_1,\dots, \ell_d$ of $\mathcal{O}_L$. As a precomputation phase, compute and store the images $\ell_j \mod \mathfrak{q}_i \mathcal{O}_L$ for each $i$ and $j$. The CRT-like decomposition of \cref{CRTlike} splits each of the $u$ coordinates of an element of $\Lambda_q$, an element of ${\mathcal{O}_L}_q$, into its mod $\mathfrak{q}_i \mathcal{O}_L$ parts. Once again, we suggest an algorithm where elements of ${\mathcal{O}_L}_q$ are stored in the form $\ell = \sum_{j = 1}^d \ell_j k_j$ for $k_j \in {\mathcal{O}_K}_q$, e.g. on elements stored as $K$-combinations of this basis. We split ${\ell \in \mathcal{O}_L}_q$ into its $\mathcal{O}_L/\mathfrak{q}_i$ components in time $O(d \cdot n \log n)$, since
\begin{align*}
\sum_{j=1}^d \ell_j k_j \mod \mathfrak{q}_i \mathcal{O}_L= \sum_{j=1}^d (\ell_j \mod \mathfrak{q}_i\mathcal{O}_L) \cdot (k_j \mod \mathfrak{q}_i \mathcal{O}_L),
\end{align*}
where each $k_j \mod \mathfrak{q}_i$ can be computed in time $O(n \log n)$ by the $K$-CRT and each $\ell_j \mod \mathfrak{q}_i \mod \mathcal{O}_L$ was computed in the precomputation phase. Consequentially, we can perform the CRT style decomposition of an element in $\Lambda_q$ whose $u$ coordinates are all stored in this manner in time $O(d^2 \cdot n \log n)$, since we must split $d^2$ elements of $\mathcal{O}_K$. This decomposing complexity is the same as in the previous case where $q$ splits completely. Following this, each ring $\mathcal{R}_i$ can be plugged in to the algorithm of \cite{caruso_fast_2017-4} to compute the multiplication in time $O(Nd^{\omega-2})$. However, since the $\ell_i$ do not correspond to a standard orthonormal basis we incur an extra cost when reversing this transformation. Namely, each of the $u$ coordinates of each ring $\mathcal{R}_i$ is output by the algorithm of \cite{caruso_fast_2017-4} as an element $\ell \in \mathcal{O}_L \mod \mathfrak{q}_i \mathcal{O}_L$ expressed in an arbitrary normal basis. Before reversing the decomposition we must allow for the complexity of expressing each element of the output in the bases obtained by the images of $\ell_1, \dots, \ell_d \mod \mathfrak{q}_i \mathcal{O}_L$, as this basis was not necessarily normal. Since $\mathcal{O}_L \mod \mathfrak{q}_i \mathcal{O}_L$ is a vector space of dimension $d$ over $\mathbb{F}_q$ this can be done via a precomputed change of basis matrix over $\mathbb{F}_q$ in time $O(d^\omega)$, and since there are $n$ rings with $d$ coordinates each the complexity of computing this on every coordinate is $O(nd^{\omega+1})$. The resulting multiplication algorithm has total complexity $O(N \log (N/d^2)) + O(Nd^{\omega-1})$. While this represents only a minor asymptotic loss, especially since we expect the first term to dominate the complexity, it is likely in practice that the extra step required to recover the basis representation would cause a tangible slowdown.

An unfortunate issue with this technique is that by replacing the orthonormal basis with an arbitrary basis we have lost \cref{gaussianorth} and thus the efficient method for sampling a discrete Gaussian in the representation $\ell = \sum_j \ell_j k_j$. However, this generalization allows for the use of an arbitrary basis $\ell_1,\dots, \ell_d$, unlike in the split case in which we chose a specific basis. Since we require that elements of $\Lambda_q$ are input into the algorithm with $u$ coordinates in the form $\sum_j \ell_j k_j$ this algorithm can be combined with the cryptosystem of \cref{clwecrypto} in the case where there is a basis $g_1, \dots, g_d$ of ${\mathcal{O}_L}_q$ over ${\mathcal{O}_K}_q$ in which one can compute the representation $\ell = \sum_j g_j k_j$ particularly efficiently. This is because one can just sample $\ell$ from the usual Gaussian distribution over the polynomial basis of ${\mathcal{O}_L}_q$, compute its representation as $\ell = \sum_j g_j k_j$, and then apply the multiplication algorithm in this form. More generally, the flexible choice of basis allows for both non-split $q$ and for a user to choose their favourite $\mathcal{O}_L$ basis properties, such as a normal basis or a basis consisting of small elements. We remark that it is likely possible to construct a pair of fields $L/K$ that allow for a basis $\ell_1, \dots, \ell_d$ permitting a fast algorithm transforming from the polynomial representation of $\mathcal{O}_L$ to the representation $\sum_i \ell_i k_i$ with each $k_i$ in polynomial representation, which would allow one to bypass the complications of sampling Gaussian distributions by just sampling in $\mathcal{O}_L$ directly.

\subsection{Generalizing to Other Centers}
In the exposition of the previous section we required that $q$ splits completely in the center $K$. This corresponds to the requirement in the ring and module cases that $q$ splits completely in the field $K$, which allows the use of the NTT to compute multiplications over a direct product of finite fields. However, there has been recent progress in loosening this requirement for the NTT and allowing the modulus $q$ to be $1$ mod $n$ rather than $1$ mod $m$, where as usual $K$ is the $m^\text{th}$ cyclotomic field of degree $n$. For example, in the second round specification of KYBER \cite{avanzi_kyber_2019} $q$ is set as $3329$ and $n = 256$, yet they still support efficient NTT based multiplication. In such cases, $q$ is `well' split but not completely split, and the fast NTT operations use the method of \cite{lyubashevsky_truly_2019}, where $q$ splits into some product of prime ideals $\mathfrak{q}_i$ whose norms can be small powers of $q$.

We observe that our methods can be partially generalized to this case in the following manner. Say $\langle q \rangle = \prod_i \mathfrak{q}_i$ is a decomposition into prime ideals in $\mathcal{O}_K$ and there exists an efficient algorithm for fast multiplication in ${\mathcal{O}_K}_q$. We can replace our condition that $q$ splits completely in $\mathcal{O}_L$ with the condition that each ideal $\mathfrak{q}_i$ in the $\mathcal{O}_K$-factorization of $q$ splits completely into a product of $d$ prime ideals $\mathfrak{q}_i \mathcal{O}_L = \prod_{j=1}^d \mathfrak{q}_{i,j}$ in $\mathcal{O}_L$ of the same norm. Then, we can replicate the method of \cref{qfast} to find a cyclic, orthonormal basis $\textbf{e}_1,\dots, \textbf{e}_d$ of $\mathcal{O}_L/\mathfrak{q}_i \mathcal{O}_L$ over $\mathcal{O}_K/\mathfrak{q}_i$ and concatenate together the bases for each $i$ to make the cyclic, orthonormal, basis $\ell_1,\dots,\ell_d$ of ${\mathcal{O}_L}_q$ over ${\mathcal{O}_K}_q$. Since the basis is orthonormal, if $\ell = \sum_i \ell_i k_i$ and $g = \sum_i \ell_i g_i$ with each $k_i, g_i \in {\mathcal{O}_K}_q$, then
\begin{align*}
    \ell \cdot g = \sum_{i=1}^d \ell_i (g_i \cdot k_i).
\end{align*}
Since the basis is cyclic,
\begin{align*}
    \theta(\ell) &= \sum_i \theta(\ell_i) k_i \\
    &= \sum_i \ell_i k_{i-1}
\end{align*}
where we define $k_0 := k_d$.

Now we are able to use existing fast multiplication algorithms in ${\mathcal{O}_K}_q$ to compute operations in ${\mathcal{O}_L}_q$ by expressing elements in this basis. Represent each $x = \sum_{i=0}^{d-1} u^i x_i \in \Lambda_q$ by expressing each $x_i \in {\mathcal{O}_L}_q$ in the $\ell_j$ basis. Then, to multiply $x$ and $y$ in $\Lambda_q$ one only has to compute multiplications in ${\mathcal{O}_K}_q$, since the operations required are just computing the non-commutative relation $\ell u = u \theta(\ell)$, which merely permutes the $\ell_i$ using $\theta$, and computing multiplication and addition, which can be done coordinatewise in the orthonormal $\ell_i$ basis. Each $L$ multiplication requires $d$ multiplications in $K$, and each $u$ coordinate of $\Lambda$ requires $d$ multiplications in $L$. Consequentially, naive multiplication in $\Lambda_q$ takes $d^3$ instances of the efficient ${\mathcal{O}_K}_q$-multiplication algorithm we have access to. For specific $K$-multiplication algorithms it is likely that this process can be streamlined; the intention of this section is merely to demonstrate that one can build efficient $\Lambda_q$ operations from more general efficient operations over the center in the same manner that the techniques of \cref{qfast} used the CRT method.

\footnotesize
\bibliographystyle{IEEEtran}
\bibliography{CLWE_references,CLWE_references_thesis}

\end{document}